\documentclass[preprint]{elsarticle} 
\usepackage[applemac]{inputenc}
\usepackage[english]{babel}
\usepackage{pdfsync,hyperref}
\usepackage{verbatim}
\usepackage{amsmath,amsthm,amssymb,bussproofs,tikz,stmaryrd,mathtools}
\usepackage[emu]{cmll}
\usepackage{fouriernc}
\usepackage{hyperref}
\usepackage{subfigure}
\usepackage{enumitem}

\usepackage{fullmacros}

\title{Interaction Graphs: Graphings}
\author[ts]{Thomas Seiller\corref{cor1}\fnref{fn1}} \ead{thomas.seiller@ihes.fr}
\cortext[cor1]{Corresponding author}
\fntext[fn1]{This work was partially supported by the ANR-10-BLAN-0213 LOGOI.}
\address[ts]{Institut des Hautes \'{E}tudes Scientifiques (IH\'{E}S),\\ Le Bois-Marie, 35 route de Chartres, 91440 Bures-sur-Yvette, France}

\renewcommand{\lg}{\textnormal{lg}}

\newcommand{\preservesmeasurable}{measurable-preserving\xspace}
\renewcommand{\plugmes}[1][]{\plug}
\renewcommand{\cyclesmes}[2][]{\cycles[\ttrm{#1}]{#2}}

\renewcommand{\pathsmes}[1]{\enpaths{#1}}

\renewcommand{\catmll}{$\mathbb{G}\mathrm{raph}_{MLL}$\xspace} % Category for MLL
\renewcommand{\concat}{$\mathbb{C}\mathrm{ond}$\xspace} % Category Cond/Equiv
\renewcommand{\behcat}{$\mathbb{B}\mathrm{ehav}$\xspace} % Category Behav/Equiv

\usepackage{color}
\begin{document}

\begin{abstract}
In two previous papers \cite{seiller-goim, seiller-goia}, we exposed a combinatorial approach to the program of Geometry of Interaction, a program initiated by Jean-Yves Girard \cite{towards}. The strength of our approach lies in the fact that we interpret proofs by simpler structures -- graphs -- than Girard's constructions, while generalising the latter since they can be recovered as special cases of our setting. This third paper extends this approach by considering a generalisation of graphs named \emph{graphings}, which is in some way a \emph{geometric realisation} of a graph on a measure space. This very general framework leads to a number of new models of multiplicative-additive linear logic which generalise Girard's geometry of interaction models and opens several new lines of research. As an example, we exhibit a family of such models which account for second-order quantification without suffering the same limitations as Girard's models.
\end{abstract}

\maketitle

\renewcommand{\sectionautorefname}{Section}
\renewcommand{\subsectionautorefname}{Section}

%\tableofcontents

\section{Introduction}

\subsection{Context}

%\subsubsection{Geometry of Interaction}

\paragraph{Geometry of Interaction} This research program was introduced by Girard \cite{multiplicatives,towards} after his discovery of linear logic \cite{ll}. In a first approximation, it aims at defining a semantics of proofs that accounts for the dynamics of cut-elimination. Namely, the geometry of interaction models differ from usual (denotational) semantics in that the interpretation of a proof $\pi$ and its normal form $\rho$ are not equal, but one has a way of computing the interpretation of the normal form $\rho$ from the interpretation of the proof $\pi$ (illustrated in \autoref{denotgoi}). As a consequence, a geometry of interaction models not only proofs -- programs -- but also their normalization -- their execution. This semantical counterpart to the cut-elimination procedure was called the \emph{execution formula} by Girard in his first papers about geometry of interaction \cite{goi1,goi2,goi3}, and it is a way of computing the solution to the so-called \emph{feedback equation}. This equation turned out to have a more general solution \cite{feedback}, which lead Girard to the definition of a geometry of interaction in the hyperfinite factor \cite{goi5}. 

\begin{figure}
\centering
\subfigure[Denotational Semantics]{
\centering
\begin{tikzpicture}[x=1.5cm,y=1.5cm]
	\node (A) at (0,0) {$\pi$};
	\node (B) at (2,0) {$\Int{\pi}{}$};
	\node (C) at (0,-1.5) {$\rho$};
	\node (D) at (2,-1.5) {$\Int{\rho}{}$};
	
	\draw[->] (A) -- (B) node [midway,above] {$\Int{\cdot}{}$};
	\draw[->] (C) -- (D) node [midway,above] {$\Int{\cdot}{}$};`
	\draw[->] (A) -- (C) node [midway,above,sloped] {\small{cut}} node [midway,below,sloped] {\small{elimination}};
	\draw[red,double] (B) -- (D) {};

\end{tikzpicture}
}
\subfigure[Geometry of Interaction]{
\centering
\begin{tikzpicture}[x=1.5cm,y=1.5cm]
	\node (A) at (0,0) {$\pi$};
	\node (B) at (2,0) {$\Int{\pi}{}$};
	\node (C) at (0,-1.5) {$\rho$};
	\node (D) at (2,-1.5) {$\Int{\rho}{}$};
	
	\draw[->] (A) -- (B) node [midway,above] {$\Int{\cdot}{}$};
	\draw[->] (C) -- (D) node [midway,above] {$\Int{\cdot}{}$};`
	\draw[->] (A) -- (C) node [midway,above,sloped] {\small{cut}} node [midway,below,sloped] {\small{elimination}};
	\draw[->,red] (B) -- (D) node [midway,above,sloped] {$\Ex(\cdot)$};

\end{tikzpicture}
}
\caption{Denotational Semantics vs Geometry of Interaction}\label{denotgoi}
\end{figure}
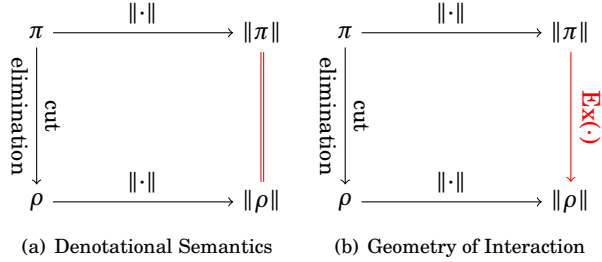

Geometry of Interaction, however, is not only about the interpretation of proofs and their dynamics, but also about reconstructing logic around this semantical counterpart to the cut-elimination procedure. This means that logic arises from the dynamics and interactions of proofs -- programs --, as a syntactical description of the possible behaviors of proofs -- programs. This aspect of the geometry of interaction program has been less studied than the proof interpretation part.

We must also point out that geometry of interaction has been successful in providing tools for the study of computational complexity. The fact that it models the execution of programs explains that it is well suited for the study of complexity classes in time \cite{baillotpedicini,lago}, as well as in space \cite{seiller-lsp,seiller-conl}. It was also used to explain \cite{AbadiGonthierLevy92b} Lamping's optimal reduction of lambda-calculus \cite{Lamping90}.

\paragraph{Interaction Graphs} They were first introduced \cite{seiller-goim} to define a combinatorial approach to Girard's geometry of interaction in the hyperfinite factor \cite{goi5}. The main idea was that the execution formula -- the counterpart of the cut-elimination procedure -- can be computed as the set of alternating paths between graphs, and that the measurement of interaction defined by Girard using the Fuglede-Kadison determinant \cite{FKdet} can be computed as a measurement of a set of cycles. 

The setting was then extended to deal with additive connectives \cite{seiller-goia}, showing by the way that the constructions were a combinatorial approach not only to Girard's hyperfinite GoI construction but also to all the earlier constructions \cite{multiplicatives,goi1,goi2,goi3}. This result could be obtained by unveiling a single geometrical property, which we called the \emph{trefoil property}, upon which all the constructions of geometry of interaction introduced by Girard are founded. This property, which can be understood as a sort of associativity, suggests that computation -- as modeled by geometry of interaction -- is closely related to algebraic topology. 

This paper takes another direction though: based on ideas that appeared in the author's phd thesis \cite{seiller-phd}, it extends the setting of graphs by considering a generalisation of graphs named \emph{graphings}, which is in some way a \emph{geometric realisation} of a graph. This very general framework leads to a number of new models of multiplicative-additive linear logic which generalise Girard's geometry of interaction models and opens several new lines of research. As an example, we exhibit a family of such models which account for second-order quantification without suffering the same limitations as Girard's models.

\subsection{Outline}

We introduce in this paper a family of models which generalises and axiomatizes the notion a GoI model. This systematic approach is obtained by extending previous work \cite{seiller-goim,seiller-goia}. While these previous constructions built models in which the objects under study were directed weighted graphs, we here consider a measure-theoretic generalisation of such graphs named \emph{graphings}. The resulting construction yields a very rich hierarchy of models parametrized by two monoids: the \emph{weight monoid} and the so-called \emph{microcosm}. 

A weight monoid is nothing more than a monoid which elements will be used to give weights to edges of the graphs. A microcosm, on the other hand, is a monoid of measurable maps, i.e.\ it is a subset of $\mathcal{M}(X)$, the set of non-singular \preservesmeasurable transformations from a measure space $(X,\mathcal{B},\mu)$ to itself, which is closed under composition and contains the identity transformation. 

Once chosen a weight monoid $\Omega$ and a microcosm $\mathfrak{m}$, we define the notion of \emph{$\Omega$-weighted graphing in $\mathfrak{m}$}. A $\Omega$-weighted graphing is simply a directed graph $F$ whose edges are weighted by elements in $\Omega$, whose vertices are measurable subsets of the measure space $(X,\mathcal{B})$, and whose edges are \emph{realised} by elements of $\mathfrak{m}$, i.e.\ for each edge $e$ there exists an element $\phi_{e}$ in $\mathfrak{m}$ such that $\phi_{e}(s(e))=t(e)$, where $s,t$ denote the source and target maps. For convenience, we use in this paper a different -- thus less \enquote{graph-like}, but equivalent definition of graphings. The reader should however understand that graphings generalise graphs in that a graphing over a discrete measure space $\measured{X}$ is nothing more than a (possibly infinite) graph whose set of vertices is a subset of $\measured{X}$.

The main result of this paper is then expressed as follows.
\begin{theorem}\label{mainthm}
Let $\Omega$ be a monoid and $\mathfrak{m}$ a microcosm. There exists a family of GoI model of multiplicative-additive linear logic (MALL) whose objects are $\Omega$-weighted graphings in $\mathfrak{m}$. 
\end{theorem}

Let us now explain how the proof of this result will be obtained. 

%\paragraph{Previous work}
Previous work \cite{seiller-goia} studied in details the construction of models of \MALL where proofs are interpreted as weighted directed graphs. This construction, however, can be performed as long as one has the right mathematical structure, and did not dwell on the fact that one considered directed weighted graphs.
%\begin{definition}
%An \emph{IG-algebra} is a set $\mathcal{U}$ endowed with an associative binary operation $\plug: \mathcal{U}\times\mathcal{U}\rightarrow\mathcal{U}$ and a measurement $\meas{\cdot,\cdot}:\mathcal{U}\times\mathcal{U}\rightarrow\realposN\cup\{\infty\}$ which satisfy the \emph{trefoil property}, i.e.\ for all $F,G,H\in\mathcal{U}$:
%$$ \meas{F\plug G,H}+\meas{F,G}=\meas{G\plug H,F}+\meas{G,H}=\meas{H\plug F,G}+\meas{H,F} $$
%We write the IG-algebra $\mathbb{A}=(\mathcal{U},\plug,\meas{\cdot,\cdot})$.
%\end{definition}
More precisely, the cited previous work \cite{seiller-goia} describes the model structure in an axiomatic way, based on a small number of operations and properties. One mainly needs two operations -- execution and measurement, and two properties -- associativity of the execution and the trefoil property. Those are recalled in \autoref{sec_recall}, together with a sketch of the model construction.
%\begin{theorem}[Seiller \cite{seiller-goia}]
%Let $\mathbb{A}=(\mathcal{U},\plug,\meas{\cdot,\cdot})$ be an IG-algebra. There exists a non-degenerate\footnote{A non-degenerate model should be understood here as a model in which $\otimes\neq\parr$, $\oplus\neq\with$, $1\neq \top$, $0\neq \bot$. Moreover, these models do not satisfy weakening and mix rules.} model a \MALL built from $\mathbb{A}$.
%\end{theorem}

Now, the proof of our main theorem will built on this axiomatic construction. The main part of the proof will therefore consist in showing that the set of graphings can be endowed with the right structure. The proof of this will be decomposed in three steps:
\begin{itemize}[noitemsep,nolistsep]
\item define the \emph{execution} $F\plugmes G$ between $F$ and $G$, and show that it is associative;
\item define the measurement $\meas{F,G}$ of the interaction between $F$ and $G$;
\item show that these two notions satisfy the so-called \emph{trefoil property}.
\end{itemize}

The structure of this paper follows this three steps decomposition. \autoref{sec_execution} is therefore concerned with the definition of the execution and the proof that it is indeed associative (\autoref{thm_associativity}). We then define in \autoref{sec_measurement} abstractly a notion of measurement and prove, under certain conditions, that the trefoil property holds (\autoref{thm_trefoilppty}). These conditions being quite involved, we then show in \autoref{sec_trefoil} that in most reasonable cases, one can define whole families of measurement satisfying them (\autoref{thm_trefoilspace}).

%\begin{proposition}
%Let $(X,\mathcal{B},\mu)$ be a measure space, where $\mathcal{B}$ is the Borel $\sigma$-algebra of the second-countable (Hausdorff) topological space $(X,\mathcal{T})$ and $\mu$ is a $\sigma$-finite Radon measure. For any microcosm $\mathfrak{m}$ on $X$, and for any measurable map $m:\Omega\rightarrow \realposN\cup\{\infty\}$, one can define a measurement $\meas{F,G}$ between $\Omega$-weighted graphings in $\mathfrak{m}$ which satisfies the trefoil property.
%\end{proposition}

We then give some examples of applications of this result, showing how all of Girard's frameworks, either based on operator algebras \cite{goi1,goi2,goi5} or on the more syntactical \enquote{unification algebra} \cite{goi3,goi6,goi6light}, can be understood as special cases of our construction.

\autoref{sec_mall2}, finally, studies a simple example of GoI models obtained from our construction. We show how this framework allows to interpret multiplicative-linear logic with second-order quantification. This result solves an important issue, due to locativity, that arose in Girard's hyperfinite geometry of interaction \cite{goi5}, and that we discuss now.

In Girard's hyperfinite geometry of interaction, proofs (resp. formulas) are interpreted on a given location, i.e.\ as operators (resp. sets of operators) acting on specific subspace explicited by the choice of a projection $p$ -- the carrier. Quantification is thus defined only over all formulas with a given carrier $p$; to interpret correctly a cut between quantified formulas with respective carriers $p,q$ it is then necessary to have a partial isometry between $p$ and $q$. This is where Girard's setting is confronted with a problem: such a partial isometry does not always exists in a type {II} (von Neumann) algebra. Thus the model obtained in type {II} algebras does not interpret second-order quantification correctly. To solve this, one could think about extending the algebra by adding all the missing partial isometries; the resulting von Neumann algebra, however, is no longer a type {II} algebra and therefore no determinant can be defined on it. But the Fuglede-Kadison determinant (i.e.\ the detemrinant in type {II} factors) was essential to Girard's construction. As a consequence, there seems to be no solution to this problem. We will explain in the last section, however, how our models based on graphings actually solve this issue. The main intuition on why we are able to solve it is that the measurement defined in \autoref{sec_trefoil} provides, in the adequate setting, a generalisation of the Fuglede-Kadison determinant to non-type-{II} algebras.

\subsection{Motivations and Perspectives}

This result is very technical, and we believe that the work needed to obtain it should be motivated. The importance of this result lies in its great generality and the new lines of research it opens. We believe that the notion of graphing is an excellent mathematical abstraction of the notion of program. First, any pure lambda-term can be represented as a graphing since Girard's GoI model \cite{goi2} is a particular case of our constructions. But one can represent a lot more. In following work, we will show how to represent quantum computation with the same objects. Results of Mazza \cite{mazza} and de Falco \cite{deFalco} leads us to believe that concurrent models of computations, such as pi-calculus can be represented as graphings. In the long-term, we would also like to obtain such a representation for other, more mathematical, computational paradigms, such as cellular automata.

We therefore consider that this framework offers a perfect mathematical abstraction of the notion of program, one that is machine-independent, which captures many different computational paradigms, and which allows for a fine control on computational principles. Indeed, while linear logic introduced a syntax in which one could talk about resource-usage, the approach taken here goes a step further and introduces even more subtle distinctions on the computational principles allowed in the model. These distinctions are made by considering the hierarchies of weight monoid and microcosms. This can lead to some interesting results in computational complexity on one hand, and to interesting models of quantum computation on the other. We now detail these two motivating perspectives.

\subsubsection{Fine-Grained Implicit Computational Complexity}

Firstly, we believe that the results obtained in previous work with C. Aubert characterizing the classes \textbf{L} and \textbf{coNL} can be improved, extended and linked with logical constructions. The following summarizes an approach to computational complexity based on the results of this paper. The interested reader can find more details in a recent perspective paper \cite{seiller-towards}.

In the previously described models of multiplicative-additive linear logic, one can define the type of binary lists $\cond{Nat}_{2}$ in a quite natural fashion (the representation of lists is thoroughly explained in previous work on complexity \cite{seiller-phd,seiller-conl,seiller-lsp}). Moreover, in a number of cases one will be able to define exponential connectives in the model and therefore consider the type $\cond{\oc Nat_{2}\multimap Bool}$. Elements of this type, when applied to a element $\oc N_{n}$ of $\oc\cond{Nat_{2}}$, yields one of the distinguished elements $\tt true\rm$ or $\tt false\rm$. Such an element $F$ thus \emph{computes} the language $\mathcal{L}(F)=\{n\in \naturalN~|~F\plug \oc N_{n}=\tt true\rm\}$. As a consequence, we can study the set of languages computed by all elements in the type $\oc\cond{Nat_{2}\multimap Bool}$. By modifying the microcosm, one then modifies the expressivity of the model.

The intuition is that a microcosm $\mathfrak{m}$ represents the set of computational principles available to write programs in the model. Considering a bigger microcosm $\mathfrak{n}\supsetneq\mathfrak{m}$ thus corresponds to extending the set of principles at disposal, consequently increasing expressivity. The set of languages characterized by the type $\oc\cond{Nat}_{2}\multimap\cond{Bool}$ becomes larger and larger as we consider extensions of the microcosms. %As an example, the microcosm $\mathfrak{t}$ corresponds to allowing oneself to compute with automata. Expanding this microcosm by $\mathfrak{p}$ corresponds to the addition of multiple heads to the automata. 
We can then work on this remark, and use intuitions gained from earlier work \cite{seiller-conl,seiller-lsp}. This leads to a perfect correspondence between a hierarchy of monoids on the measure space $\integerN\times[0,1]^{\naturalN}$ and a hierarchy of classes of languages in between regular languages and logarithmic space predicates (both included) \cite{seiller-goic}.

\subsubsection{Quantum Computing and Unitary Bases}

The second interest in the hierarchy of models thus obtained concerns quantum computation. Indeed, quantum computation can be represented in a very natural way as graphings. Using the weight monoid to include complex coefficients, one can represent qbits in a natural way as graphings corresponding to their density matrix. One can then represent unitary operators in a similar way and show that the execution as graphings corresponds to the conjugation of the density matrix by the unitary operator. Using these techniques, and showing that one can construct semantically a tensor product of qbits that allows for entanglement, we can obtain a representation of quantum computation that allows for entanglement of functions. 

The construction just described opens a particularly interesting line of research if one considers the fact that the obtained model lives in the hierarchy of models described by the microcosms. Indeed, contrarily to theoretical quantum computation which allows the use of unitary gates for any unitary operator, a real, physical, quantum computer would allow only for a finite number of already implemented unitary gates. The set of such gates is called a \emph{basis of unitary operators} and satisfied the property that the span, under composition, of this set of operators is dense in the set of all unitaries. In other words, any unitary operator can be approximated by composites of elements of the basis.

Such a restriction is necessary, but there are a number of different choices for the basis. To compare two different bases, one can ask a physicist which unitary gates would be easier to create. This is how bases are compared nowadays: the one which is easiest to construct is considered as better. However, the choice of the basis may have important consequences from a computational and/or logical point of view. Our construction provides an adequate framework to tackle this issue. Indeed, one can restrict the microcosm to allow only unitary gates in a given basis and study the obtained model of computation.

\section{Interaction Graphs: Execution and the Trefoil Property}\label{sec_recall}

We defined in earlier work \cite{seiller-goim,seiller-goia} a graph-theoretical construction where proofs -- or more precisely paraproofs, that is generalised proofs -- are interpreted by finite objects\footnote{Even though the graphs we consider can have an infinite set of edges, linear logic proofs are represented by finite graphs (disjoint unions of transpositions).}. The graphs we considered were directed and weighted, with the weights taken in a monoid $(\Omega,\cdot)$. We briefly expose the main results obtained in these papers, since the constructions in the next section are a far-reaching generalisations of those.

\begin{definition}
A \emph{directed weighted graph} is a tuple $G$, where $V^{G}$ is the set of vertices, $E^{G}$ is the set of edges, $s^{G}$ and $t^{G}$ are two functions from $E^{G}$ to $V^{G}$, the \emph{source} and \emph{target} functions, and $\omega^{G}$ is a function $E^{G} \rightarrow \Omega$.
\end{definition}

The construction is centered around the notion of alternating paths. Given two graphs $F$ and $G$, an alternating path is a path $e_{1}\dots e_{n}$ such that $e_{i}\in E^{F}$ if and only if $e_{i+1}\in E^{G}$. The set of alternating paths will be used to define the interpretation of cut-elimination in the framework, i.e.\ the graph $F\plug G$ -- the \emph{execution of $F$ and $G$} -- is defined as the graph of alternating paths between $F$ and $G$ whose source and target are in the symmetric difference $V^{F}\Delta V^{G}$. The weight of a path is naturally defined as the product of the weights of the edges it contains. One easily verifies that this operation is associative: as long as the three graphs $F,G,H$ satisfy $V^{F}\cap V^{G}\cap V^{H}=\emptyset$, we have:
$$(F\plug G)\plug H=F\plug(G\plug H)$$

As it is usual in mathematics, this notion of paths cannot be considered without the associated notion of cycle: an \emph{alternating cycle} between two graphs $F$ and $G$ is a cycle which is an alternating path $e_{1}e_{2}\dots e_{n}$ such that $e_{1}\in V^{F}$ if and only if $e_{n}\in V^{G}$. For technical reasons, we actually consider the related notion of $1$-circuit, which is a cycle satisfying some technical property.

\begin{definition}\label{defcircuits}
We define the following notions of cycles:
\begin{itemize}[noitemsep,nolistsep]
\item a \emph{cycle} in a graph $F$ is a sequence $\pi=e_{0}\dots e_{n}$ of edges such that for all $i< n$ the source of the edge $e_{i+1}$ coincides with the target of the edge $e_{i}$;
\item a \emph{$1$-cycle} in a graph $F$ is a cycle $\pi$ such that there are no cycle $\rho$ and integer $k>1$ with $\pi=\rho^{k}$, where $\rho^{k}$ denotes the concatenation of $k$ copies of $\rho$;
\item a \emph{circuit} is an equivalence class of cycles for the equivalence relation defined by $e_{0}\dots e_{n}\sim f_{0}\dots f_{n}$ if and only if there exists an integer $k$ such that for all $i$, $e_{i}=f_{j}$ with $j=k+i [n+1]$.
\item a \emph{$1$-circuit} $\rho$ is a circuit which is not a proper power of a smaller circuit, i.e.\ is the equivalence class of a $1$-cycle.
\end{itemize}
\end{definition}

We will denote by $\uncircuits{F,G}$ the set of $1$-circuits. It can be shown that these notions of paths and cycles satisfy a property we call the \emph{trefoil property} which turns out to be fundamental for constructing models of linear logic. This property states the existence of weight-preserving bijections between sets of $1$-circuits:
\begin{equation}
\uncircuits{F\plug G,H}\cup\uncircuits{F,G}\cong\uncircuits{G\plug H,F}\cup\uncircuits{G,H}\cong\uncircuits{H\plug F,G}\cup\uncircuits{H,F}\label{geometrictrefoil}
\end{equation}

In this setting, one can define the multiplicative and additive connectives of Linear Logic. This construction is parametrized by a map from the set $\Omega$ to $\mathbb{R}_{\geqslant 0}\cup\{\infty\}$. We thus obtain not only one but a whole family of models. This parameter is introduced to define the notion of orthogonality in our setting and is used to measure the sets of $1$-circuits. Indeed, given a map $m$ and two graphs $F,G$ we define $\meas{F,G}$ as the sum $\sum_{\pi\in\uncircuits{F,G}} m(\omega(\pi))$, where $\omega(\pi)$ is the weight of the cycle $\pi$.

Let us sketch the construction now. For pedagogical purposes, we will only sketch the construction of multiplicative connectives. The reader interested in the additive construction can deduce it from the construction sketched in \autoref{sec_basicdefs}. For a detailed study, we refer to the author's paper on additives \cite{seiller-goia}.

For technical reasons explained in earlier work \cite{seiller-goim}, one has to work with couples $(a,A)$ -- named \emph{projects} -- of a real number $a$ together with a graph $A$. From the measurement just mentionned, one can define an orthogonality relation between projects $(a,A)$ and $(b,B)$ such that $A,B$ have the same set of vertices as follows:
$$(a,A)\poll{}(b,B)\text{ if and only if }a+b+\meas{A,B}\neq 0,\infty$$
Now, the notion of \emph{project} is meant to interpret proofs, and \emph{orthogonality} somehow interprets negation \cite{seiller-axioms}. From the orthogonality relation, one can define a notion of type which is quite natural and reminiscent of work on classical realisability \cite{krivine1,riba,krivine2}. These objects generalising the notion of type are called \emph{conducts}. A conduct is simply a set of projects which is equal to its bi-orthogonal closure, i.e.\ $\cond{A}=\cond{A}^{\pol\pol}$. Notice that all projects in a given conduct have the same set of vertices, and we can therefore talk about the set of vertices of a conduct.

The connectives are then defined first between projects -- i.e.\ (generalised) proofs. This low-level definition is then lifted to the conducts naturally. The multiplicative conjunction of linear logic, $\otimes$, is for instance defined as follows: if $(a,A)$ and $(b,B)$ are projects whose graphs have disjoint sets of vertices, their tensor $(a,A)\otimes(b,B)$ is defined as $(a+b,A\cup B)$. Now, given two conducts $\cond{A,B}$ defined on disjoint sets of vertices, the tensor on projects is lifted as follows:
$$ \cond{A\otimes B}=\{(a,A)\otimes(b,B)~|~(a,A)\in\cond{A},(b,B)\in\cond{B}\}^{\pol\pol} $$

From any value of $m$, one obtains in this way \cite{seiller-goim,seiller-goia} a $\ast$-autonomous category \catmll{} with $\parr\not\cong\otimes$ and $1\not\cong\bot$, i.e.\ a non-degenerate denotational semantics for Multiplicative Linear Logic (MLL). A consequence of the trefoil property is that this category can be quotiented by an observational equivalence while conserving its structure of $\ast$-autonomous category. 

When considering the full construction, i.e.\ with additive connectives, the categorical model obtained in this way has two layers (see \autoref{catmodels}). The first layer consists in this non-degenerate (i.e.\ $\otimes\neq\parr$ and $\cond{1}\neq\cond{\bot}$) $\ast$-autonomous category \concat obtained as a quotient, hence a denotational model for MLL with units. The second layer is a full subcategory \behcat which does not contain the multiplicative units but is a non-degenerate model (i.e.\ $\otimes\neq\parr$, $\oplus\neq\with$ and $\cond{0}\neq\cond{\top}$) of MALL with additive units that does not satisfy the mix and weakening rules.

\begin{figure}
\centering
\begin{tikzpicture}[x=1.2cm,y=0.8cm]
	\draw[fill,opacity=0.1] (0,0) .. controls  (0,4.5) and (0.5,5) .. (5,5) .. controls (9.5,5) and (10,4.5) .. (10,0) .. controls (10,-4.5) and (9.5,-5) .. (5,-5) .. controls (0.5,-5) and (0,-4.5) .. (0,0) ;
		\node (A) at (2,0) {\begin{tabular}{c}\small{\concat}\\\small{($\ast$-autonomous)}\end{tabular}};
	\draw[fill,opacity=0.2] (4,0) .. controls  (4,2.5) and (4.5,3) .. (6,3) .. controls (7.5,3) and (8,2.5) .. (8,0) .. controls (8,-2.5) and (7.5,-3) .. (6,-3) .. controls (4.5,-3) and (4,-2.5) .. (4,0) ;
		\node (B) at (6,0) {\begin{tabular}{c}\small{\behcat}\\\small{(closed under $\otimes,\multimap,\with,\oplus,(\cdot)^{\pol}$)}\\\small{NO weakening, NO mix}\end{tabular}};
	\node (bot) at (2,-3) {$\bullet_{\bot}$};
	\node (one) at (3,-3) {$\bullet_{\cond{1}}$};
	
	\node (top) at (6,-2) {$\bullet_{\cond{T}}$};
	\node (zero) at (7,-2) {$\bullet_{\cond{0}}$};
\end{tikzpicture}
\caption{Structure of the categorical models}\label{catmodels}
\end{figure}
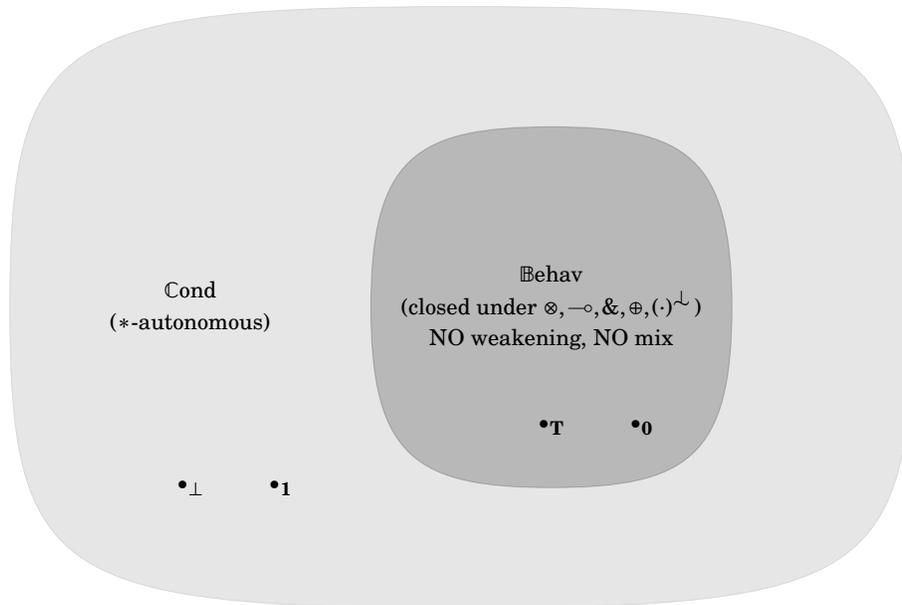

\section{Graphings, Paths and Execution}\label{sec_execution}

%\referee{1}{Forget about almost-everywhere equality at first and reintroduce it later.}

We define in this section the notion of graphing. This notion was first considered by Adams \cite{adams}, and later by Levitt and Gaboriau \cite{levitt_graphings,gaboriaucost} in order to study measurable group actions. It generalises in the setting of measure theory the topological notion of pseudo-group \cite{pseudogroup_book} which was introduced by E. Cartan \cite{cartan1,cartan2}.

We will here consider graphings as a generalisation of graph or, more acutely, as a \emph{realiser} of a graph. We will therefore define the notion of path and built upon it a generalisation of the notion of execution mentioned in the previous section. To stay consistent with the spirit of measure-theory, we will first examine the notion of almost-everywhere equality between graphings.

\begin{remark}
The consideration of almost-everywhere equality is not essential for the first parts of the paper. Indeed the trefoil property is, as it was in the setting of graphs, still a purely combinatorial property. One could therefore work without quotienting up to almost-everywhere equality up to the point of providing explicit measurements, i.e.\ up to \autoref{sec_trefoil}. However, our measurements are built from the underlying measure space structure, and then almost-everywhere equality becomes natural.  

This explains why we whose to work from the start with a measure-theoretic approach. A more combinatorial approach may be considered, but we believe that this would make more sense in a topological framework such as the one sketched in \autoref{topographings}.
\end{remark}

%\textcolor{red}{In this section, we define .... and notion of paths ... and the definition of execution .. carvings ... and show associativity ...}

\subsection{First Definitions}

The idea is that a graphing is a sort of \enquote{geometric realisation} of a graph: the vertices correspond to measurable subsets of a measure space, and edges correspond to measurable maps\footnote{To be exact, we will consider graphings whose edges are taken in a \emph{microcosm}, that is a subset of all measurable maps which is closed under composition and contains the identity.} from the source subset onto the target subset. Some difficulties arise when one wants to define a tractable notion of graphing. Indeed, a new phenomenon appears when vertices are measurable sets: what should one do when two vertices are neither disjoint or equal, i.e.\ when two vertices are not equal but their intersection is not of null measure? One solution would be to define graphings where vertices are disjoint subsets (i.e.\ their intersection is of null measure), but this makes the definition of execution extremely complex.

Let us consider for instance two graphings with a single edge each, and whose plugging is represented in \autoref{graphagesplugex1}. As suggested by the representation, the measurables sets are subject to the following inclusions: $V_{t}\subset U_{s}$ and $V_{s}\subset U_{t}$. To represent the set of alternating paths whose source and target are subsets of the symmetric difference of the carriers -- the execution of the two graphs -- we would need to decompose each of the measurable sets into a disjoint union of sets, each one corresponding to the source and/or target of a path. In the particular case we show in the figure, this operation is a bit complicated but still tractable: it is sufficient to consider the sets $(\psi\phi)^{-k}(\phi^{-1}(U_{t}-V_{s}))\cap(U_{s}-V_{t})$ as the set of sources. Indeed, those are the subsets of $U_{s}-V_{t}$ containing all points such that $\phi(\psi\phi)^{k}(x)$ is defined and is an element of $U_{t}-V_{s}$, i.e.\ the sets of points that are in the domain of the \enquote{length $k$} alternating path between $\phi$ and $\psi$. However, the operation quickly becomes much more complicated as we add new edges and create cycles. \autoref{graphagesplugex2} represents the case of two graphings with two edges each. Defining the decomposition of the set of vertices induced by the execution is -- already in this case -- very difficult. In particular, since the sets of vertices considered can be infinite (but countable), the number of cycles can be infinite, and the operation is then of an extreme complexity.

As a consequence, we have chosen to work with a different presentation of graphings, where two distinct vertices can have a intersection of strictly positive measure -- they can even be equal. We will now define the notion of graphing taking into account these remarks. The terminology is borrowed from works of Levitt \cite{levitt_graphings} and Gaboriau \cite{gaboriaucost}, in which the underlying notion of graphing (forgetting about the weights) is defined.

\begin{figure}
\centering
\begin{tikzpicture}[x=0.9cm,y=0.9cm]
	\draw[|-|,red] (0,0) -- (3,0) node [very near start, above] {$U_{s}$};
	\draw[|-|,red] (5,0) -- (8,0) node [very near end, above] {$U_{t}$};
	
	\draw[|-|,blue] (1,-0.25) -- (2,-0.25) node [near start, below] {$V_{t}$};
	\draw[|-|,blue] (6,-0.25) -- (7,-0.25) node [near end, below] {$V_{s}$};
	
	\draw[->,red] (1.5,0.25) .. controls (1.5,1.5) and (6.5,1.5) .. (6.5,0.25) node [midway,above] {$\phi$};
	\draw[<-,blue] (1.5,-0.5) .. controls (1.5,-1.5) and (6.5,-1.5) .. (6.5,-0.5) node [midway,below] {$\psi$};
\end{tikzpicture}
\caption{Example of a plugging between graphings}\label{graphagesplugex1}
\end{figure}
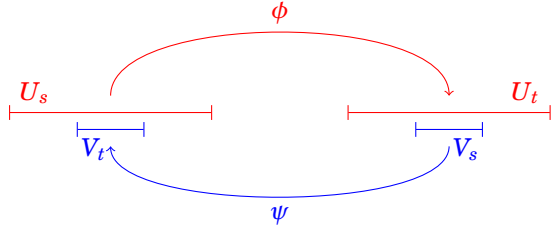

\begin{figure}
\centering
\begin{tikzpicture}[x=0.9cm,y=0.9cm]
	\draw[|-|,red] (0,0) -- (4,0) {};
	\draw[|-|,red] (5,0) -- (9,0) {};
	\draw[|-|,red] (10,0) -- (11,0) {};
	\draw[|-|,red] (7.5,0) -- (8.5,0) {};
	
	\draw[|-|,blue] (1,-0.25) -- (2,-0.25) {};
	\draw[|-|,blue] (6,-0.25) -- (8,-0.25) {};
	\draw[|-|,blue] (8,-0.25) -- (9,-0.25) {};
	\draw[|-|,blue] (10,-0.25) -- (12,-0.25) {};
	
	\draw[->,red] (2,0.25) .. controls (2,1.5) and (7,1.5) .. (7,0.25) {};
	\draw[->,red] (10.5,0.25) .. controls (10.5, 1.5) and (8,1.5) .. (8,0.25) {};
	\draw[<-,blue] (1.5,-0.5) .. controls (1.5,-1.5) and (8.5,-1.5) .. (8.5,-0.5) {};
	\draw[->,blue] (7,-0.5) .. controls (7,-1.5) and (11,-1.5) .. (11,-0.5) {};
\end{tikzpicture}
\caption{Example of a plugging between graphings}\label{graphagesplugex2}
\end{figure}
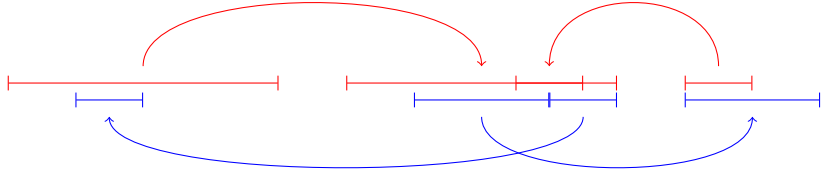

%\textcolor{red}{Change the notion of microcosm by considering \preservesmeasurable maps. Explain that in practice we use measure-preserving or measure-inflating maps and that these are bimeasurable bijections hence borel preserving. Change also the definition of flows to make flows borel preserving (not inclusion of variables but equality... this is not a problem since this is Girard's def.) + explain that one could very well consider a weaker notion of graphing, which would more or less be a notion of "pseudo-monoid" to built on the notion of pseudo-group. Then this would generalise the unit. alg. when considering inclusion of sets of variables (maybe explain why this is WAY less interesting).}

\begin{definition}
Let $(X,\mathcal{B},\lambda)$ be a measure space. We denote by $\mathcal{M}(X)$ the set of non-singular \preservesmeasurable transformations\footnote{A non-singular transformation $f:X\rightarrow X$ is a measurable map which preserves the sets of null measure, i.e.\ $\lambda(f(A))=0$ if and only if $\lambda(A)=0$. A map $f:X\rightarrow X$ is \preservesmeasurable if it maps every measurable set to a measurable set.} $X\rightarrow X$. A \emph{microcosm}  of the measure space $X$ is a subset $\mathfrak{m}$ of $\mathcal{M}(X)$ which is closed under composition and contains the identity.
\end{definition}

In the following, we will consider a notion of graphing depending on a \emph{weight-monoid} $\Omega$, i.e.\ a monoid $(\Omega,\cdot,1)$ which contains the possible weights of the edges. %In examples, we will use the set $\Omega=]0,1]$ endowed with the usual multiplication.

\begin{definition}[Graphings]
Let $\mathfrak{m}$ be a microcosm of a measure space $(X,\mathcal{B},\lambda)$ and $V^{F}$ a measurable subset of $X$. A \emph{$\Omega$-weighted graphing in $\mathfrak{m}$} with carrier $V^{F}$ is a countable family $F=\{(\omega_{e}^{F},\phi_{e}^{F}: S_{e}^{F}\rightarrow T_{e}^{F})\}_{e\in E^{F}}$, where, for all $e\in E^{F}$ (the set of \emph{edges}):
\begin{itemize}[noitemsep,nolistsep]
\item $\omega_{e}^{F}$ is an element of $\Omega$, the \emph{weight} of the edge $e$;
\item $S_{e}^{F}\subset V^{F}$ is a measurable set, the \emph{source} of the edge $e$;
\item $\phi_{e}^{F}$ is the restriction of an element of $\mathfrak{m}$ to $S_{e}^{F}$, the \emph{realiser} of the edge $e$;
\item $T_{e}^{F}=\phi_{e}^{F}(S_{e}^{F})\subset V^{F}$ is a measurable set, the \emph{target} of the edge $e$.
\end{itemize}
Given a graphing $F$, we define its \emph{effective carrier} as the measurable set $\bigcup_{e\in E^{F}} S_{e}^{F}$.
\end{definition}

%\begin{remark}
%In particular, one can notice that if $F$ is a weighted graphing, then for all $e\in E^{F}$, $\lambda(S_{e}^{F})=\lambda(T_{e}^{F})$.
%\end{remark}
%
%
%
%

\subsection{Almost-Everywhere Equality}

For the remaining of this section, we consider that we fixed once and for all the weight monoid $\Omega$ and the microcosm $\mathfrak{m}$. We will therefore refer to $\Omega$-weighted graphings in $\mathfrak{m}$ simply as \emph{graphings}.

It is usual, when doing measure theory, to work modulo sets of null measure. Similarly, we will work with graphings modulo almost everywhere equality, a notion that we need to define first. Before giving the definition, we will define the useful notion of empty graphing. An empty graphing will be almost everywhere equal to the graphing without edges.

\begin{definition}[Empty graphings]\label{emptygraphing}
A graphing $F$ is said to be \emph{empty} if its effective carrier is of null measure.
\end{definition}

\begin{definition}[Almost Everywhere Equality]\label{equivalencepp}
Two graphings $F,G$ are \emph{almost everywhere equal} if there exists two empty graphings $0_{F},0_{G}$ and a bijection $\theta: E^{F}\disjun E^{0_{F}}\rightarrow E^{G}\disjun E^{0_{G}}$ such that:
\begin{itemize}[noitemsep,nolistsep]
\item for all $e\in E^{F}\disjun E^{0_{F}}$, $\omega_{e}^{F\cup 0_{F}}= \omega_{\theta(e)}^{G\cup 0_{G}}$;
\item for all $e\in E^{F}\disjun E^{0_{F}}$, $S_{e}^{F\cup 0_{F}}\Delta S_{\theta(e)}^{G\cup 0_{G}}$ is of null measure;
\item for all $e\in E^{F}\disjun E^{0_{F}}$, $T_{e}^{F\cup 0_{F}}\Delta T_{\theta(e)}^{G\cup 0_{G}}$ is of null measure;
\item for all $e\in E^{F}\disjun E^{0_{F}}$, $\phi_{\theta(e)}^{G\disjun 0_{G}}$ and $\phi_{e}^{F\disjun 0_{F}}$ are equal almost everywhere on $S_{\theta(e)}^{G\disjun 0_{G}}\cap S_{e}^{F\disjun 0_{F}}$;
\end{itemize}
\end{definition}

\begin{proposition}
We define the relation $\sim_{a.e.}$ between graphings:
\begin{equation*}
\text{$F\sim_{a.e.} G$ if and only if $F$ and $G$ are almost everywhere equal}
\end{equation*}
This relation is an equivalence relation.
\end{proposition}

\begin{proof}
It is obvious that this relation is reflexive and symmetric (it suffices to take the bijection $\theta^{-1}$). We therefore only need to show transitivity. Let $F,G,H$ be three graphings such that $F\sim_{a.e.} G$ and $G\sim_{a.e.} H$. Therefore there exists four empty graphings $0_{F},0_{G^{F}},0_{G^{H}},0_{H}$ and two bijections $\theta_{F,G}: E^{F}\disjun E^{0_{F}}\rightarrow E^{G}\disjun E^{0_{G^{F}}}$ and $\theta_{G,H}: E^{G}\disjun E^{0_{G^{H}}}\rightarrow E^{H}\disjun E^{0_{H}}$ that satisfy the properties listed in the preceding definition. We notice that $0_{F}\disjun 0_{G^{H}}$ and $0_{G^{F}}\disjun 0_{H}$ are empty graphings. One can then define $\theta_{F,H}=(\theta_{G,H}\disjun \text{Id}_{E^{0_{G^{F}}}})\circ (\text{Id}_{E^{G}}\disjun\tau)\circ(\theta_{F,G}\disjun \text{Id}_{E^{0_{G^{H}}}})$, where $\tau$ represents the symmetry $E^{0_{G^{F}}}\disjun E^{0_{G^{H}}}\rightarrow E^{0_{G^{H}}}\disjun E^{0_{G^{F}}}$;
\begin{center}
\begin{tikzpicture}
	\node (dom) at (0,0) {$E^{F}\disjun E^{0_{F}}\disjun E^{0_{G^{H}}}$};
	\node (im1) at (6,0) {$E^{G}\disjun E^{0_{G^{F}}}\disjun E^{0_{G^{H}}}$};
	\node (im11) at (6,-3) {$E^{G}\disjun E^{0_{G^{H}}}\disjun E^{0_{G^{F}}}$};
	\node (im2) at (0,-3) {$E^{H}\disjun E^{0_{H}}\disjun E^{0_{G^{F}}}$};
	
	\draw[->] (dom) -- (im1) node [midway,above] {$\theta_{F,G}\disjun \text{Id}_{E^{0_{G^{H}}}}$};
	\draw[->] (im1) -- (im11) node [midway,right] {$\text{Id}_{E^{G}}\disjun \tau$};
	\draw[->] (im11) -- (im2) node [midway,above] {$\theta_{G,H}\disjun \text{Id}_{E^{0_{G^{F}}}}$};
	
	\draw[->,dotted] (dom) -- (im2) node [midway,left] {$\theta_{F,H}$};
\end{tikzpicture}
\end{center}
It is then easy to verify that the three first properties of almost everywhere equality are satisfied. We will only detail the proof that the fourth property also holds. We will forget about the superscripts in order to simplify notations. We will moreover denote by $\tilde{\theta}_{F,G}$ (resp. $\tilde{\tau}$, resp. $\tilde{\theta}_{G,H}$) the function $\theta_{F,G}\disjun\text{Id}_{E^{0_{G^{H}}}}$ (resp. $\text{Id}_{E^{G}}\disjun\tau$, resp. $\theta_{G,H}\disjun\text{Id}_{E^{0_{G^{F}}}}$). 

Chose $e\in E^{F}\disjun E^{0_{F}}\disjun E^{0_{G^{H}}}$:
\begin{itemize}[noitemsep,nolistsep]
\item if $e\in E^{0_{G^{H}}}$, then $\tilde{\theta}_{F,G}(e)=e$, and $\phi_{\tilde{\theta}(e)}=\phi_{e}$;
\item if $e\in E^{F}\disjun E^{0_{F}}$ then, by the definition of $\theta_{F,G}$, $\phi_{\tilde{\theta(e)}}$ is almost everywhere equal to $\phi_{e}$ on $S_{e}\cap S_{\tilde{\theta(e)}}$.
\end{itemize}
Thus $\phi_{\tilde{\theta}(e)}$ and $\phi_{e}$ are equal almost everywhere on $S_{e}\cap S_{\tilde{\theta}(e)}$ in all cases. A similar reasoning shows that for all $f\in E^{G}\disjun E^{0_{H}}\disjun E^{0_{G^{F}}}$, the functions $\phi_{\theta_{G,H}(f)}$ and $\phi_{f}$ are almost everywhere equal on $S_{\theta_{G,H}(f)}\cap S_{f}$.

Moreover, $\phi_{\tilde{\theta}_{F,G}(e)}$ and $\phi_{\tilde{\tau}(\tilde{\theta}_{F,G}(e))}$ are equal and have the same domain $S_{\tilde{\theta}_{F,G}(e)}=S_{\tilde{\tau}(\tilde{\theta}_{F,G}(e))}$. Thus $\phi_{\tilde{\tau}(\tilde{\theta}_{F,G}(e))}$ and $\phi_{e}$ are almost everywhere equal on the intersection $S_{\tilde{\tau}(\tilde{\theta}(e))}\cap S_{e}$. Moreover,  $\phi_{\tilde{\tau}(\tilde{\theta}_{F,G}(e))}$ and $\phi_{\tilde{\theta}_{G,H}(\tilde{\tau}(\tilde{\theta}_{F,G}(e)))}$ are almost everywhere equal on the intersection $S_{\tilde{\tau}(\tilde{\theta}_{F,G}(e))}\cap S_{\tilde{\theta}_{G,H}(\tilde{\tau}(\tilde{\theta}_{F,G}(e)))}$. We deduce from this that the functions $\phi_{e}$ and $\phi_{\theta_{F,H}(e)}$ are almost everywhere equal on 
$$S_{e}\cap S_{\theta_{F,H}(e)}\cap S_{\tilde{\tau}(\tilde{\theta}_{F,G}(e))}=S_{e}\cap S_{\theta_{F,H}(e)}\cap S_{\tilde{\theta}_{F,G}(e)}$$
We denote by $Z$ the set of null measure on which they differ. Since $S_{e}\Delta S_{\tilde{\theta}_{F,G}(e)}$ is of null measure, there exists two sets $X,Y$ of null measure such that $S_{e}\cup X=S_{\tilde{\theta}_{F,G}(e)}\cup Y$. We can deduce\footnote{One can chose $Y$ in such a way so that $S_{\tilde{\theta}_{F,G}(e)}\cap Y=\emptyset$.} that $S_{\tilde{\theta}_{F,G}(e)}=S_{e}\cup X -Y$. Thus 
\begin{eqnarray*}
\lefteqn{S_{e}\cap S_{\theta_{F,H}(e)}\cap S_{\tilde{\theta}_{F,G}(e)}}\\
&=&S_{e}\cap S_{\tilde{\theta}_{F,H}(e)}\cap (S_{e}\cup X-Y)\\
&=&S_{e}\cap S_{\tilde{\theta}_{F,H}(e)}\cap S_{e}-Y\\
&=&S_{e}\cap S_{\tilde{\theta}_{F,H}(e)}-Y
\end{eqnarray*} 
We then conclude that the functions $\phi_{e}$ and $\phi_{\tilde{\theta}_{F,H}(e)}$, restricted to $S_{e}\cap S_{\tilde{\theta}_{F,H}(e)}$, are equal outside of $Y\cup Z$ which is a set of null measure.
\end{proof}

\subsection{Paths and Execution}

We now need to define what is a path, since we won't be able to work with the usual notion of a path in a graph. Obviously, a path will be a finite sequence of edges. We will replace the condition that the source of an edge be equal to the target of the preceding edge by the condition that the intersection of these source and target sets be of non-null measure.

\begin{definition}[Plugging]
Being given two graphings $F,G$, we define their plugging $F\bicolmes G$ as the graphing $F\disjun G$ endowed with the coloring function $\delta: E^{F\disjun G} \rightarrow \{0,1\}$ such that $\delta(e)=1$ if and only if $e\in E^{G}$.
\end{definition}

\begin{definition}[Alternating Paths]
A path in a graphing $F$ is a finite sequence $\{e_{i}\}_{i=0}^{n}$ of elements of $E^{F}$ such that for all $0\leqslant i\leqslant n-1$, $T_{e_{i}}^{F}\cap S_{e_{i+1}}^{F}$ is of strictly positive measure.

An \emph{alternating path} between two graphings $F,G$ is a path $\{e_{i}\}_{i=0}^{n}$ in the graphing $F\bicolmes G$ such that for all $0\leqslant i\leqslant n-1$, $\delta(e_{i})\neq\delta(e_{i+1})$. We will denote by $\pathsmes{F,G}$ the set of alternating paths in $F\bicolmes G$.

We also define the \emph{weight} of a path $\pi=\{e_{i}\}_{i=0}^{n}$ in the graphing $F$ as the scalar $\omega_{\pi}^{F}=\prod_{i=0}^{n}\omega^{F}_{e_{i}}$.
\end{definition}

Given a path $\{e_{i}\}_{i=0}^{n}$ in a graphing $F$, one can define a function $\phi_{\pi}^{F}$ as the partial transformation:
\begin{equation*}
\phi_{\pi}^{F}=\phi_{e_{n}}^{F}\circ \chi_{T_{e_{n}}^{F}\cap S_{e_{n-1}}^{F}}\circ \phi_{e_{n-1}}^{F}\circ\chi_{T_{e_{n-1}}^{F}\cap S_{e_{n-2}}^{F}}\circ\dots\circ\chi_{T_{e_{1}}^{F}\cap S_{e_{0}}^{F}}\circ\phi_{e_{0}}^{F}
\end{equation*}
where for all measurable set $A$, the function $\chi_{A}$ is the partial identity $A\rightarrow A$.

We denote by $S_{\pi}$ and $T_{\pi}$ respectively the domain and codomain of this partial transformation $S_{e_{0}}^{F}\rightarrow T_{e_{n}}^{F}$. It is then clear that the transformation $\phi_{\pi}^{F}: S_{\pi}\rightarrow T_{\pi}$ is measurable. Moreover, if all $\phi_{e_{i}}$ are in a microcosm $\mathfrak{m}$, the transformation $\phi_{\pi}$ is itself in the microcosm $\mathfrak{m}$.

We now introduce the notion of \emph{carving} of a graphing along a measurable set $C$. This operation will consists in replacing an edge by four disjoint edges whose source and target are either subsets of $C$ or subsets of the complementary set of $C$.

\begin{definition}[Carvings]
Let $\phi: S\rightarrow T$ be a measurable transformation, $C$ a measurable set and $C^{c}$ its complementary set. We define the measurable transformations:
\begin{eqnarray*}
[\phi]{}_{i}^{i}&=&\phi\restr{C\cap\phi^{-1}(C)}:A\cap C\cap \phi^{-1}(C)\rightarrow B\cap \phi(C)\cap C\\{}
[\phi]{}_{i}^{o}&=&\phi\restr{C\cap\phi^{-1}(C^{c})}:A\cap C\cap \phi^{-1}(C^{c})\rightarrow B\cap \phi(C)\cap C^{c}\\{}
[\phi]{}_{o}^{i}&=&\phi\restr{C^{c}\cap\phi^{-1}(C)}:A\cap C^{c}\cap \phi^{-1}(C)\rightarrow B\cap \phi(C^{c})\cap C\\{}
[\phi]{}_{o}^{o}&=&\phi\restr{C^{c}\cap\phi^{-1}(C^{c})}:A\cap C^{c}\cap \phi^{-1}(C^{c})\rightarrow B\cap \phi(C^{c})\cap C^{c}
\end{eqnarray*}
We will denote by $[S]_{a}^{b}, [T]_{a}^{b}$ ($a,b\in\{i,o\}$) the domain and codomain of $[\phi]{}_{a}^{b}$. 

If $F$ is a graphing in a microcosm $\mathfrak{m}$, define the \emph{carving of $F$ along $C$} as the graphing $F^{\decoupe{C}}=\{(\omega_{e}^{F},[\phi_{e}^{F}]{}_{a}^{b})~|~e\in E^{F}, a,b\in\{i,o\}\}$ which is a graphing in the microcosm $\mathfrak{m}$.
\end{definition}

In some cases, the carving a graphing $G$ along a measurable set $C$ is almost the same as $G$. Indeed, if each edge have its source and target (up to a null-measure set) either in $C$ or in its complementary set, the graphing obtained from the carving operation is almost everywhere equal to $G$.

\begin{definition}
Let $A,B$ be two measurable sets. We say that $A$ intersects $B$ trivially if $\lambda(A\cap B)=0$ or $\lambda(A\cap B^{c})=0$.

If $F$ is a graphing and for all $e\in E^{F}$, $S^{F}_{e}$ and $T_{e}^{F}$ intersect $C$ trivially, then $F$ will be said to be $C$-tough.
\end{definition}

\begin{lemma}
Let $F$ be a graphing and $C$ be a measurable set. If $F$ is $C$-tough, then $F^{\decoupe{C}}\sim_{a.e.} F$.
\end{lemma}

\begin{proof}
Chose $e\in E^{F}$. Since $F$ is $C$-tough, we are in one of the four following cases:
\begin{itemize}[noitemsep,nolistsep]
\item $S^{F}_{e}\cap C$ and $T^{F}_{e}\cap C$ are of null measure;
\item $S^{F}_{e}\cap C$ and $T^{F}_{e}\cap C^{c}$ are of null measure;
\item $S^{F}_{e}\cap C^{c}$ and $T^{F}_{e}\cap C$ are of null measure;
\item $S^{F}_{e}\cap C^{c}$ and $T^{F}_{e}\cap C^{c}$ are of null measure;
\end{itemize}
These four cases are treated in a similar way. Indeed, among the functions $[\phi^{F}_{e}]{}_{a}^{b}$, $a,b\in\{i,o\}$, only one is of domain (and thus of codomain) a set of strictly positive measure. We thus define an empty graphing $0_{F}$, with $E^{0_{F}}=E^{F}\times\{1,2,3\}$, and a bijection $E^{F^{\decoupe{C}}}\rightarrow E^{F}\disjun E^{0_{F}}$ which associates to the element $(e,a,b)$ ($e\in E^{F}$, $a,b\in\{i,o\}$) the element $e\in E^{F}$if the domain of $[\phi_{e}^{F}]{}_{a}^{b}$ is of strictly positive measure, and one of the elements $(e,i)\in E^{0_{F}}$ otherwise. One can then easily show that this bijection satisfies all the necessary properties to conclude that $F\sim_{a.e.} F^{\decoupe{C}}$.
\end{proof}

Thanks to the carving operation, we are now able to define the execution of two graphings $F$ and $G$: we consider the set of alternating paths between $F$ and $G$, and we then keep the part of each path which is external to the intersection $C$ -- the location of the \emph{cut} -- of the carriers of $F$ and $G$. The execution for graphings is therefore the natural generalisation of the execution we defined earlier on graphs.

\begin{definition}[Execution]
Let $F,G$ be two graphings in a microcosm $\mathfrak{m}$ of respective carriers $V^{F},V^{G}$ and let $C=V^{F}\cap V^{G}$. We define the execution of $F$ and $G$, denoted by $F\plugmes G$, as the graphing in the microcosm $\mathfrak{m}$ with carrier $V^{F}\cup V^{G}-C$ defined as follows (the carving is taken along $C$): 
\begin{equation*}
\{(\omega_{\pi}^{F\bicolmes G},\phi_{\pi}^{F\bicolmes G}: [S_{\pi}]_{o}^{o}\rightarrow [T_{\pi}]_{o}^{o})~|~\pi\in\pathsmes{F,G}, \lambda([S_{\pi}]_{o}^{o})\neq 0\}
\end{equation*}
\end{definition}

\subsection{Cycles}

\begin{definition}[Alternating Cycles]
A cycle in a graphing $F$ is a path $\{e_{i}\}_{i=0}^{n}$ in $F$ such that $S_{e_{0}}^{F}\cap T_{e_{n}}^{F}$ is of strictly positive measure.

An alternating cycle between two graphings $F,G$ is a cycle $\{e_{i}\}_{i=0}^{n}$ in $F\bicolmes G$ which is an alternating path and such that $\delta(e_{0})\neq\delta(e_{n})$. We will denote by $\cyclesmes{F,G}$ the set of alternating cycles between $F$ and $G$.
\end{definition}

\begin{proposition}\label{cheminspp}
Let $F,F',G$ be graphings such that $F\sim_{a.e.} F'$. Then there exists a bijection 
$$\theta: \pathsmes{F,G}\rightarrow \pathsmes{F',G}$$
such that $\phi_{\pi}=_{a.e.}\phi_{\theta(\pi)}$ for all path $\pi$.
\end{proposition}

\begin{proof}
By definition, there exists two empty graphings $0_{F},0_{F'}$ and a bijection $\theta: E^{F'}\disjun E^{0_{F'}}\rightarrow E^{F}\disjun E^{0_{F}}$ such that:
\begin{itemize}[noitemsep,nolistsep]
\item for all $e\in E^{F'}\disjun E^{0_{F'}}$, $\omega_{e}^{F'\cup 0_{F}}= \omega_{\theta(e)}^{F\cup 0_{F}}$;
\item for all $e\in E^{F'}\disjun E^{0_{F'}}$, $S_{e}^{F'\cup 0_{F'}}\Delta S_{\theta(e)}^{F\cup 0_{F}}$ is of null measure;
\item for all $e\in E^{F'}\disjun E^{0_{F'}}$, $T_{e}^{F\cup 0_{F'}}\Delta T_{\theta(e)}^{F\cup 0_{F}}$ is of null measure;
\item for all $e\in E^{F'}\disjun E^{0_{F'}}$, $\phi_{\theta(e)}^{F\disjun 0_{F}}$ et $\phi_{e}^{F'\disjun 0_{F'}}$ are almost everywhere equal;
\end{itemize}
Let $\pi$ be an alternating path in $F\bicolmes G$. We treat the case $\pi=f_{0}g_{0}\dots f_{n}g_{n}$ as an example, the other cases are dealt with in a similar way. We can define a path $\theta^{-1}(\pi)$ in $F'\bicolmes G$ by $\pi'=\theta^{-1}(f_{0})g_{0}\theta^{-1}(f_{1})\dots \theta^{-1}(f_{n})g_{n}$. Indeed, since the sets $S_{f_{i+1}}\cap T_{g_{i}}$ (resp. $S_{g_{i}}\cap T_{f_{i}}$) are of strictly positive measure, then $\theta^{-1}(f_{i+1})$ (resp. $\theta^{-1}(f_{i})$) is of strictly positive measure (hence an element of $E^{F'}$) and moreover satisfies that $S_{\theta^{-1}(f_{i+1})}\cap T_{g_{i}}$ (resp. $S_{g_{i}}\cap T_{\theta^{-1}(f_{i})}$) is of strictly positive measure. That the realisers are almost everywhere equal is clear.

Conversely, a path $\pi'=e_{0}g_{0}\dots e_{n}g_{n}$ in $F'\bicolmes G$ allows one to define a path $\theta(\pi')=\theta(e_{0})g_{0}\theta(e_{1})\dots \theta(e_{n})g_{n}$. It is clear that $\theta(\theta^{-1}(\pi))=\pi$ (resp. $\theta^{-1}(\theta(\pi'))=\pi'$) for all path $\pi$ (resp. $\pi'$) in $F\bicolmes G$ (resp. $F'\bicolmes G$).
\end{proof}

We also need to check that the operation of execution is compatible with the notion of almost everywhere equality. Indeed, since we want to work with graphings considered up to almost everywhere equality, the result of the execution should not depend on the representative of the equivalence class considered.

\begin{corollary}
Let $F,F',G$ be graphings such that $F\sim_{a.e.} F'$. Then $F\plugmes G\sim_{a.e.} F'\plugmes G$.
\end{corollary}

\begin{proof}
Let $\theta$ be the bijection defined in the statement of the preceding proposition. We notice that $\omega_{\pi}=\omega_{\theta^{-1}(\pi)}$, and that $\phi_{\pi}$ and $\phi_{\theta^{-1}(\pi)}$ are almost everywhere equal as compositions of pairwise almost everywhere equal maps. In particular, their domain and codomain are equal up to a set of null measure. We can then conclude that $\theta:E^{F'\plugmes G}\rightarrow E^{F\plugmes G}$ satisfies all the necessary properties: $F'\plugmes G$ and $F\plugmes G$ are almost everywhere equal.
\end{proof}

\begin{corollary}
Let $F,F',G$ be graphings such that $F\sim_{a.e.} F'$. Then there exists a bijective correspondence sending each map in $\cyclesmes{F, G}$ to an (almost everywhere) equal map in $\cyclesmes{F', G}$.
\end{corollary}

\begin{proof}
Let $\theta$ be the bijection defined in the proof of \autoref{cheminspp}.
The functions $\phi_{\pi}$ and $\phi_{\theta(\pi)}$ are almost everywhere equal and their domains and codomains are equal up to a set of null measure. We can deduce from this that $[\phi_{\pi}]_{i}^{i}$ and $[\phi_{\theta^{-1}(\pi)}]_{i}^{i}$ are almost everywhere equal, and their domains and codomains are equal up to a set of null measure.
\end{proof}

\subsection{Carvings, Cycles and Execution}

We now show a technical result that will be useful later, and which gives better insights on the operation of execution between two graphings. The execution of the graphings $F$ and $G$ is defined as a restriction of the set of alternating paths between $F$ and $G$. One could have also considered the carvings of $F$ and $G$ along the intersection $C$ of the carriers of $F$ and $G$, and then define the execution as the set of alternating paths whose source and target lie outside of the set $C$. The technical lemma we now state and prove shows that these two operations are equivalent.

Let $F,G$ be two graphings and $C=V^{F}\cap V^{G}$. One can notice that there should be a bijective correspondence between the edges of $F\plugmes G$ and those of $F^{\decoupe{C}}\plugmes G^{\decoupe{C}}$. Indeed, for two edges $g,f$ to follow each other in a path, one should have that $S_{g}\cap T_{f}$ is of strictly positive measure. But since $S_{g}\cap T_{f}$ and $S_{f}\cap T_{g}$ are subsets of $C$, the following expressions are equal:
\begin{equation*}
\chi_{S_{g}\cap T_{f}}\circ \phi_{f}^{F}\circ\chi_{S_{f}\cap T_{g}}~~~~~\text{and}~~~~~\chi_{S_{g}\cap T_{f}\cap C\cap \phi_{f}^{F}(C)}\circ (\phi_{f}^{F})\restr{C\cap (\phi_{f}^{F})^{-1}(C)}\circ\chi_{S_{f}\cap T_{g}\cap C\cap (\phi_{f}^{F})^{-1}(C)}
\end{equation*}
One can deduce from this the following equality:
\begin{equation*}
\chi_{S_{g}\cap T_{f}}\circ \phi_{f}^{F}\circ\chi_{S_{f}\cap T_{g}}=\chi_{S_{g}\cap [T_{f}]_{i}^{i}}\circ [\phi_{f}^{F}]_{i}^{i}\circ\chi_{[S^{f}]_{i}^{i}\cap T_{g}}
\end{equation*}
One can obtain the following equalities in a similar manner: 
\begin{eqnarray*}
\chi_{C^{c}}\circ \phi_{f}\circ \chi_{S_{f}\cap T_{g}}&=&\chi_{C^{c}}\circ [\phi_{f}]_{i}^{o}\circ \chi_{[S_{f}]_{i}^{o}\cap T_{g}}\\
\chi_{S_{g}\cap T_{f}}\circ \phi_{f}\circ \chi_{C^{c}}&=&\chi_{S_{g}\cap[T_{f}]_{o}^{i}}\circ [\phi_{f}]_{o}^{i}\circ \chi_{C^{c}}\\
\chi_{C^{c}}\circ \phi_{f}\circ \chi_{C^{c}}&=&\chi_{C^{c}}\circ [\phi_{f}]_{o}^{o}\circ\chi_{C^{c}}
\end{eqnarray*}

\begin{lemma}
Let $F,G$ be two graphings, $V^{F},V^{G}$ their carrier and $C=V^{F}\cap V^{G}$. Then:
\begin{equation*}
F\plugmes G=F^{\decoupe{C}}\plugmes G
\end{equation*}
\end{lemma}

\begin{proof}
By definition, the execution $F\plugmes G$ is the graphing:
\begin{equation*}
\{(\omega_{\pi}^{F\bicolmes G},\phi_{\pi}^{F\bicolmes G}: [S_{\pi}]_{o}^{o}\rightarrow [T_{\pi}]_{o}^{o})~|~\pi\in\pathsmes{F,G}, \lambda([S_{\pi}]_{o}^{o})\neq 0\}
\end{equation*}
Similarly, the execution $F^{\decoupe{C}}\plugmes G$ is the graphing:
\begin{equation*}
\{(\omega_{\pi}^{F^{\decoupe{C}}\bicolmes G},\phi_{\pi}^{F^{\decoupe{C}}\bicolmes G}: [S_{\pi}]_{o}^{o}\rightarrow [T_{\pi}]_{o}^{o})~|~\pi\in\pathsmes{F^{\decoupe{C}},G}, \lambda([S_{\pi}]_{o}^{o})\neq 0\}
\end{equation*}

Let $\pi$ be an alternating path in $F\bicolmes G$. Then $\pi$ is an alternating sequence of elements in $E^{F}$ and elements in $E^{G}$. Suppose for instance $\pi=f_{0} g_{0}f_{1}\dots f_{k}g_{k}f_{k+1}$, and let us define $\tilde{\pi}=[f_{0}]_{i}^{o} g_{0}[f_{1}]_{i}^{i}\dots [f_{k}]_{i}^{i}g_{k}[f_{k+1}]_{o}^{i}$. The function $[\phi_{\pi}^{F\bicolmes G}]_{o}^{o}$ is equal to:
\begin{multline*}
\chi_{C^{c}\cap \phi_{\pi}^{F\bicolmes G}(C^{c})}\circ \phi_{f_{k+1}}^{F}\circ \chi_{S_{f_{k+1}}\cap T_{g_{k}}}\circ\phi_{g_{k}}^{G}\circ\dots\\
\dots\circ\phi_{g_{i+1}}^{G} \circ\chi_{S_{g_{i+1}}\cap T_{f_{i+1}}}\circ \phi_{f_{i+1}}^{F}\circ \chi_{S_{f_{i+1}}\cap T_{g_{i}}}\circ\phi_{g_{i}}^{G}\circ\dots\\
\dots\circ\phi_{g_{0}}^{G} \circ\chi_{S_{g_{i+1}}\cap T_{f_{i+1}}}\circ \phi_{f_{i+1}}^{F}\circ \chi_{C^{c}\cap (\phi_{\pi}^{F\bicolmes G})^{-1}(C^{c})}
\end{multline*}
From the remarks preceding the statement of the lemma, one can conclude that $[\phi_{\pi}^{F\bicolmes G}]_{o}^{o}$ is equal to:
\begin{multline*}
\chi_{C^{c}\cap \phi_{\tilde{\pi}}^{F^{\decoupe{C}}\bicolmes G}(C^{c})}\circ [\phi_{f_{k+1}}^{F}]_{i}^{o}\circ \chi_{[S_{f_{k+1}}]_{i}^{o}\cap T_{g_{k}}}\circ\phi_{g_{k}}^{G}\circ\dots\\
\dots\circ\phi_{g_{i+1}}^{G} \circ\chi_{S_{g_{i+1}}\cap [T_{f_{i+1}}]_{i}^{i}}\circ [\phi_{f_{i+1}}^{F}]_{i}^{i}\circ \chi_{[S_{f_{i+1}}]_{i}^{i}\cap T_{g_{i}}}\circ\phi_{g_{i}}^{G}\circ\dots\\
\dots\circ\phi_{g_{0}}^{G} \circ\chi_{S_{g_{i+1}}\cap [T_{f_{i+1}}]_{o}^{i}}\circ [\phi_{f_{i+1}}^{F}]_{o}^{i}\circ \chi_{C^{c}\cap (\phi_{\tilde{\pi}}^{F^{\decoupe{C}}\bicolmes G})^{-1}(C^{c})}
\end{multline*}
We therefore obtain that $[\phi_{\pi}^{F\bicolmes G}]_{o}^{o}=[\phi_{\tilde{\pi}}^{F^{\decoupe{C}}\bicolmes G}]_{o}^{o}$. Conversely, each alternating path in $F^{\decoupe{C}}\bicolmes G$ whose first and last edges are elements of $F^{\decoupe{C}}$ is necessarily of the form $[f_{0}]_{i}^{o} g_{0}[f_{1}]_{i}^{i}\dots [f_{k}]_{i}^{i}g_{k}[f_{k+1}]_{o}^{i}$ where the path $f_{0} g_{0}f_{1}\dots f_{k}g_{k}f_{k+1}$ is an alternating path in $F\bicolmes G$. 

The other cases are treated in a similar way.
\end{proof}

%\begin{lemma}
%Soit $F,G$ de graphages, $V^{F},V^{G}$ leur lieux et $C=V^{F}\cap V^{G}$. Alors il existe une bijection:
%\begin{equation*}
%\cyclesmes{F,G}\cong\cyclesmes{F^{\decoupe{C}},G}
%\end{equation*}
%\end{lemma}

\begin{corollary}\label{decoupeinvariant}
Let $F,G$ be two graphings, $V^{F},V^{G}$ their carrier and $C=V^{F}\cap V^{G}$. Then:
\begin{equation*}
F\plugmes G=F^{\decoupe{C}}\plugmes G^{\decoupe{C}}
\end{equation*}
\end{corollary}

We can also show the sets of cycles are equal.

\begin{lemma}
Let $F,G$ be two graphings, $V^{F},V^{G}$ their carrier, and $C=V^{F}\cap V^{G}$. Then:
\begin{equation*}
\cyclesmes{F,G}=\cyclesmes{F^{\decoupe{C}},G}
\end{equation*}
\end{lemma}

\begin{proof}
The argument is close to the one used in the preceding proof. Indeed, if $\pi=e_{0}\dots e_{n}$ is an alternating cycle between $F$ and $G$, then we can associate it to the cycle $[\pi]=[e_{0}]_{i}^{i}[e_{1}]_{i}^{i}\dots[e_{n}]^{i}_{i}$. Conversely, if $\pi'$ is a cycle in $F^{\decoupe{C}}\bicolmes G$, then each edge in $\pi'$ necessarily is of the form $[e]_{i}^{i}$ for an element $e$ in $F$. Moreover, the associated functions are equal, i.e.\ $[\phi_{\pi}]_{i}^{i}=\phi_{[\pi]}$.
\end{proof}

Notice that the following lemma is the only place where the fact that our transformations are non-singular is used. It is however fundamental, as it is the key to obtain the associativity of execution.

\begin{lemma}\label{lemmesupportchem}
Let $F$ be a graphing, and $\pi=e_{0}\dots e_{n}$ a path in $F$ such that $S_{\pi}$ is of strictly positive measure. We define, for all couple of integers $i<j$, $\rho_{i,j}$ the path $e_{i}e_{i+1}\dots e_{j}$. Then:
\begin{itemize}[noitemsep,nolistsep]
\item for all $0<i<j\leqslant n$, $S_{\rho_{i,j}}\cap T_{e_{i-1}}$  is of strictly positive measure;
\item for all $0\leqslant i<j<n$, $T_{\rho_{i,j}}\cap S_{e_{j+1}}$ is of strictly positive measure.
\end{itemize}
\end{lemma}

\begin{proof}
Let us fix $i,j$. We suppose that $S_{\rho_{i,j}}\cap T_{e_{i-1}}$ is of null measure. Then for all $x\in S_{pi}$, $\phi_{e_{0}\dots e_{i-1}}$ is defined at $x$, and such that $\phi_{e_{i}\dots e_{n}}$ is defined at $\phi_{e_{0}\dots e_{i-1}}(x)$. In particular, $\phi_{e_{i}\dots e_{j}}$ is defined at $\phi_{e_{0}\dots e_{i-1}}(x)$, i.e.\ $\phi_{e_{0}\dots e_{i-1}}(x)$ is an element in $S_{\rho_{i,j}}$. Moreover, by the definition, $\phi_{e_{0}\dots e_{i-1}}(x)$ is an element of $T_{e_{i-1}}$. Thus $\phi_{e_{0}\dots e_{i-1}}(S_{\pi})\subset S_{\rho_{i,j}}\cap T_{e_{i-1}}$. Since $\phi_{e_{0}\dots e_{i-1}}$ is a non-singular transformation which is defined at all $x\in S_{\pi}$, we deduce that $\lambda(S_{\pi})=0\Leftrightarrow \lambda(S_{\rho_{i,j}}\cap T_{e_{i-1}})=0$. This lead us to a contradiction since this implies that $\lambda(S_{\pi})=0$.

A similar argument shows that $T_{\rho_{i,j}}\cap S_{e_{j+1}}$ is of strictly positive measure.
\end{proof}
%
%\begin{remark}
%In the preceding proof, we showed that $\lambda(S_{\pi})\leqslant  \lambda(S_{\rho_{i,j}}\cap T_{e_{i-1}})$ for all $i,j$. In particular, $\lambda(S_{\pi})\leqslant \lambda(S_{e_{i}}\cap T_{e_{i-1}})$, and therefore:
%\begin{equation*}
%\lambda(S_{\pi})\leqslant \min\{\lambda(S_{e_{i}}\cap T_{e_{i-1}})~|~i=1\dots n\}
%\end{equation*}
%\end{remark}

As this was the case with the execution between graphs in earlier constructions \cite{seiller-goim,seiller-goia}, we can show the associativity of execution under the hypothesis that the intersection of the carriers is of null measure.

\begin{theorem}[Associativity of Execution]\label{thm_associativity}
Let $F,G,H$ be three graphings such that $\lambda(V^{F}\cap V^{G}\cap V^{H})=0$. Then:
\begin{equation*}
F\plugmes(G\plugmes H)=(F\plugmes G)\plugmes H
\end{equation*}
\end{theorem}

\begin{proof}
We can first suppose that $F$ (resp. $G$, resp. $H$) is $C_{F}=V^{F}\cap(V^{G}\cup V^{H})$-tough (resp. $C_{G}=V^{G}\cap (V^{F}\cup V^{H})$-tough, resp. $C_{H}=V^{H}\cap(V^{F}\cup V^{G})$-tough). Indeed, if this was not the case, we can always consider the carving along the set $C_{F}$ (resp. $C_{G}$, resp. $C_{H}$). This simplifies the following argument since it allows us to consider paths instead of restrictions of paths. The proof then follows the proof of the associativity of the execution for directed weighted graphs obtained in our previous papers \cite{seiller-goim,seiller-goia}. We detail it anyway for the sake of self-containment.

We can define the simultaneous plugging of the three graphings $F,G,H$ as the graphing $F\disjun G\disjun H$ endowed with a coloring map $\delta$ defined by $\delta(e)=0$ when $e\in E^{F}$, $\delta(e)=1$ when $e\in E^{G}$ and $\delta(e)=2$ when $e\in E^{H}$. We can then define the set of $3$-alternating paths between $F,G,H$ as the paths $e_{0}e_{1}\dots e_{n}$ such that $\delta(e_{i})\neq\delta(e_{i+1})$.

If $e_{0}f_{0}e_{1}\dots f_{k-1}e_{k}f_{k}$ is an alternating path in $F\plugmes(G\plugmes H)$, where every $e_{i}$ is an alternating path $e_{i}=g_{0}^{i}h_{0}^{i}\dots g_{n_{i}}^{i}h_{n_{i}}^{i}$, then the sequence of edges obtained by replacing each $e_{i}$ by the associated sequence (and forgetting about parentheses)  is a path. Indeed, we know that, for instance, $S_{e_{i}}\cap T_{f_{i-1}}$ is of strictly positive measure, and $S_{e_{i}}\subset S_{g_{0}}$, thus $S_{g_{0}}\cap T_{f_{i-1}}$ is of strictly positive measure. We therefore defined a $3$-alternating path between $F, G$ and $H$. The two paths define the same measurable partial transformation, and have the same domains and codomains.

Conversely, if $e_{0}e_{1}\dots e_{n}$ is a $3$-alternating path between $F, G$ and $H$, then we can see it as an alternating sequence of edges in $F$ and alternating sequences between $G$ and $H$. Let $\pi=g_{0}h_{0}\dots g_{k}h_{k}$ be the path defined by such a sequence appearing in the path $e_{0}\dots e_{n}$. We can use the preceding lemma to ensure that $S_{\pi}\cap  T_{e_{j}}$ is of strictly positive measure. Similarly, $T_{\pi}\cap S_{e_{k}}$ is of strictly positive measure. We thus showed that we had an edge in $F\plugmes(G\plugmes H)$. The two paths define the same partial measurable transformation, and have the same domains and codomains.
\end{proof}

\section{Cycles, Circuits and The Trefoil Property}\label{sec_measurement}

In this section, we go a bit further into the theory of graphings. Indeed, one of the main motivations behind the use of continuous sets as vertices of a graph instead of usual discrete sets lies in the idea that edges of a graphing may be \emph{split} into sub-edges. In order to formalize this idea (illustrated in \autoref{illustrationrefinement}), we define the notion of refinement of a graphing. Once this notion defined, we will want to consider graphings \emph{up to refinement}, meaning that a graphing and one of it refinements should intuitively represent the same computation. We therefore define a equivalence relation on the set of graphings by saying that two graphings are equivalent if they possess a common refinement. We show that this equivalence relation is compatible with the execution defined in the previous section. 

We then explore the notion of cycles between graphings. If a cycle might be defined in the obvious and natural way, we want to define a measurement which is compatible with the equivalence relation based on refinements. This means that the measurement should be \enquote{refinement-invariant}, something that involves lots of complex combinatorics. Moreover, this measurement should satisfy the trefoil property with respect to the execution defined earlier. We therefore introduce the notion of \enquote{circuit-quantifying map} which satisfies two abstract properties. We then show that any map satisfying these properties defines a measurement which is both refinement-invariant and satisfies the trefoil property.

\subsection{Refinements}

We now define the notion of \emph{refinement} of a graphing. This a very natural operation to consider. A simple example of refinement is to consider a graphing $F$ and one of its edges $e\in E^{F}$: one can obtain a refinement of $F$ by replacing $e$ with two edges $f,f'$ such that $S_{f}\cup S_{f'}=S_{e}$ and $S_{f}\cap S_{f'}$ is of null measure (one should then define $T_{f}=\phi_{e}(S_{f})$ and $T_{f'}=\phi_{e}(S_{f'})$ accordingly). This is illustrated in \autoref{illustrationrefinement}.

\begin{figure}
\begin{center}
\subfigure[A graphing $G$]{
\begin{tikzpicture}
	\draw[-] (0,0) -- (2,0) node [below,very near start] {$[0,2]$};
	\draw[-] (3,0) -- (5,0) node [below,very near end] {$[3,5]$};
	\node (1) at (1,0) {};
	\node (4) at (4,0) {};
	
	\draw[->,red] (1) .. controls (1,1.5) and (4,1.5) .. (4) node [midway,above] {$x\mapsto 5-x$};
\end{tikzpicture}}
\subfigure[A refinement of $G$]{
\begin{tikzpicture}
	\draw[-] (7.5,0) -- (8.5,0) node [below,very near start] {$[0,1]$};
	\draw[-] (9,0) -- (10,0) node [below,very near start] {$[1,2]$};
	\draw[-] (10.5,0) -- (11.5,0) node [below,very near end] {$[3,4]$};
	\draw[-] (12,0) -- (13,0) node [below,very near end] {$[4,5]$};
	\node (8) at (8,0) {};
	\node (95) at (9.5,0) {};
	\node (11) at (11,0) {};
	\node (125) at (12.5,0) {};
	
	\draw[->,red] (8) .. controls (8,2) and (12.5,2) .. (125) node [midway,above] {$x\mapsto 5-x$};
	\draw[->,red] (95) .. controls (9.5,1) and (11,1) .. (11) node [midway,above] {$x\mapsto 5-x$};
\end{tikzpicture}}
\end{center}
\caption{Illustration of refinement}\label{illustrationrefinement}
\end{figure}
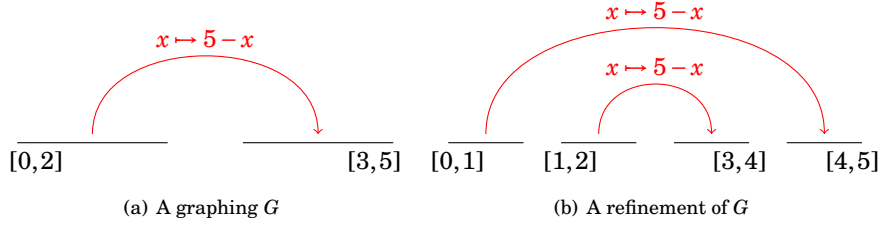

\begin{definition}[Refinements]
Let $F,G$ be two graphings. We will say that $F$ is a \emph{refinement} of $G$ -- denoted by $F\leqslant G$ -- if there exists a function $\theta:E^{F}\rightarrow E^{G}$ such that:
\begin{itemize}[noitemsep,nolistsep]
\item for all $e,e'\in E^{F}$ such that $\theta(e)=\theta(e')$ and $e\neq e'$, $S_{e}^{F}\cap S_{e'}^{F}$ and $T_{e}^{F}\cap T_{e'}^{F}$ are of null measure;
%\item pour tout $e\in E^{G}$, on a $\card{\theta^{-1}(e)}<\infty$;
\item for all $e\in E^{F}$, $\omega_{\theta(e)}^{G}=\omega_{e}^{F}$;
\item for all $f\in E^{G}$, $S_{f}^{G}$ and $\cup_{e\in\theta^{-1}(f)} S_{e}^{F}$ are equal up to a set of null measure;
\item for all $f\in E^{G}$, $T_{f}^{G}$ and $\cup_{e\in\theta^{-1}(f)} T_{e}^{F}$ are equal up to a set of null measure;
\item for all $e\in E^{F}$, $\phi_{\theta(e)}^{G}$ and $\phi_{e}^{F}$ are equal almost everywhere $S_{\theta(e)}^{G}\cap S_{e}^{F}$.
\end{itemize}
We will say that $F$ is a \emph{refinement of $G$ along $g\in E^{G}$} if there exists a set $D$ of elements of $E^{F}$ such that:
\begin{itemize}[noitemsep,nolistsep]
\item $\theta^{-1}(g)=D$;
\item $\theta\restr{E^{F}-D}: E^{F}-D\rightarrow E^{G}-\{g\}$ is bijective.
\end{itemize}
If $D$ contains only two elements, we will say that $F$ is a \emph{simple refinement along $g$}. The refinements will sometimes be written $(F,\theta)$ in order to precise the function $\theta$. 
\end{definition}

\begin{proposition}
We define the relation $\sim_{\leqslant}$ on the set of graphings as follows:
\begin{equation*}
F\sim_{\leqslant} G\Leftrightarrow \exists H, (H\leqslant F)\wedge (H\leqslant G)
\end{equation*}
This is an equivalence relation.
\end{proposition}

\begin{proof}
Reflexivity and symmetry are straightforward. We are therefore left with transitivity: let $F,G,H$ be three weighted graphs such that $F\sim_{\leqslant} G$ and $G\sim_{\leqslant} H$. We will denote by $(P_{F,G},\theta)$ (resp. $(P_{G,H},\rho)$) a common refinement of $F$ and $G$ (resp. of $G$ and $H$). We will now define a graphing $P$ such that $P\leqslant P_{F,G}$ and $P\leqslant P_{G,H}$. Let us define:
\begin{itemize}[noitemsep,nolistsep]
\item $E^{P}=\{(e,f)\in E^{P_{F,G}}\times E^{P_{G,H}}~|~\theta(e)=\rho(f)\}$;
\item $S^{P}_{(e,f)}=S^{P_{F,G}}_{e}\cap S^{P_{G,H}}_{f}$ when $\theta(e)=\rho(f)$;
\item $T^{P}_{(e,f)}=T^{P_{F,G}}_{e}\cap T^{P_{G,H}}_{f}$ when $\theta(e)=\rho(f)$;
\item $\omega^{P}_{(e,f)}=\omega^{P_{F,G}}_{e}=\omega^{P_{G,H}}_{f}$;
\item $\phi^{P}_{(e,f)}$ is the restriction of $\phi^{P_{F,G}}_{e}$ to $S^{P}_{(e,f)}$;
\item $\mu_{F,G}:(e,f)\mapsto e$ and $\mu_{G,H}: (e,f)\mapsto f$.
\end{itemize}
It is then easy to check that $(P,\mu_{F,G})$ (resp. $(P,\mu_{G,H})$) is a refinement of $P_{F,G}$ (resp. of $P_{G,H}$).

Since $(P_{F,G},\theta)$ is a refinement of $F$ and $(P,\mu_{F,G})$ is a refinement of $P_{F,G}$, it is clear that $(P,\theta\circ\mu_{F,G})$ is a refinement of $F$. In a similar way, $(P,\theta\circ\mu_{G,H})$ is a refinement of $H$. Finally, $P\leqslant F$ and $P\leqslant H$, which shows that $F\sim_{\leqslant} H$.
\end{proof}

\begin{proposition}\label{subsetrelation}
The relation $\sim_{\leqslant}$ contains the relation $\sim_{a.e.}$.
\end{proposition}

\begin{proof}
Let $F,G$ be two graphings such that $F\sim_{a.e.} G$. We will show that $F\sim_{\leqslant} G$. We will use the notations of \autoref{equivalencepp}: $0_{F},0_{G}$ for the empty graphings and $\theta$ for the bijection between the sets of vertices. First, we notice that $F\leqslant F\disjun 0_{F}$ and $G\leqslant G\disjun 0_{G}$. As a consequence, $F\sim_{\leqslant} F\disjun 0_{F}$ and $G\sim_{\leqslant} G\uplus 0_{G}$. Moreover, the bijection $\theta: E^{F}\disjun E^{0_{F}}\rightarrow E^{G}\disjun E^{0_{G}}$ clearly satisfies the necessary conditions for $(F\disjun 0_{F},\theta)$ to be a refinement of $G\disjun 0_{G}$, which implies that $F\disjun 0_{F}\sim_{\leqslant} G\disjun 0_{G}$. Using the transitivity of $\sim_{\leqslant}$, we can now conclude that $F\sim_{\leqslant} G$.
\end{proof}

Of course, the carving of a graphing $G$ along a measurable set $C$ defines a refinement of $G$ where each edge is replaced by exactly four disjoint edges. This will be of use to simplify some conditions later: a measurement that is invariant under refinements will be invariant under carvings too.

\begin{proposition}
Let $G$ be a graphing and $C$ a measurable set. The graphing $G^{\decoupe{C}}$ is a refinement of  $G$. 
\end{proposition}

\begin{proof}
It is sufficient to verify that the function $\theta: E^{G^{\decoupe{C}}}\rightarrow E^{G}$, $(e,a,b)\rightarrow e$ satisfies all the necessary conditions. Firstly, using the definition, the weights $\omega_{(e,a,b)}^{G^{\decoupe{C}}}$ and $\omega_{e}^{G}$ are equal. Then, the sets $[S_{e}^{G}]_{a}^{b}$, $a,b\in\{i,o\}$ (resp. $[T_{e}^{G}]_{a}^{b}$) define a partition of $S_{e}^{G}$ (resp. $T_{e}^{G}$). Finally, using the definition again, $[\phi_{e}^{G}]_{a}^{b}$ is equal to $\phi_{e}^{G}$ on its domain.
\end{proof}

\begin{lemma}
Let $F,G$ be two graphings, $e\in E^{F}$ and $F^{(e)}$ be a simple refinement of $F$ along $e$. Then $F^{(e)}\plugmes G$ is a refinement of $F\plugmes G$.
\end{lemma}

\begin{proof}
By definition,
\begin{equation*}
F\plugmes G=\{(\omega_{\pi}^{F\bicolmes G},\phi_{\pi}^{F\bicolmes G}: [S_{\pi}]_{o}^{o}\rightarrow [T_{\pi}]_{o}^{o})~|~\pi\in\pathsmes{F,G}, \lambda([S_{\pi}]_{o}^{o})\neq 0\}
\end{equation*}
Since $F^{(e)}$ is a simple refinement of $F$ along $e$, there exists a partition of $S_{e}^{F}$ in two sets $S_{1},S_{2}$, and a partition of $T_{e}^{F}$ in two sets $T_{1},T_{2}$ such that $\phi_{e}^{F}(S_{i})=T_{i}$. We can suppose, without loss of generality, that $S_{1}\cap S_{2}=\emptyset$ since there exists a graphing which is almost everywhere equal to $F^{(e)}$ and satisfies this additional condition, and since execution is compatible with almost everywhere equality. This additional assumption implies in particular that $T_{1}\cap T_{2}$ is of null measure. We denote by $f_{1},f_{2}$ the two elements of $E^{F^{(e)}}$ whose image is $e$ by $\theta$.

To any element $\pi=\{e_{i}\}_{i=0}^{n}$ of $F^{(e)}\plugmes G$, we associate the path $\theta(\pi)=\{\theta(a_{i})\}_{i=0}^{n}$. We now need to check that this is indeed a refinement. Let $\pi_{1},\pi_{2}$ be two distinct paths such that $\theta(\pi_{1})=\theta(\pi_{2})$. We want to show that $S_{\pi_{1}}\cap S_{\pi_{2}}$ is of null measure. Since $\pi_{1}=\{p_{i}\}_{i=0}^{n_{1}}$ and $\pi_{2}=\{q_{i}\}_{i=0}^{n_{2}}$ are distinct, they differ at least on one edge. Let $k$ be the smallest integer such that $p_{k}\neq q_{k}$. We can suppose without loss of generality that $p_{k}=f_{1}$ and $q_{k}=f_{2}$. If $x\in S_{\pi_{1}}$, then $x\in\phi_{p_{0}\dots p_{k-1}}^{-1}(S_{1})$. Similarly, if $x\in S_{\pi_{2}}$, then $x\in\phi_{q_{0}\dots q_{k-1}}^{-1}(S_{2})=\phi_{p_{0}\dots p_{k-1}}^{-1}(S_{2})$. Since we supposed that $S_{1}\cap S_{2}=\emptyset$, we deduce that $S_{\pi_{1}}\cap S_{\pi_{2}}=\emptyset$.

By definition, the weight of a path $\pi$ is equal to the weight of every path $\pi'$ such that $\theta(\pi')=\pi$. Moreover, the functions $\phi_{\pi'}$ and $\phi_{\theta(\pi')}$ are by definition almost everywhere equal on the intersection of their domain since every $\phi_{e}$ is almost everywhere equal to $\phi_{\theta(e)}$.

We are now left to show $S_{\pi}=\cup_{\pi'\in\theta^{-1}(\pi)} S_{\pi'}$ (the result concerning $T_{\pi}$ is then obvious). It is clear that $S_{\pi'}\subset S_{\pi}$ when $\theta(\pi')=\pi$, and it is therefore enough to show one inclusion: that for all $x\in S_{\pi}$ there exists a $\pi'$ with $\theta(\pi')=\pi$ such that $x\in S_{\pi'}$. Let $\pi=\pi_{0}e_{0}\pi_{1}e_{1}\dots \pi_{n}e_{n}\pi_{n+1}$ where for all $i$, $e_{i}=e$, and $\pi_{i}$ is a path (that could be empty if $i=0$ or $i=n+1$). Now chose $x\in S_{\pi}$. Then for all $i=0,\dots, n$, $\phi_{\pi_{0}e_{0}\dots \pi_{i}}(x)\in S_{e_{i}}=S_{e}$, thus $\phi_{\pi_{0}e_{0}\dots \pi_{i}}(x)$ is either in $S_{1}$ or in $S_{2}$. We obtain in this way a sequence $a_{0},\dots, a_{n}$ in $\{1,2\}^{n}$. It is then easy to see that $x\in S_{\pi'}$ where $\pi'=\pi_{0}f_{a_{0}}\pi_{1}f_{a_{1}}\dots \pi_{n}f_{a_{n}}\pi_{n+1}$.
\end{proof}

\begin{lemma}
Let $F,G$ be two graphings, $e\in E^{F}$ and $(F',\theta)$ a refinement of $F$ along $e$. Then $F'\plugmes G$ is a refinement of $F\plugmes G$.
\end{lemma}

\begin{proof}
This is a simple adaptation of the proof of the preceding lemma. Let $D$ be the set of elements such that $\theta^{-1}(e)=D$; we can suppose, modulo considering an almost everywhere equal graphing, that the sets $S_{d}$ ($d\in D$) are pairwise disjoint. To every path $\pi=(f_{i})_{i=0}^{n}$ in $F'\plugmes G$, we associate $\tilde{\theta}(\pi)=\{\theta(f_{i})\}_{i=0}^{n}$. Conversely, a path $\pi=(g_{i})_{i=0}^{n}$ in $F\plugmes G$ defines a countable set of paths:
\begin{equation*}
C_{\pi}=\{(f_{i})_{i=0}^{n}~|~\theta(f_{i})=g_{i}\}
\end{equation*}
We are left with the task of checking that $\tilde{\theta}:\pi\mapsto\tilde{\theta}(\pi)$ is a refinement. For this, we consider two paths $\pi_{1}$ and $\pi_{2}$ such that $\tilde{\theta}(\pi_{1})=\tilde{\theta}(\pi_{2})$. Using the same argument as in the preceding proof, we show that $S_{\pi_{1}}\cap S_{\pi_{2}}$ is of null measure. The verification concerning the weights is straightforward, as is the fact that the functions are almost everywhere equal on the intersection of their domains. The last thing left to show is that $S_{\pi}=\cup_{\pi'\in C_{\pi}} S_{\pi'}$. Here, the argument is again the same as in the preceding proof: an element $x\in S_{\pi}$ is in the domain of one and only one $S_{\pi'}$ for $\pi'\in C_{\pi}$.
\end{proof}

\begin{theorem}\label{invarianceplug}
Let $F,G$ be graphings and $(F',\theta)$ be a refinement of $F$. Then $F'\plugmes G$ is a refinement of $F\plugmes G$.
\end{theorem}

\begin{proof}
If $\pi$ is an alternating path $f_{0}g_{0}f_{1}\dots f_{n}g_{n}$ between $F'$ and $G$, we define $\theta(\pi)=\theta(f_{0})g_{0}\dots \theta(f_{1})\dots g_{n}\theta(f_{n})$. This clearly defines a path, since $S_{f_{i}}\subset S_{\theta(f_{i})}$ (resp. $T_{f_{i}}\subset T_{\theta(f_{i})}$) and $\pi$ is itself a path.

Let us denote by $f_{0},\dots f_{n},\dots$ the edges of $F$. We define the graphings $F^{n}$ as the following restrictions of  $F$: $\{(\omega_{f_{i}}^{F},\phi_{f_{i}}^{F})\}_{i=0}^{n}$. We define the corresponding restrictions of $F'$ as the graphings $(F')^{n}=\{(\omega_{e}^{F'},\phi_{e}^{F'})~|~e\in\theta^{-1}(f_{i})\}_{i=0}^{n}$. By an iterated use of the preceding lemma, we obtain that $((F')^{n}\plugmes G,\theta)$ is a refinement of $F^{n}\plugmes G$ for every integer $n$. It is then easy to see that $(F'\plugmes G,\theta)=(\cup_{n\geqslant 0} (F')^{n}\plugmes G,\theta)$ is a refinement of $\cup_{n\geqslant 0} F^{n}\plugmes G$, i.e.\ of $F\plugmes G$.
\end{proof}

\subsection{Measurement of circuits}

To define a measurement, we first want to define the functions representing the circuits between graphings. This is where things get a little bit more complicated: if $\pi_{1}$ and $\pi_{2}$ are two cycles representing the same circuit (i.e.\ $\pi_{1}$ is a cyclic permutation of $\pi_{2}$) the functions $\phi_{\pi_{1}}$ and $\phi_{\pi_{2}}$ are not equal in general! We will skip this complication by considering \enquote{choices of representative of circuits}, i.e. we fix a particular representative of each equivalence class of cycles modulo cyclic permutations. We will however need to take this non-uniformity later, when defining the notion of circuit-quantifying maps (in the cases -- which are those of interest-- where these maps depend on the functions $\phi_{\pi_{1}},\phi_{\pi_{2}}$ associated to the representatives of circuits).

\begin{definition}
Let $F,G$ be two graphings. A \emph{choice of representatives of circuits} is a set $\repcirc{F,G}$ such that:
\begin{itemize}
\item for all cycle $\rho$ in $\cyclesmes{F,G}$ there exists a unique element $\pi$ in $\repcirc{F,G}$ such that $\bar{\rho}=\bar{\pi}$ ($\bar{\pi}$ denotes the equivalence class of $\pi$ modulo the action of cyclic permutations, see \autoref{defcircuits});
\item if $\pi^{k}\in\repcirc{F,G}$, then $\pi\in\repcirc{F,G}$.
\end{itemize}
%A \emph{choice of representatives of $1$-circuits} is a subset $\repcircone{F,G}$ of $\repcirc{F,G}$ such that all $\rho$ in $\cyclesmes{F,G}$ is a $1$-circuit, i.e.\ is not a proper power of a smaller circuit.
\end{definition}

The second condition implies that a choice of representative of circuits is uniquely determined by the choice of representatives for all $1$-circuits (\autoref{defcircuits}).

%\begin{definition}[Circuits]
%If $F$ and $G$ are graphings and $\repcirc{F,G}$ is a choice of representatives of circuits between $F$ and $G$, we define:
%\begin{eqnarray*}
%\{F,G}&=&\{[\phi_{\pi}]_{i}^{i}~|~\pi\in\repcirc{F,G}\}\\
%%\circmesone{F,G}&=&\{[\phi_{\pi}]_{i}^{i}~|~\pi\in\repcircone{F,G}\}
%\end{eqnarray*}
%\end{definition}

%\begin{remark}
%In the following, if not explicitely mentioned, the terms \emph{circuits} and \emph{$1$-circuits} will be interchangeable. i.e.\ all results and proofs hold regardless of wether we are working with a choice of representatives of circuits or with a choice of representatives of $1$-circuits. We will therefore be showing how to obtain a trefoil property in two different ways: by considering circuits or by considering $1$-circuits.
%\end{remark}

Now we would like to define a measurement of circuits between two graphings $F$ and $G$ in such a way that if $(F',\theta)$ is a refinement of $F$, the measurements $\meas{F,G}$ and $\meas{F',G}$ are equal. Firstly, one should be aware that to define the notion of circuit-quantifying maps one should take into account the fact that if $\pi_{1},\pi_{2}$ are two representatives of a given circuit, the functions $\phi_{\pi_{1}}$ and $\phi_{\pi_{2}}$ are not equal in general.

Secondly, suppose that we obtained such a map $q$ (which does not depend on the choice of representatives), that $\pi$ is an alternating cycle between $F$ and $G$ and that $(F',\theta)$ is a refinement of $F$. We will try to understand what the set of circuits induced by $\pi$ in $F'\bicolmes G$ looks like. We first notice that the cycle $\pi$ corresponds to a family $E_{\pi}$ of alternating cycles between $F'$ and $G$. If for instance $\pi=f_{0}g_{0}\dots f_{n}g_{n}$, one should consider the set of sequences $\{f_{0}'g_{0}\dots f_{n}'g_{n}~|~\forall i, f_{i}'\in\theta^{-1}(f_{i})\}$. However, each of these sequences does not necessarily define a path: it is possible that $S_{g_{i}}\cap T_{f'_{i}}$ (or $S_{f'_{i+1}}\cap T_{g_{i}}$) is of null measure. It is even possible that such a sequence will be a path without being a cycle, and that a cycle of length $l$, once decomposed along the refinement, becomes a cycle of length $m\times l$, where $m$ is an arbitrary integer. \autoref{exempledecompocycles} shows how a cycle of length $2$ can induce either a cycle of length $4$ or a set of two cycles of length $2$ after a refinement. However, a cycle of length $4$ could very well be induced by the cycle $\pi^{2}$ if the latter is an element of $\cyclesmes{F,G}$. The following definitions takes all these remarks into account to introduce the notion of \enquote{refinement-invariant maps}.

\begin{figure}
\centering
\subfigure[A cycle doubling its length]{
\begin{tikzpicture}[x=0.8cm,y=0.8cm]
	\draw[-] (0,0) -- (2,0) node [below,very near start] {$[0,2]$};
	\draw[-] (3,0) -- (5,0) node [below,very near end] {$[3,5]$};
	\node (1) at (1,0) {};
	\node (4) at (4,0) {};
	
	\draw[->,red] (1) .. controls (1,1.5) and (4,1.5) .. (4) node [midway,above] {$x\mapsto 5-x$};
	\draw[->,blue] (4) .. controls (4,-1.5) and (1,-1.5) .. (1) node [midway,below] {$x\mapsto x-3$};
	
	\draw[-] (7.5,0) -- (8.5,0) node [below,very near start] {$[0,1]$};
	\draw[-] (9,0) -- (10,0) node [below,very near start] {$[1,2]$};
	\draw[-] (10.5,0) -- (11.5,0) node [below,very near end] {$[3,4]$};
	\draw[-] (12,0) -- (13,0) node [below,very near end] {$[4,5]$};
	\node (8) at (8,0) {};
	\node (95) at (9.5,0) {};
	\node (11) at (11,0) {};
	\node (125) at (12.5,0) {};
	
	\draw[->,red] (8) .. controls (8,2) and (12.5,2) .. (125) node [midway,above] {$x\mapsto 5-x$};
	\draw[->,red] (95) .. controls (9.5,1) and (11,1) .. (11) node [midway,above] {$x\mapsto 5-x$};
	\draw[->,blue] (125) .. controls (12.5,-1.5) and (9.5,-1.5) .. (95) node [near start,below] {};
	\draw[->,blue] (11) .. controls (11,-1.5) and (8,-1.5) .. (8) node [near end,below] {};
	\node[blue] (a) at (10.25,-1.4) {$x\mapsto x-3$};
	
\end{tikzpicture}
}
\subfigure[A cycle splitting in two cycles]{
\begin{tikzpicture}[x=0.8cm,y=0.8cm]
	\draw[-] (0,0) -- (2,0) node [below,very near start] {$[0,2]$};
	\draw[-] (3,0) -- (5,0) node [below,very near end] {$[3,5]$};
	\node (1) at (1,0) {};
	\node (4) at (4,0) {};
	
	\draw[->,red] (1) .. controls (1,1.5) and (4,1.5) .. (4) node [midway,above] {$x\mapsto x+3$};
	\draw[->,blue] (4) .. controls (4,-1.5) and (1,-1.5) .. (1) node [midway,below] {$x\mapsto x-3$};
	
	\draw[-] (7.5,0) -- (8.5,0) node [below,very near start] {$[0,1]$};
	\draw[-] (9,0) -- (10,0) node [below,very near start] {$[1,2]$};
	\draw[-] (10.5,0) -- (11.5,0) node [below,very near end] {$[3,4]$};
	\draw[-] (12,0) -- (13,0) node [below,very near end] {$[4,5]$};
	\node (8) at (8,0) {};
	\node (95) at (9.5,0) {};
	\node (11) at (11,0) {};
	\node (125) at (12.5,0) {};
	
	\draw[->,red] (8) .. controls (8,1.5) and (11,1.5) .. (11) {};
	\draw[->,red] (95) .. controls (9.5,1.5) and (12.5,1.5) .. (125) {};
	\draw[->,blue] (125) .. controls (12.5,-1.5) and (9.5,-1.5) .. (95) node [near start,below] {};
	\draw[->,blue] (11) .. controls (11,-1.5) and (8,-1.5) .. (8) node [near end,below] {};
	\node[blue] (a) at (10.25,-1.4) {$x\mapsto x-3$};
	\node[red] (a) at (10.25,1.4) {$x\mapsto x+3$};
	
\end{tikzpicture}
}
\caption{Examples of the evolution of a cycle when performing a refinement}\label{exempledecompocycles}
\end{figure}
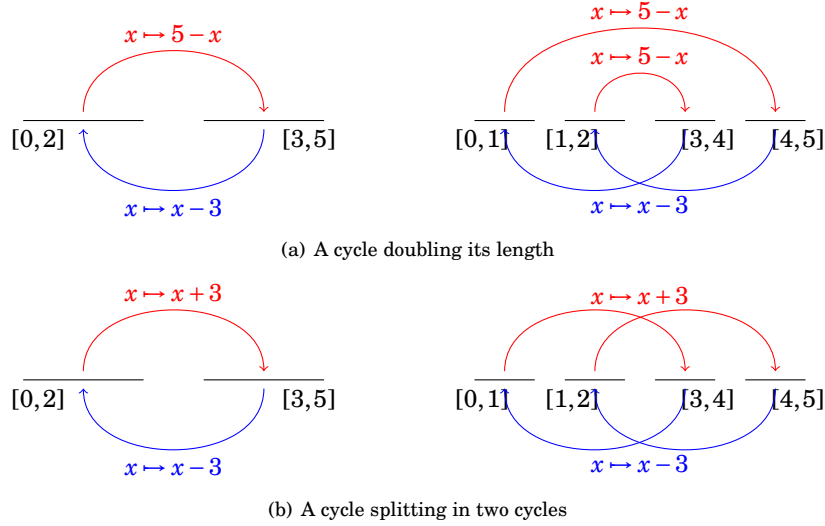

\begin{definition}
Let $\pi$ be a cycle between two graphings $F,G$, and $\repcirc{F,G}$ a choice of representatives of circuits. We define $\pi^{\omega}_{\repcirc{F,G}}$ -- or simply $\pi^{\omega}$ when the context is clear -- as the set $\{\pi^{k}~|~k\in\naturalN\}\cap\repcirc{F,G}$. %If $\bar{\pi}$ is a circuit, we define $\bar{\pi}^{\omega}$ the set $\cup_{\pi\in\bar{\pi}}\pi^{\omega}$.
\end{definition}

\begin{definition}\label{invarianceparraff}
Let $\pi$ be a cycle between two graphings $F,G$, and $\repcirc{F,G}$ a choice of representatives of circuits. Let $(F',\theta)$ be a refinement of $F$. For every choice of representatives of circuits $\repcirc{F',G}$,  we define 
\begin{equation*}
E^{(F',\theta)}_{\pi,\repcirc{F',G}}=\{\rho=f_{0}'g_{0}f'_{1}g_{1}\dots f'_{n}g_{n}\in\repcirc{F',G}~|~\exists k\in\naturalN, \theta(\rho)\in\pi^{\omega}\}
\end{equation*}
Here again, we will usually write this set as $E^{(F',\theta)}_{\pi}$ when the context is clear.
\end{definition}

\begin{definition}\label{invarianceparraff}
A function $q$ from the set\footnote{As in the graph setting, we work modulo a the renaming of edges. As a consequence, if the class of cycles is not a priori a set, the function $q$ will not depend on the name of edges, but only on the weight and the transformation associated to the cycle. We can therefore define $q$ as a function on the set of equivalence classes of cycles modulo renamings of edges.} of cycles into $\realposN\cup\{\infty\}$ is \emph{refinement-invariant} if for all graphings $F,G$, all simple refinement $(F',\theta)$ of $F$, and all choice of representative of circuits $\repcirc{F,G}$, there exists a choice of representative of circuits $\repcirc{F',G}$ such that the following equality holds for every $1$-cycle $\pi$ between $F$ and $G$:
\begin{equation}
\sum_{\rho\in\pi^{\omega}_{\repcirc{F,G}}}q(\rho)=\sum_{\rho\in E^{(F',\theta)}_{\pi,\repcirc{F',G}}} q(\rho)\label{eq_refinv}
\end{equation}
\end{definition}

\begin{definition}[Circuit-Quantifying Maps]
Let $\mathfrak{m}$ denote a microcosm. A map $q$ from the set of $\mathfrak{m}$-cycles\footnote{This is the set of weighted cycles realised by maps in the microcosm $\mathfrak{m}$, i.e.\ the set of couples $(\omega_{\pi},\phi_{\pi})$ of a weight in $\Omega$ and a restriction $\phi_{\pi}$ of an element $m\in\microcosm{m}$ to a measurable set $S$ such that $m(S)=S$.} into $\realposN\cup\{\infty\}$ is a ($\mathfrak{m}$-)circuit-quantifying map if:
\begin{enumerate}[noitemsep,nolistsep]
\item for all representatives $\pi_{1},\pi_{2}$ of a circuit $\pi$, $q(\pi_{1})=q(\pi_{2})$;
\item $q$ is refinement-invariant.
\end{enumerate}
\end{definition}

A circuit-quantifying map should therefore meet quite complex conditions and it is a very natural question to ask wether such maps exist. We will define in the next section a family of circuit-quantifying maps for the set of $1$-circuits (i.e.\ the chosen map sends all circuits which are not $1$-circuits to $0$), answering this question positively.

We now define the measurement associated to a circuit-quantifying map. 
%This shows the following proposition.

\begin{definition}[Measure]\label{measurement}
Let $q$ be a circuit-quantifying map. We define the associated measure of interaction as the function $\meas[q]{\cdot,\cdot}$ which associates, to all couple of graphings $F,G$, the quantity:
\begin{equation*}
\meas[q]{F,G}=\sum_{\pi\in\repcirc{F,G}} q(\pi)
\end{equation*}
%Where $\circmes{F,G}$ depends on a choice of a set of representatives of circuits.
\end{definition}

Notice that although the formal definition depends on a choice of a set of representatives of circuits, the result $\meas{F,G}$ is obviously independent of it. This is a consequence of the first condition in the definition of circuit-quantifying maps. 
Indeed, one can rewrite the expression above as:
$$\meas[q]{F,G}=\sum_{\rho\in\cyclesmes{F,G}} \frac{q(\rho)}{\card{\bar{\rho}}}$$
As a consequence, we will be able to pick, when needed, adequate choices of representative of circuits to ease the proofs.

A natural question is why we are interested in measuring circuits instead of cycles, since the above expression on cycles clearly does not depend on the choice of representatives of circuits. The reason for this can be found in the proof of the trefoil property. Indeed, the above expression does not depend on the choice of representatives, but it depends on the length of cycles. However, the geometric identity shown in \autoref{geometrictrefoil} does not hold for cycles as the length of the cycles forming the corresponding circuits, in e.g. $\circuits{F\plug G,H}$ and $\circuits{G\plug H,F}$, are not equal in general.

%Notice however that this expression seems to depend on the length of cycles, since the cardinality of $\bar{\rho}$ depends on the length of $\rho$. This intuition is wrong, however, as refinement-invariance implies that the measurement of a cycle does not depend on its length, something which is essential to obtain the trefoil property.

\begin{lemma}
Let $F,G$ be two graphings, $e\in E^{F}$, and $F^{(e)}$ a simple refinement along $e$ of $F$. Then:
\begin{equation*}
\meas[q]{F,G}=\meas[q]{F^{(e)},G}
\end{equation*}
\end{lemma}

\begin{proof}
We write $\theta: E^{F^{(e)}}\rightarrow E^{F}$ and $\{f,f'\}=\theta^{-1}(e)$. We will use the notations introduced in \autoref{invarianceparraff}. Moreover, we write $\repcirc{F,G}^{\{e\}}$ the set of cycles in $\repcirc{F,G}$ that go through the edge $e$ at least once. The two sets $\repcirc{F,G}^{\{e\}}$ and $\repcirc{F-\{e\},G}$ form a partition of $\repcirc{F,G}$. Similarly, we write $\repcirc{F^{(e)},G}^{\{f,f'\}}$ the set of cycles in $\repcirc{F^{(e)},G}$ that go through $f$ or $f'$ at least once (notice that it does not necessarily go through both).

We will also denote by $O(\{e\},G)$ the set of $1$-cycles in $\repcirc{F,G}^{\{e\}}$. Then the family $\{\pi^{\omega}\}_{\pi\in O(\{e\},G)}$ is a partition of $\repcirc{F,G}^{\{e\}}$: it is clear that if $\pi,\pi'$ are two distinct elements of $O(\{e\},G)$, $\pi^{\omega}$ and $(\pi')^{\omega}$ are disjoint, and it is equally obvious that $\repcirc{F,G}^{\{e\}}=\cup_{\pi\in O(\{e\},G)} \pi^{\omega}$ since the fact that $\pi^{k}\in\repcirc{F,G}$ implies that $\pi\in\repcirc{F,G}$.

Now, we will consider a choice of representations of circuits $\repcirc{F^{(e)},G}$ such that:
\begin{itemize}[noitemsep,nolistsep]
\item $\repcirc{F^{(e)}-\{f,f'\},G}$ coincides with $\repcirc{F-\{e\},G}$;
\item $\sum_{\pi\in E_{\rho}^{(F^{(e)},\theta)}} q(\pi)=\sum_{\pi\in \rho^{\omega}} q(\pi)$ for all $\rho\in O(\{e\},G)$.
\end{itemize}
The existence of such a choice of representatives is clear from the definition of renfinement-invariance (\autoref{invarianceparraff}).

We can now compute:
\begin{eqnarray*}
\meas[q]{F^{(e)},G}&=&\sum_{\pi\in\repcirc{F^{(e)},G}} q(\pi)\\
&=&\sum_{\pi\in\repcirc{F^{(e)}-\{f,f'\},G}} q(\pi)+\sum_{\pi\in \repcirc{F^{(e)},G}^{\{f,f'\}}} q(\pi)\\
&=&\sum_{\pi\in\repcirc{F-\{e\},G}} q(\pi)+\sum_{\rho \in O(\{e\},G)}\sum_{\pi\in E_{\rho}^{(F^{(e)},\theta)}} q(\pi)\\
&=&\sum_{\pi\in\repcirc{F-\{e\},G}} q(\pi)+\sum_{\rho\in O(\{e\},G)} \sum_{\pi\in \rho^{\omega}} q(\pi)\\
&=&\sum_{\pi\in\repcirc{F-\{e\},G}} q(\pi)+\sum_{\pi\in\repcirc{F,G}^{\{e\}}} q(\pi)\\
&=&\sum_{\pi\in\repcirc{F,G}} q(\pi)
\end{eqnarray*}
Which shows that $\meas[q]{F^{(e)},G}=\meas[q]{F,G}$.
\end{proof}

\begin{theorem}\label{invariancemesure}
Let $F,G$ be graphings and $(F',\theta)$ a refinement of $F$. Then:
\begin{equation*}
\meas[q]{F,G}=\meas[q]{F',G}
\end{equation*}
\end{theorem}

\begin{proof}
The argument is now usual. We first enumerate the edges of $F$, and denote them by $f_{0},\dots, f_{n},\dots$. We then define: 
\begin{eqnarray*}
F^{n}&=&\{(\omega_{f_{i}}^{F},\phi_{f_{i}}^{F})\}_{i=0}^{n}\\
(F')^{n}&=&\{(\omega_{e}^{F'},\phi_{e}^{F'})~|~\theta(e)=f_{i}\}_{i=0}^{n}
\end{eqnarray*}
Then $((F')^{n},\theta)$ is a refinement of $F^{n}$, and an iterated use of the preceding lemma shows that:
\begin{equation*}
\meas[q]{(F')^{n},G}=\meas[q]{F^{n},G}
\end{equation*}
Then:
\begin{eqnarray*}
\meas[q]{F',G}&=&\sum_{\pi\in \repcirc{F',G}} q(\pi)\\
&=&\lim_{n\rightarrow\infty} \sum_{\pi\in\repcirc{(F')^{n},G}} q(\pi)\\
&=&\lim_{n\rightarrow\infty} \meas[q]{(F')^{n},G}\\
&=&\lim_{n\rightarrow\infty} \meas[q]{F^{n},G}\\
&=&\lim_{n\rightarrow\infty} \sum_{\pi\in\repcirc{F^{n},G}} q(\pi)\\
&=&\sum_{\pi\in\repcirc{F,G}} q(\pi)
\end{eqnarray*}
Finally, we showed that $\meas[q]{F',G}=\meas[q]{F,G}$.
\end{proof}

\begin{theorem}[Trefoil Property]\label{thm_trefoilppty}
Let $F,G,H$ be graphings satisfying the condition $\lambda(V^{F}\cap V^{G}\cap V^{H})=0$. Then:
\begin{equation*}
\meas[q]{F,G\plugmes H}+\meas[q]{G,H}=\meas[q]{H\plugmes F,G}+\meas[q]{H,F}
\end{equation*}
\end{theorem}

\begin{proof}
We consider the expression $\meas[q]{F,G\plugmes H}+\meas[q]{G,H}$. We can suppose without loss of generality that $F$ (resp. $G$, resp $H$) is $V^{F}\cap (V^{G}\cup V^{H})$-tough (resp. $V^{G}\cap (V^{F}\cup V^{H})$-tough, resp. $V^{H}\cap (V^{F}\cup V^{G})$-tough). Indeed, if this is not the case the preceding proposition allows us to replace $F,G,H$ by the adequate carvings without changing the measure of interaction.

The end of the proof is now very similar to the proof of the trefoil property for graphs.

Let $\pi$ be an element in $\repcirc{F,G\plugmes H}$. Then $\pi$ is an alternating path between $F$ and $G\plugmes H$, for instance $\pi=f_{0}\rho_{0}f_{1}\dots f_{n}\rho_{n}$. Now, each $\rho_{i}$ is an alternating path between $G$ and $H$. Either each $\rho_{i}$ is an element of $G$ in which case $\pi$ is an alternating path between $F$ and $G$, or at least one of the $\rho_{i}$ contains an edge of $H$. In the first case, picking the right choice of representation of circuits, there is an element $\mu\in\repcirc{F,G}$ corresponding to $\pi$, i.e.\ of equal weight and realised by the same function -- which implies that $q(\pi)=q(\mu)$. Similarly, in this second case it is clear that there is an element of $\repcirc{F\plugmes G, H}$ corresponding to $\pi$ (we use \autoref{lemmesupportchem} to ensure that alternating paths between $F$ and $G$ that appear as part of $\pi$ have a domain -- hence a codomain -- of strictly positive measure). Similarly, an element of $\repcirc{G,H}$ can be shown to correspond to an element of $\repcirc{F\plugmes G,H}$. 
\end{proof}

We now will work with equivalence classes of graphings for the equivalence relation $\sim_{\leqslant}$. \autoref{invariancemesure} and \autoref{invarianceplug} ensure us that the operations of plugging and execution are well defined in this setting. %If the fact that we will work with graphings considered up to the $\sim_{a.e.}$ equivalence (which is contained in $\sim_{\leqslant}$, Proposition \ref{subsetrelation}) is quite natural, it can seem strange to work with graphings up to the equivalence relation $\sim_{\leqslant}$. Our reasons for this should be clarified when we will define second order quantification: working modulo $\sim_{\leqslant}$ allows to define universal quantification as an intersection in a very natural way.
As we showed, each circuit-quantifying map gives rise to a measurement on graphings that satisfies the trefoil property. One question remains unanswered at this point: do such functions exists? The existence of such maps is not clear from their definition, which is quite involved. The next section will give explicit constructions of such maps exists in a very general setting.

\section{Models of Multiplicative-Additive Linear Logic}\label{sec_trefoil}

We have now shown how to define execution between graphings and, given a map $m$, a measurement between graphings. We have shown that the execution is associative and that if $m$ is a \emph{circuit-quantifying map}, that is $m$ satisfies a number of conditions, the trefoil property holds. We have therefore almost finished the proof that for any microcosm $\mathfrak{m}$ and monoid $\Omega$ one can construct a GoI model based on $\Omega$-weighted graphings in $\mathfrak{m}$. The last step is to show the existence of circuit-quantifying maps, i.e.\ exhibit at least one measurement which satisfies the trefoil property. In fact, we will show how to define, given a measure space $X$ satisfying some reasonable properties, one can define a whole family of such circuit-quantifying maps, a family parametrized by the choice of a map $m:\Omega\rightarrow\realposN\cup\{\infty\}$.
%
%\begin{remark}
%We did not consider $\Omega$ endowed with a topology or measure in the definition. In this discrete case, $\Omega$ is thus considered with the discrete topology and induced Borel algebra and the counting measure, which means that measurability is obtained for free. We insist however in writing down that $m$ should be measurable as this allows the consideration of some further structures. For instance, if $\Omega$ were to be chosen a locally compact group, one could consider the usual Haar measure on it, providing models that differ from the ones obtained by looking at $\Omega$ as \enquote{just a group}.
%\end{remark}

To define this family of measurement, we will now restrict our attention to $1$-circuits. Although it should be possible to obtain a same result for circuits without restrictions, the case of $1$-circuits is much more interesting as it was shown to provide a combinatorial expression of the (Fuglede-Kadison) determinant \cite{seiller-goim}.

We first define the notion of \emph{trefoil space}, and prove some basic but essential properties on them. We then define a family of functions and prove that they are circuit-quantifying maps.

%\textcolor{red}{In this section, we show the existence of families of circuit-quantifying maps. The maps we define are built using integral, and we therefore restrict to the case of measure spaces with a $\sigma$-finite radon measure. Moreover, we will need the set of fixed point of a measurable map to be a measurable subset of the space, which will induce a (small) restriction on the type of space considered.}

\subsection{Trefoil Spaces}

We now introduce the notion of trefoil space. One of the conditions for a space $X$ to be a trefoil space is that it should be second-countable. This is justified by the fact that we will need the measurability of the fixed-point set of a measurable transformation on $X$. We will indeed use the following theorems.

\begin{proposition}[Draveck\'{y} \cite{dravecky}]
Let $(Y,\mathcal{T})$ be a measure space. The following statements are equivalent:
\begin{itemize}[noitemsep,nolistsep]
\item $(Y,\mathcal{T})$ has a measurable diagonal;
\item For every measure space $(X,\mathcal{S})$ and every measurable mapping $f:X\rightarrow Y$, the graph of $f$ is measurable.
\end{itemize}
\end{proposition}

\begin{proposition}[Draveck\'{y} \cite{dravecky}]
Let $(Y,\mathcal{G})$ be a topological space and $\mathcal{T}$ the $\sigma$-algebra generated by $\mathcal{G}$. Then $(Y,\mathcal{T})$ has a measurable diagonal if and only if there is a topology $\mathcal{H}\subset\mathcal{G}$ such that $(Y,\mathcal{H})$ is a second-countable $T_{0}$ space.
\end{proposition}

Now, we will therefore ask our spaces to be second-countable\footnote{And satisfying the first axiom of separation $T_{0}$, but this will be strengthened.} in order to obtain the measurability of the fixed point sets of measurable transformations. Moreover, our measurements will be defined using integrals, and we thus need a space in which one can define a reasonable notion of integral. In particular, we will ask our space to be Hausdorff\footnote{In fact, this restriction could be weakened, as integration with respect to Radon measures can be defined in non-Hausdorff spaces. However it is not clear that the weakening this condition would be of interest. We therefore consider the case of Hausdorff spaces, keeping in mind that this is not an essential condition.}, and endowed with its Borel $\sigma$-algebra and a $\sigma$-additive Radon measure.

\begin{definition}
Let $(X,\mathcal{T})$ be a second-countable Hausdorff space, and $(X,\mathcal{B},\mu)$ be a measure space where $\mathcal{B}$ is the Borel $\sigma$-algebra and $\mu$ a $\sigma$-additive Radon measure on $(X,\mathcal{B})$. Such a measure space will be referred to as a \emph{trefoil space}.
\end{definition}

\begin{proposition}\label{fixedpointset}
Let $X$ be a trefoil space. For all measurable map $\phi: X\rightarrow X$, the fixed point set $\mathcal{F}(\phi)=\{x\in X~|~\phi(x)=x\}$ is measurable.
\end{proposition}

\begin{proof}
The point is that the diagonal $\Delta=\{(x,x)~|~x\in X\}\subset X\times X$ is a measurable set for the product $\sigma$-algebra. This is true because of second countability. Then, we have that $\mathcal{F}(\phi)=F^{-1}(\Delta)$ for $F(x)=(x,f(x))$ measurable from $X$ to $X\times X$. Hence $\mathcal{F}(\phi)$ is measurable.
\end{proof}

%\begin{remark}
%The conditions stated in this theorem cannot be weakened.
%\end{remark}

\begin{corollary}
Let $X$ be a trefoil space and $\phi: X\rightarrow X$ be a measurable map. Then the following map is measurable:
\begin{equation*}
\rho_{\phi}:\left\{\begin{array}{rcll}
		 X&\rightarrow& \naturalN^{\ast}\cup\{\infty\}&\\
		 x&\mapsto&\inf\{n\in \naturalN^{\ast}~|~\phi^{n}(x)=x\}&\text{ if }\{n\in \naturalN^{\ast}~|~\phi^{n}(x)=x\}\neq\emptyset\\
		 x&\mapsto&\infty&\text{ otherwise}
		 \end{array}\right.
\end{equation*}
\end{corollary}

\begin{proof}
We define $X_{i}=\rho_{\phi}^{-1}(i)$ for all integer $i\in\naturalN^{\ast}$. Then it is clear that $X_{i}$ is equal to the fixed point set $\mathcal{F}(\phi^{i})$ of $\phi^{i}$. Applying \autoref{fixedpointset}, we deduce that $X_{i}$ is measurable. Finally, the set $X_{\infty}=X-\cup_{i\in\naturalN^{\ast}} X_{i}$ is also measurable.
\end{proof}

%\begin{lemma}
%Let $X$ be a trefoil spacen and $\phi: X\rightarrow X$ be a measurable map. The following map is measurable:
%\begin{equation*}
%\lambda_{\phi}: x\mapsto \int_{\rho^{-1}_{\phi}(\rho_{\phi}(x))} d\lambda(y)%= \int_{\rho_{\phi}(x)} d\lambda_{\ast}\rho_{\phi}(y)
%\end{equation*}
%\end{lemma}
%
%\begin{proof}
%The proof is straightforward. The map $\lambda_{\phi}$ is constant on the measurable sets $\rho_{\phi}^{-1}(i)$ ($i\in\naturalN$), hence measurable.
%\end{proof}

\subsection{Circuit-Quantifying Maps for Measure-Inflating Transformations}

We have chosen to explain the construction of circuit-quantifying maps on the microcosm of measure-inflating transformations first. Indeed, this particular case allows for a simpler definition of the maps which should be more intuitive for the reader. We will then built on the results of this section to define circuit-quantifying maps in the general setting.

\begin{definition}
Let $\phi: X\rightarrow X$ be a non-singular transformation. We define the measurable set:
\begin{equation*}
\{\phi\}=\bigcap_{n\in \naturalN} \phi^{n}(X)\cap\phi^{-n}(X)
\end{equation*}
\end{definition}

\begin{definition}
Let $\pi$ be a cycle in the weighted graphing $F$. Then the map $\phi_{\pi}$ restricted to $X=\{\phi_{\pi}\}$ is a non-singular transformation $X\rightarrow X$. We can then define the map $\rho_{\phi_{\pi}}$ on $X$. We define the \emph{support} $\supp{\pi}$ of $\pi$ as the set $\rho_{\phi_{\pi}}^{-1}(\naturalN^{\ast})$.
\end{definition}

\begin{remark}
In the author's PhD, a similar work was presented, only restricted to the particular case of the microcosm of measure-preserving maps on the real line. We showed in this case the existence of a family of circuit-quantifying maps. Indeed, for any map $m:\Omega\rightarrow\realN\cup\{\infty\}$, we defined $q_{m}$ as the function:
\begin{equation*}
q_{m}:\pi\mapsto\int_{\supp{\pi}} \frac{m(\omega(\pi)^{\rho_{\phi_{\pi}}(x)})}{\rho_{\phi_{\pi}}(x)}d\lambda(x)
\end{equation*} 
This function was a circuit-quantifying map for the above mentioned microcosm. As it turns out, this approach can be generalised to the microcosm of measure-preserving maps on any trefoil space by using the very same formula.
%\begin{equation}\label{eq1}
%q_{m}:\pi\mapsto\int_{\supp{\pi}} \frac{m(\omega(\pi)^{\rho_{\phi_{\pi}}(x)})}{\rho_{\phi_{\pi}}(x)}d\lambda(x)
%\end{equation}
\end{remark}

We now want to extend these circuit-quantifying maps to the microcosm of measure-inflating maps\footnote{We use this terminology for maps that transport the measure $\mu$ onto a scalar multiple of $\mu$.} on any trefoil space (i.e.\ transporting $\mu$ to a scalar multiple of $\mu$). One easy way to do so is by considering an extension using push-forward measures. Recall that if $\mu$ is a measure on $X$, and $f: X\rightarrow Y$ a measurable map, then the push-forward measure $f_{\ast}\mu$ is defined by $f_{\ast}\mu(A)=\mu(f^{-1}(A))$ and satisfies, for all $g$ measurable such that $g\circ f$ is integrable (this is equivalent to saying that $g$ is $f_{\ast}\mu$ integrable):
$$\int_{Y} g(y)df_{\ast}\mu(y)=\int_{X}g(f(x))d\mu(x)$$
Notice that the notation for pushforward measure reverses the composition order, i.e.\ $(f\circ g)_{\ast}\mu=g_{\ast}f_{\ast}\mu$.
%
%\begin{equation}\label{eq2}
%q_{m}:\pi\mapsto\int_{\supp{\pi}} \frac{m(\omega(\pi)^{\rho_{\phi_{\pi}}(x)})}{\lambda_{\phi_{\pi}}(x)\rho_{\phi_{\pi}}(x)}d\lambda(x)
%\end{equation}
%\item to the microcosm of all non-singular transformations an any trefoil space by using the push-forward measure: if $\mu$ is a measure on $X$, and $f: X\rightarrow Y$ a measurable map, then the push-forward measure $\mu_{\ast}f$ is defined by $\mu_{\ast}f(A)=\mu(f^{-1}(A))$ and satisfies, for all $g$ measurable such that $g\circ f$ is integrable (this is equivalent to saying that $g$ is $\mu_{\ast}f$ integrable):
%$$\int_{Y} g(y)d\mu_{\ast}f(y)=\int_{X}g(f(x))d\mu(x)$$
%We can then define the measurement by:
%\begin{equation}\label{eq3}
%q_{m}:\pi\mapsto\int_{\supp{\pi}} \frac{m(\omega(\pi)^{\rho_{\phi_{\pi}}(x)})}{\rho_{\phi_{\pi}}(x)}d\lambda_{\ast}\phi_{\pi}(x)
%\end{equation}
%
%The collapse of notations will be justified by Proposition \ref{collapse} which states that the latter general expression of $q_{m}$ simplifies to the expressions found in equations \ref{eq1} and \ref{eq2} when restricting to the mentioned microcosms.
%\end{remark}

%We state the following easy lemma in order to justify that the following definition makes sense.
%
%\begin{lemma}
%Let $f:X\rightarrow \naturalN$ and $g:X\rightarrow \realN$ be measurable maps. Then the following map is measurable:
%$$h:x\mapsto (g(x))^{f(x)}$$
%\end{lemma}
%
%\begin{proof}
%The restriction of $h(x)$ to the measurable subsets $f^{-1}(k)\subset X$ for each $k\in\naturalN$ is a measurable map. Hence $h$ is measurable.
%\end{proof}

\begin{definition}
Let $(X,\mathcal{B},\mu)$ be a trefoil space, and $m:\Omega\rightarrow \realposN\cup\{\infty\}$ be a map. We define the following function on the set of $1$-cycles:
\begin{equation*}
q_{m}:\pi=e_{0}\dots e_{n}\mapsto\frac{1}{n+1}\sum_{i=0}^{n}\int_{\supp{\pi}} \frac{m(\omega(\pi)^{\rho_{\phi_{\pi}}(x)})}{\rho_{\phi_{\pi}}(x)}d(\phi_{n}\circ\phi_{n-1}\circ\dots\circ\phi_{i})_{\ast}\lambda(x)
\end{equation*}
This extends to a function $q_{m}$ on cycles by letting $q_{m}(\rho)=0$ when $\rho$ is not a $1$-cycle.
\end{definition}

To show that this defines a circuit-quantifying map, we will need an additional hypothesis on the map $m$, that of \emph{traciality}.

\begin{definition}
Let $m$ be a map of domain a monoid $\Omega$. We say that $m$ is tracial when $m(ab)=m(ba)$ for all $a,b\in\Omega$.
\end{definition}

\begin{remark}
As already noted in earlier work \cite{seiller-goia}, traciality is necessary to define a satisfying measurement when the monoid of weights $\Omega$ is not commutative. Notice however that if $\Omega$ is commutative, then any map of domain $\Omega$ is tracial.
\end{remark}

%\begin{proposition}\label{collapse}
%Let $(X,\mathcal{B},\mu)$ be a trefoil space, $\mathfrak{m}$ the microcosm of measure-inflating/expanding maps, and $m:\Omega\rightarrow \realposN$ be a measurable map. The map $q_{m}$ can be expressed as:
%\begin{equation*}
%q_{m}:\pi\mapsto\int_{\supp{\pi}} \frac{m(\omega(\pi)^{\rho_{\phi_{\pi}}(x)})}{\lambda(\supp{\pi})\rho_{\phi_{\pi}}(x)}d\lambda(x)
%\end{equation*} 
%\end{proposition}

Now, the following result shows that the defined map generalises the circuit-quantifying maps introduced in the author's PhD \cite{seiller-phd}.

\begin{proposition}
Let $(X,\mathcal{B},\mu)$ be a trefoil space, $\mathfrak{m}$ the microcosm of measure-preserving maps, and $m:\Omega\rightarrow \realposN\cup\{\infty\}$ be a map. The map $q_{m}$ can be expressed as:
\begin{equation*}
q_{m}:\pi\mapsto\int_{\supp{\pi}} \frac{m(\omega(\pi)^{\rho_{\phi_{\pi}}(x)})}{\rho_{\phi_{\pi}}(x)}d\lambda(x)
\end{equation*} 
\end{proposition}

\begin{proof}
Let $\pi=e_{0}e_{1}\dots e_{n}$ be a cycle, $\supp{\pi}$ its support. For all $i\in\naturalN^{\ast}$, we write $(\supp{\pi})_{i}=\rho_{\pi}^{-1}(i)$.
If all $\phi_{k}$ are measure-preserving maps, then we have, for all integer $j$, the following equality:
\begin{equation*}
\int_{\supp{\pi}_{j}} \frac{m(\omega(\pi)^{j})}{j}d(\phi_{n}\circ\phi_{n-1}\circ\dots\circ\phi_{i})_{\ast}\lambda(x)=\int_{\supp{\pi}_{j}} \frac{m(\omega(\pi)^{j})}{j}d\lambda(x)
\end{equation*}
We can then compute:
\begin{eqnarray*}
q_{m}(\pi)&=&\frac{1}{\lg(\pi)}\sum_{i=0}^{\lg(\pi)-1}\int_{\supp{\pi}} \frac{m(\omega(\pi)^{\rho_{\phi_{\pi}}(x)})}{\rho_{\phi_{\pi}}(x)}d(\phi_{n}\circ\phi_{n-1}\circ\dots\circ\phi_{i})_{\ast}\lambda(x)\\
&=&\frac{1}{\lg(\pi)}\sum_{i=0}^{\lg(\pi)-1}\sum_{j\in\naturalN}\int_{\supp{\pi}_{j}} \frac{m(\omega(\pi)^{j})}{j}d(\phi_{n}\circ\phi_{n-1}\circ\dots\circ\phi_{i})_{\ast}\lambda(x)\\
&=&\sum_{j\in\naturalN^{\ast}}\sum_{i=0}^{\lg(\pi)-1}\int_{\supp{\pi}_{j}} \frac{1}{\lg(\pi)}\frac{m(\omega(\pi)^{j})}{j}d(\phi_{n}\circ\phi_{n-1}\circ\dots\circ\phi_{i})_{\ast}\lambda(x)\\
&=&\sum_{j\in\naturalN^{\ast}}\sum_{i=0}^{\lg(\pi)-1}\int_{\supp{\pi}_{j}} \frac{1}{\lg(\pi)}\frac{m(\omega(\pi)^{j})}{j}d\lambda(x)\\
&=&\sum_{j\in\naturalN^{\ast}}\int_{\supp{\pi}_{j}}\frac{m(\omega(\pi)^{j})}{j}d(\phi_{n}\circ\phi_{n-1}\circ\dots\circ\phi_{i})_{\ast}\lambda(x)\\
&=&\int_{\supp{\pi}} \frac{m(\omega(\pi)^{\rho_{\phi_{\pi}}(x)})}{\rho_{\phi_{\pi}}(x)}d\lambda(x)
\end{eqnarray*}
\end{proof}

\begin{lemma}\label{lemma1}
For all tracial map $m$, the function $q_{m}$ has a constant value on the equivalence classes of cycles modulo the action of cyclic permutations.
\end{lemma}

\begin{proof}
Let $\pi=e_{0}e_{1}\dots e_{n}$ be a cycle, $\supp{\pi}$ its support. For all $i\in\naturalN^{\ast}$, we write $(\supp{\pi})_{i}=\rho_{\pi}^{-1}(i)$. Consider now $\pi^{(1)}=e_{1}e_{2}\dots e_{n}e_{0}$, and $\supp{\pi^{(1)}}$ its support. We define $(\supp{\pi^{(1)}})_{i}=\rho_{\pi^{(1)}}^{-1}(i)$. We will first show that $(\supp{\pi^{(1)}})_{i}=\phi_{e_{0}}((\supp{\pi})_{i})$ for all integer $i$. 

Let us now pick $x\in(\supp{\pi^{(1)}})_{i}$, which means that $x\in \supp{\pi^{(1)}}$ and $\phi_{\pi^{(1)}}^{i}(x)=x$. Since $\phi_{\pi^{(1)}}(x)=\phi_{e_{0}}(\phi_{e_{1}\dots e_{n}}(x))$, we have $x=\phi_{e_{0}}(\phi_{e_{1}\dots e_{n}}\phi_{\pi^{(1)}}^{i-1}(x))$. We now define $y=\phi_{e_{1}\dots e_{n}}\phi_{\pi^{(1)}}^{i-1}(x))$ and we will show that $y\in (\supp{\pi})_{i}$. Since $\phi_{e_{0}}(y)\in \supp{\pi^{(1)}}$, we have $\phi_{e_{0}}\in S_{e_{1}\dots e_{n}}$, and therefore $y\in S_{\pi}$. Moreover, 
\begin{eqnarray*}
\phi_{\pi}^{k}(y)&=&\phi_{\pi}^{i}(\phi_{e_{1}\dots e_{n}}\phi_{\pi^{(1)}}^{i-1}(x))\\
&=&\phi_{\pi}(\phi_{\pi}^{i-2}(\phi_{e_{1}\dots e_{n}}\phi_{\pi^{(1)}}^{i-1}(x)))\\
&=&\phi_{e_{1}\dots e_{n}}(\phi_{e_{0}}(\phi_{\pi}^{i-1}(\phi_{e_{1}\dots e_{n}}(\phi_{\pi^{(1)}}^{i-1}(x)))))\\
&=&\phi_{e_{1}\dots e_{n}}(\phi_{\pi^{(1)}}^{i}(\phi_{\pi^{(1)}}^{i-1}(x)))\\
&=&\phi_{e_{1}\dots e_{n}}(\phi_{\pi^{(1)}}^{i-1}(\phi_{\pi^{(1)}}^{i}(x)))\\
&=&\phi_{e_{1}\dots e_{n}}(\phi_{\pi^{(1)}}^{i-1}(x))\\
&=&y
\end{eqnarray*}
Thus $y$ is an element in $\supp{\pi^{(1)}}$, and more precisely an element in $(\supp{\pi^{(1)}})_{i}$. We therefore showed that $(\supp{\pi^{(1)}})_{i}\subset\phi_{e_{0}}((\supp{\pi})_{i})$.

To show the converse inclusion, we take $x=\phi_{e_{0}}(y)$ with $y\in (\supp{\pi})_{i}$. Then $y\in S_{\pi^{k}}$ and therefore $y\in S_{e_{0}e_{1}\dots e_{n}e_{0}}$. Finally $\phi_{e_{0}}(y)\in S_{\pi^{(1)}}$. Moreover, we have:
\begin{eqnarray*}
\phi_{\pi^{(1)}}^{k}(x)&=&\phi_{\pi^{(1)}}^{k}(\phi_{e_{0}}(y))\\
&=&\phi_{e_{0}}(\phi_{\pi}^{k}(y))\\
&=&\phi_{e_{0}}(y)\\
&=&x
\end{eqnarray*}
As a consequence, $x$ is an element in $(\supp{\pi})_{i}$, which shows the converse inclusion.

More generally, if $\pi^{(k)}$ denotes the cycle $e_{k}e_{k+1}\dots e_{n}e_{0}\dots e_{k-1}$, we have 
$$\phi_{e_{0}\dots e_{k}}(\supp{\pi}_{i})=\supp{\pi^{(k)}}_{i}$$
A similar argument shows that $\phi_{e_{n}}(\supp{\pi^{(n)}}_{i})=\supp{\pi}_{i}$.

We then compute\footnote{When $i=n$ in the last three lines, the expression $d(\phi_{e_{n-1}}\circ\dots\circ \phi_{e_{i}})_{\ast}\lambda$ stands for $d\lambda$.}:
\begin{eqnarray*}
q_{m}(\pi)&=&\frac{1}{\lg(\pi)}\sum_{i=0}^{\lg(\pi)-1}\int_{\supp{\pi}} \frac{m(\omega(\pi)^{\rho_{\phi_{\pi}}(x)})}{\rho_{\phi_{\pi}}(x)}d(\phi_{e_{n}}\circ\dots\circ\phi_{e_{i}})_{\ast}\lambda(x)\\
%&=&\frac{1}{\lg(\pi)}\sum_{i=0}^{\lg(\pi)-1}\int_{\supp{\pi}} \frac{m(\omega(\pi)^{\rho_{\phi_{\pi}}(x)})}{\rho_{\phi_{\pi}}(x)}d\lambda_{\ast}\phi_{n}\circ\phi_{n-1}\circ\dots\circ\phi_{i}(x)\\
%%
%&=&\frac{1}{\lg(\pi)}\sum_{i=0}^{\lg(\pi)-1}\int_{\phi_{n}(\supp{\pi})} \frac{m(\omega(\pi)^{\rho_{\phi_{\pi}}(x)})}{\rho_{\phi_{\pi}}(x)}d\lambda_{\ast}\phi_{n-1}\circ\dots\circ\phi_{i}(x)\\
&=&\frac{1}{\lg(\pi)}\sum_{i=0}^{\lg(\pi)-1}\sum_{j\in\naturalN^{\ast}}\int_{\supp{\pi}_{j}} \frac{m(\omega(\pi)^{j})}{j}d(\phi_{e_{n}}\circ\dots\circ\phi_{e_{i}})_{\ast}\lambda(x)\\
&=&\frac{1}{\lg(\pi)}\sum_{i=0}^{\lg(\pi)-1}\sum_{j\in\naturalN^{\ast}}\int_{\phi_{e_{n}}^{-1}(\supp{\pi}_{j})} \frac{m(\omega(\pi)^{j})}{j}d(\phi_{e_{n-1}}\circ\dots\circ\phi_{e_{i}})_{\ast}\lambda(x)\\
&=&\frac{1}{\lg(\pi)}\sum_{i=0}^{\lg(\pi)-1}\sum_{j\in\naturalN^{\ast}}\int_{\supp{\pi^{(n)}}_{j}} \frac{m(\omega(\pi)^{j})}{j}d(\phi_{e_{n-1}}\circ\dots\circ\phi_{e_{i}})_{\ast}\lambda(x)\\
&=&\frac{1}{\lg(\pi)}\sum_{i=0}^{\lg(\pi)-1}\int_{\supp{\pi^{(n)}}}  \frac{m(\omega(\pi)^{\rho_{\phi_{\pi}}(x)})}{\rho_{\phi_{\pi}}(x)}d(\phi_{e_{n-1}}\circ\dots\circ\phi_{e_{i}})_{\ast}\lambda(x)
\end{eqnarray*}

One can now notice that $(\phi_{n}\circ\phi_{n-1}\circ\dots\circ\phi_{0})_{\ast}\lambda=\lambda$. Indeed, $\phi_{\pi}$ is measure-inflating as a composition of measure-inflating maps. But since it has its domain equal to its codomain, it is necessarily measure-preserving. The same argument applies to the maps $\phi_{\pi^{(i)}}$. This implies in particular that for all measurable map $f:X\rightarrow\realN$:
$$\int_{\supp{\pi^{(n)}}} f(x)d\lambda(x)=\int_{\supp{\pi^{(n)}}} f(x)d(\phi_{e_{n-1}}\circ\phi_{e_{n-2}}\circ\dots\circ\phi_{e_{0}}\circ\phi_{e_{n}})_{\ast}\lambda(x)$$
Using this equality and the traciality of $m$ -- which ensures that $m(\omega(\pi))=m(\omega(\pi^{(n)}))$, we compute:

\begin{eqnarray*}
q_{m}(\pi)&=&\frac{1}{\lg(\pi)}\int_{\supp{\pi^{(n)}}}  \frac{m(\omega(\pi)^{\rho_{\phi_{\pi}}(x)})}{\rho_{\phi_{\pi}}(x)}d\lambda(x)\\
&&~~~~~~+\frac{1}{\lg(\pi)}\sum_{i=0}^{\lg(\pi)-2}\int_{\supp{\pi^{(n)}}}  \frac{m(\omega(\pi)^{\rho_{\phi_{\pi}}(x)})}{\rho_{\phi_{\pi}}(x)}d(\phi_{e_{n-1}}\circ\dots\circ\phi_{e_{i}})_{\ast}\lambda(x)\\
&=&\frac{1}{\lg(\pi)}\int_{\supp{\pi^{(n)}}}  \frac{m(\omega(\pi)^{\rho_{\phi_{\pi}}(x)})}{\rho_{\phi_{\pi}}(x)}d(\phi_{e_{n-1}}\circ\phi_{e_{n-2}}\circ\dots\circ\phi_{e_{0}}\circ\phi_{e_{n}})_{\ast}\lambda(x)\\
&&~~~~~~+\frac{1}{\lg(\pi)}\sum_{i=0}^{\lg(\pi)-2}\int_{\supp{\pi^{(n)}}}  \frac{m(\omega(\pi)^{\rho_{\phi_{\pi}}(x)})}{\rho_{\phi_{\pi}}(x)}d(\phi_{e_{n-1}}\circ\dots\circ\phi_{e_{i}})_{\ast}\lambda(x)\\
&=&q_{m}(\pi^{(n)})
\end{eqnarray*}
Thus, by induction, $q_{m}(\pi)=q_{m}(\pi^{(j)})$ for all integer $j$. We can now conclude that $q_{m}$ takes a constant value over the equivalence classes of cycles modulo the action of cyclic permutations.
\end{proof}

\begin{lemma}\label{lemma2}
For all tracial map $m$, the function $q_{m}$ is refinement-invariant.
\end{lemma}

\begin{proof}
Let $F,G$ be weighted graphings, and $F^{(e)}$ a simple refinement of $F$ along $e\in E^{F}$. We will denote by $f,f'$ the two elements of $F^{(e)}$ which are the decompositions of $e$. Up to almost everywhere equality, one can suppose that $S_{f}\cap S_{f'}=\emptyset$. Let us now chose $\pi$ a representative of a $1$-circuit $\bar{\pi}$. We suppose that $\pi$ contains occurrences of $e$, and write $\pi=\rho_{0}e_{i_{0}}\rho_{1}e_{i_{1}}\dots e_{i_{n-1}}\rho_{n}$ where for all $j$, $e_{i_{j}}=e$ and  $\rho_{j}$ is a path (where the paths $\rho_{0}$ and $\rho_{n}$ may be empty). We denote by $E_{\pi}$ the set of $1$-cycles $\mu=\rho_{0}\epsilon_{i_{0}^{0}}\rho_{1}\epsilon_{i_{1}^{0}}\dots \epsilon_{i_{n-1}^{0}}\rho_{n}\rho_{0}\epsilon_{i_{0}^{1}}\rho_{1}\dots \epsilon_{i_{n-1}^{1}}\rho_{n}\dots \rho_{0}\epsilon_{i_{0}^{k}}\dots \epsilon_{i_{n-1}^{k}}\rho_{n}$ where $k\in\naturalN$ -- which we will denote by $\lg^{\pi}(\mu)$, and where for all values of $l,m$, $\epsilon_{i_{l}^{m}}$ is either equal to $f$ or equal to $f'$. We will write $\bar{E}_{\pi}$ the subset of $E_{\pi}$ consisting of those paths that are elements of $\repcirc{F^{(e)},G}$, i.e.\ $\bar{E}_{\pi}$ is equal to the set of $1$-cycles in $E^{(F',\theta)}_{\pi}$.

%Notice that the set $E_{\pi}$ is the subset of $E^{(F,\theta)}_{\pi}$ (\autoref{invarianceparraff}) consisting of those $1$-cycles $\rho$ such that $\theta(\rho)\in\pi^{\omega}$. As a consequence the union, over all $\pi\in\bar{\pi}$, of the sets $E_{pi}$, which we denote $E_{\bar{pi}}$, is the set of all $1$-cycles in $E^{(F,\theta)}_{\pi}$, i.e.\ $E_{\bar{\pi}}$ is equal to the homonimous set in \autoref{lemmainvraff1circ}.

%We will denote by $\bar{E}_{\pi}$ the set of $1$-circuits in $E_{\pi}$ and by $\bar{E}_{\pi}$ the set of $1$-circuits in $E_{\bar{\pi}}$. 

We will now show that $q_{m}(\pi)=\sum_{\bar{\mu}\in\bar{E}_{\pi}} q_{m}(\bar{\mu})$. Since $q_{m}(\rho)=0$ as soon as $\rho$ is not a $1$-cycle, we can rewrite this equality as: $\sum_{\rho\in\pi^{\omega}} q_{m}(\pi)=\sum_{\bar{\mu}\in\bar{E}_{\pi}} q_{m}(\bar{\mu})$. We now remark that the set $\bar{E}_{\pi}$ appearing on the right-hand side is the subset of $1$-cycles in $E^{(F',\theta)}_{\pi}$, and use once again the fact that $q_{m}(\rho)$ is null as soon as $\rho$ is not a $1$-cycle. Consequently, showing that $q_{m}(\pi)=\sum_{\bar{\mu}\in\bar{E}_{\pi}} q_{m}(\bar{\mu})$ turns out to be equivalent to showing that:
$$ \sum_{\rho\in\pi^{\omega}} q_{m}(\pi)=\sum_{\rho\in E^{(F',\theta)}_{\pi}} q_{m}(\rho)$$
i.e.\ showing that $q_{m}(\pi)=\sum_{\bar{\mu}\in\bar{E}_{\pi}} q_{m}(\bar{\mu})$ amounts to showing that $q_{m}$ is refinement-invariant.

%Since we are working with $1$-circuits, for all representative $\pi$ of $\bar{\pi}$ the set $\pi^{\omega}$ can be reduced to $\{\pi\}$, and to show that $q$ is refinement-invariant, we only need to check (\autoref{lemmainvraff1circ}) that the following equality holds:
%$$ q_{m}(\pi)=\sum_{\mu\in E_{\pi}} \frac{q_{m}(\mu)}{\card{\bar{\mu}}} $$ 
%Since for all representatives $\mu,\mu'$ of $\bar{\mu}\in\bar{E}_{\pi}$, we have $q_m(\mu)=q_m(\mu')$, this becomes:
%$$ q_{m}(\pi)=\sum_{\bar{\mu}\in \bar{E}_{\pi}} q_{m}(\bar{\mu}) $$

Let us pick $x\in \supp{\pi}$. Then $x\in(\supp{\pi})_{k}$ for a given value $k$ in $\naturalN^{\ast}$, i.e.\ $\phi_{\pi}^{k}(x)=x$. Since $S_{e}=S_{f}\cup S_{f'}$, we have, for each occurrence $e_{i_{p}}$ of $e$ and each integer $l$: 
\begin{equation*}
\phi^{k}_{\pi}=\phi_{\pi}^{l}\circ\phi_{\rho_{p+1}e_{i_{p+1}}\dots e_{i_{n}}\rho_{n}}\circ\phi_{e_{i_{p}}}\circ\phi_{\rho_{0}e_{i_{0}}\rho_{1}\dots e_{i_{p-1}}\rho_{j}}\circ\phi_{\pi}^{k-l-1}
\end{equation*}
Then $\phi_{\rho_{0}e\dots\rho_{j}}\circ\phi_{\pi}^{k-l-1}(x)$ is either an element in $S_{f}$ or an element in $S_{f'}$. For each occurrence $e_{i}$ of $e$, we will write $d_{i_{p,l}}=f$ or $d_{i_{p,l}}=f'$ according to whether $\phi_{\rho_{0}e\dots\rho_{j}}\circ\phi_{\pi}^{k-l-1}(x)$ is an element in $S_{f}$ or an element in $S_{f'}$. We then obtain, for all integer $0\leqslant l\leqslant k$, paths $\nu_{l}=\rho_{0}d_{i_{0,l}}\rho_{1}d_{i_{1,l}}\dots d_{i_{n-1},l}\rho_{n}$. By concatenation, we can define a cycle $\nu=\nu_{0}\nu_{1}\dots\nu_{k}$. This cycle is a $d$-cycle for a given integer $d$, i.e.\ $\nu=\tilde{\pi}^{d}$ where $\tilde{\pi}$ is a $1$-cycle in $E_{\pi}$. It is clear from the definition of $\tilde{\pi}$ that $x\in \supp{\tilde{\pi}}$ and that, for all $1$-cycle $\mu$ in $E_{\pi}$, $x\not\in \supp{\mu}$ when $\mu\neq\tilde{\pi}$.

Moreover, it is clear that if $x\in\supp{\mu}$ for a given $1$-cycle $\mu\in E_{\pi}$, then one necessarily has $x\in\supp{\pi}$. We deduce from this that the family $(\supp{\mu})_{\mu\in E_{\pi}}$ is a partition of the set $\supp{\pi}$. Notice that $\omega_{\mu}=\omega_{\pi}^{\lg^{\pi}(\mu)}$. Moreover, for all $x\in \supp{\mu}$, one has $\rho_{\phi_{\mu}}(x)\times \lg^{\pi}(\mu)=\rho_{\phi_{\pi}}(x)$, and therefore $\omega_{\pi}^{\rho_{\phi_{\pi}}(x)}=\omega_{\mu}^{\rho_{\phi_{\mu}}(x)}$.

We now notice that if $\mu=\mu_{1}\dots\mu_{\lg^{\pi}(\mu)}\in E_{\pi}$ (here the $\mu_{i}$ are \enquote{refinements} of $\pi$), and if $\sigma$ is the cyclic permutations over $\{1,\dots,\lg^{\pi}(\mu)\}$ such that $\sigma(i)=i+1$, then the $1$-cycles 
$$\mu_{\sigma^{k}}=\mu_{\sigma^{k}(1)}\mu_{\sigma^{k}(2)}\dots\mu_{\sigma^{k}(\lg^{\pi}(\mu))}$$
for $0\leqslant k\leqslant \lg^{\pi}(\mu)-1$ are pairwise disjoint elements in $E_{\pi}$. Indeed, these are $1$-cycles since $\mu$ is a $1$-cycle, and they are pairwise disjoint because if $\mu_{\sigma^{k}}=\mu_{\sigma_{k'}}$ (supposing that $k>k'$), we can show that $\mu_{\sigma(k-k')}=\mu$ and that $k-k'$ divides $\lg^{\pi}(\mu)$, which contradicts the fact that $\mu$ is a $1$-cycle.

We can now deduce that:
\begin{eqnarray*}
\lefteqn{\int_{\supp{\pi}} \frac{m(\omega(\pi)^{\rho_{\phi_{\pi}}(x)})}{\rho_{\phi_{\pi}}(x)} d(\phi_{e_{n}}\circ\dots\circ\phi_{e_{i}})_{\ast}\lambda(x)}\\
&=&\sum_{\mu\in E_{\pi}}\int_{\supp{\mu}} \frac{m(\omega(\pi)^{\rho_{\phi_{\pi}}(x)})}{\rho_{\phi_{\pi}}(x)} d(\phi_{e_{n}}\circ\dots\circ\phi_{e_{i}})_{\ast}\lambda(x)\\
&=&\sum_{\mu\in E_{\pi}}\int_{\supp{\mu}} \frac{m(\omega(\mu)^{\rho_{\phi_{\mu}}(x)})}{\rho_{\phi_{\pi}}(x)} d(\phi_{e_{n}}\circ\dots\circ\phi_{e_{i}})_{\ast}\lambda(x)\\
&=&\sum_{\bar{\mu}\in \bar{E}_{\pi}}\int_{\supp{\bar{\mu}}} \frac{\lg^{\pi}(\bar{\mu})m(\omega(\bar{\mu})^{\rho_{\phi_{\bar{\mu}}}(x)})}{\rho_{\phi_{\pi}}(x)} d(\phi_{e_{n}}\circ\dots\circ\phi_{e_{i}})_{\ast}\lambda(x)\\
&=&\sum_{\bar{\mu}\in \bar{E}_{\pi}}\int_{\supp{\bar{\mu}}} \frac{\lg^{\pi}(\bar{\mu})m(\omega(\bar{\mu})^{\rho_{\phi_{\bar{\mu}}}(x)})}{\rho_{\phi_{\bar{\mu}}}(x)\times \lg^{\pi}(\bar{\mu})} d(\phi_{e_{n}}\circ\dots\circ\phi_{e_{i}})_{\ast}\lambda(x)\\
&=&\sum_{\bar{\mu}\in \bar{E}_{\pi}}\int_{\supp{\bar{\mu}}} \frac{m(\omega(\bar{\mu})^{\rho_{\phi_{\bar{\mu}}}(x)})}{\rho_{\phi_{\bar{\mu}}}(x)} d(\phi_{e_{n}}\circ\dots\circ\phi_{e_{i}})_{\ast}\lambda(x)
\end{eqnarray*}

We will now use the fact that if $\bar{\mu}\in\bar{E}_{\pi}$, then the associated map $\phi_{\bar{\mu}}$ is equal --  on $\supp{\bar{\mu}}$ --  to $\phi_{\pi}^{k}$ for $k=\lg^{\pi}(\bar{\mu})$. We will also need to name the edges in an element $\bar{\mu}\in\bar{E}_{\pi}$; we will denote them by $f_{0},f_{1},\dots,f_{p}$ where $p=\lg(\bar{\mu})=\lg^{\pi}(\bar{\mu})\times\lg(\pi)$. Using this notation and the preceding remark, we have that for all measurable map $f:X\rightarrow \realN$ and all integer $l\in\{0,\dots,k-1\}$:
$$\int_{\supp{\bar{\mu}}}f(x)d((\phi_{\pi})^{l}\circ\phi_{e_{n}}\circ\phi_{e_{n-1}}\circ\dots\circ\phi_{e_{i}})_{\ast}\lambda(x)=\int_{\supp{\bar{\mu}}} f(x)d(\phi_{f_{p}}\circ\phi_{f_{p-1}}\circ\dots\circ\phi_{f_{i+(k-l)\lg(\pi)}})_{\ast}\lambda(x)$$

Using what we have proved up to now, we can compute $q_{m}$:
\begin{eqnarray*}
q_{m}(\pi)&=&\frac{1}{\lg(\pi)}\sum_{i=0}^{\lg(\pi)-1}\int_{\supp{\pi}} \frac{m(\omega(\pi)^{\rho_{\phi_{\pi}}(x)})}{\rho_{\phi_{\pi}}(x)}d(\phi_{e_{n}}\circ\phi_{e_{n-1}}\circ\dots\circ\phi_{e_{i}})_{\ast}\lambda(x)\\
&=&\frac{1}{\lg(\pi)}\sum_{i=0}^{\lg(\pi)-1}\sum_{\bar{\mu}\in \bar{E}_{\pi}}\int_{\supp{\bar{\mu}}} \frac{m(\omega_{\bar{\mu}}^{\rho_{\phi_{\bar{\mu}}}(x)})}{\rho_{\phi_{\bar{\mu}}}(x)} d(\phi_{e_{n}}\circ\dots\circ\phi_{e_{i}})_{\ast}\lambda(x)\\
&=&\sum_{\bar{\mu}\in \bar{E}_{\pi}}\frac{1}{\lg(\mu)}\left(\frac{\lg(\mu)}{\lg(\pi)}\sum_{i=0}^{\lg(\pi)-1}\int_{\supp{\bar{\mu}}} \frac{m(\omega(\bar{\mu})^{\rho_{\phi_{\bar{\mu}}}(x)})}{\rho_{\phi_{\bar{\mu}}}(x)} d(\phi_{e_{n}}\circ\dots\circ\phi_{e_{i}})_{\ast}\lambda(x)\right)\\
&=&\sum_{\bar{\mu}\in \bar{E}_{\pi}}\frac{1}{\lg(\mu)}\left(\sum_{l=0}^{\frac{\lg(\mu)}{\lg(\pi)}-1}\sum_{i=0}^{\lg(\pi)-1}\int_{\supp{\bar{\mu}}} \frac{m(\omega(\bar{\mu})^{\rho_{\phi_{\bar{\mu}}}(x)})}{\rho_{\phi_{\bar{\mu}}}(x)} d(\phi_{e_{n}}\circ\dots\circ\phi_{e_{i}})_{\ast}\lambda(x)\right)\\
&=&\sum_{\bar{\mu}\in \bar{E}_{\pi}}\frac{1}{\lg(\mu)}\left(\sum_{l=0}^{\frac{\lg(\mu)}{\lg(\pi)}-1}\sum_{i=0}^{\lg(\pi)-1}\int_{\supp{\bar{\mu}}} \frac{m(\omega(\bar{\mu})^{\rho_{\phi_{\bar{\mu}}}(x)})}{\rho_{\phi_{\bar{\mu}}}(x)} d(\phi_{\pi}^{l}\circ\phi_{e_{n}}\circ\dots\circ\phi_{e_{i}})_{\ast}\lambda(x)\right)\\
&=&\sum_{\bar{\mu}\in \bar{E}_{\pi}}\frac{1}{\lg(\mu)}\sum_{i=0}^{\lg(\mu)-1}\int_{\supp{\bar{\mu}}} \frac{m(\omega(\bar{\mu})^{\rho_{\phi_{\bar{\mu}}}(x)})}{\rho_{\phi_{\bar{\mu}}}(x)} d(\phi_{f_{p}}\circ\phi_{f_{p-1}}\circ\dots\circ\phi_{f_{i}})_{\ast}\lambda(x)\\
&=&\sum_{\bar{\mu}\in \bar{E}_{\pi}}q_{m}(\bar{\mu})\\
\end{eqnarray*}

Which shows that $q_{m}$ is refinement-invariant.
\end{proof}

The two preceding lemmas have as a direct consequence the following proposition which shows that we defined a family of circuit-quantifying maps.

\begin{proposition}
Let $(X,\mathcal{B},\mu)$ be a trefoil space, $\mathfrak{m}$ the microcosm of measure-inflating maps, and $m:\Omega\rightarrow \realposN\cup\{\infty\}$ be a tracial map. The function $q_{m}$ is a $\mathfrak{m}$-circuit-quantifying map.
\end{proposition}

\subsection{Circuit-Quantifying Map in the General Case}

This result can now be extended to the microcosm of all non-singular measurable-preserving transformations (the \emph{macrocosm} on $X$). We first show an easy lemma.

\begin{lemma}
If $\rho:X\rightarrow \naturalN$ is measurable and for all $i\in\naturalN$ the maps $\phi_{i}$ are measurable, then the following map is measurable:
$$f(x)=\sum_{i=0}^{\rho(x)} \phi_{i}(x)$$
\end{lemma}

\begin{proof}
Indeed, if $X_{i}$ denotes the measurable set $\rho^{-1}(i)$ for all integer $i$, then the restriction of $f$ to $X_{i}$ is equal to the finite sum $\sum_{k=0}^{i}\phi_{i}(x)$ which is  measurable on $X_{i}$. 
\end{proof}

This lemma ensures us that the following definition makes sense.

\begin{definition}
Let $X$ be a trefoil space, and $m:\Omega\rightarrow \realposN\cup\{\infty\}$ be a function. We define the map:
\begin{equation*}
\bar{q}_{m}:\pi=e_{0}\dots e_{n}\mapsto\sum_{j=0}^{n}\int_{\supp{\pi}} \sum_{k=0}^{\rho_{\phi_{\pi}}(x)-1}\frac{m(\omega(\pi)^{\rho_{\phi_{\pi}}(\phi_{\pi}^{k}(x))})}{(n+1)\rho_{\phi_\pi}(x)\rho_{\phi_{\pi}}(\phi_{\pi}^{k}(x))}d(\phi_{e_{n}}\circ\phi_{e_{n-1}}\circ\dots\circ\phi_{e_{j}})_{\ast}\lambda(x)
\end{equation*}
\end{definition}

We now have to check that \autoref{lemma1} and \autoref{lemma2} still hold in this general setting. This can easily be seen because of the following computation, where we use the convention that $e_{k}$ denote $e_{k \mod{n+1}}$ and $\tilde{\pi}^{i}$ denotes the restriction of $\pi^{i}$, the $i$-times concatenation of $\pi$, to $\supp{\pi}_{i}$:
\begin{eqnarray*}
\bar{q}_{m}(\pi)&=&\sum_{j=0}^{\lg(\pi)}\int_{\supp{\pi}} \sum_{k=0}^{\rho_{\phi_{\pi}}(x)-1}\frac{m(\omega(\pi)^{\rho_{\phi_{\pi}}(\phi_{\pi}^{k}(x))})}{\lg(\pi)\rho_{\phi_\pi}(x)\rho_{\phi_{\pi}}(\phi_{\pi}^{k}(x))}d(\phi_{e_{n}}\circ\phi_{e_{n-1}}\circ\dots\circ\phi_{e_{j}})_{\ast}\lambda(x)\\
&=&\sum_{j=0}^{\lg(\pi)}\sum_{i\in\naturalN^{\ast}}\int_{\supp{\pi}_{i}} \sum_{k=0}^{i}\frac{m(\omega(\pi)^{\rho_{\phi_{\pi}}(\phi_{\pi}^{k}(x))})}{\lg(\pi)i\rho_{\phi_{\pi}}(\phi_{\pi}^{k}(x))}d(\phi_{e_{n}}\circ\phi_{e_{n-1}}\circ\dots\circ\phi_{e_{j}})_{\ast}\lambda(x)\\
&=&\sum_{j=0}^{\lg(\pi)}\sum_{i\in\naturalN^{\ast}}\sum_{k=0}^{i}\int_{\supp{\pi}_{i}} \frac{m(\omega(\pi)^{\rho_{\phi_{\pi}}(\phi_{\pi}^{k}(x))})}{\lg(\pi)i\rho_{\phi_{\pi}}(\phi_{\pi}^{k}(x))}d(\phi_{e_{n}}\circ\phi_{e_{n-1}}\circ\dots\circ\phi_{e_{j}})_{\ast}\lambda(x)\\
&=&\sum_{i\in\naturalN^{\ast}}\sum_{j=0}^{\lg(\pi)}\frac{1}{\lg(\pi)i}\sum_{k=0}^{i}\int_{\supp{\pi}_{i}} \frac{m(\omega(\pi)^{\rho_{\phi_{\pi}}(\phi_{\pi}^{k}(x))})}{\rho_{\phi_{\pi}}(\phi_{\pi}^{k}(x))}d(\phi_{e_{n}}\circ\phi_{e_{n-1}}\circ\dots\circ\phi_{e_{j}})_{\ast}\lambda(x)\\
&=&\sum_{i\in\naturalN^{\ast}}\sum_{j=0}^{\lg(\pi)}\frac{1}{\lg(\pi)i}\sum_{k=0}^{i}\int_{\supp{\pi}_{i}} \frac{m(\omega(\pi)^{\rho_{\phi_{\pi}}(x)})}{\rho_{\phi_{\pi}}(x)}d(\phi_{\pi}^{k}\circ\phi_{e_{n}}\circ\phi_{e_{n-1}}\circ\dots\circ\phi_{e_{j}})_{\ast}\lambda(x)\\
&=&\sum_{i\in\naturalN^{\ast}}\sum_{j=0}^{\lg(\pi^{i})}\frac{1}{\lg(\pi^{i})}\int_{\supp{\pi}_{i}} \frac{m(\omega(\pi)^{\rho_{\phi_{\pi}}(x)})}{\rho_{\phi_{\pi}}(x)}d(\phi_{e_{n\times i}}\circ\phi_{e_{n\times i-1}}\circ\dots\circ\phi_{e_{j}})_{\ast}\lambda(x)\\
&=&\sum_{i\in\naturalN^{\ast}}q_{m}(\tilde{\pi}^{i})
\end{eqnarray*}
%\textcolor{red}{Explain the computations (this is not the only place).}
From this result, and the fact that $\phi_{\tilde{\pi}^{i}}$ is measure-preserving, one can adapt the proofs of \autoref{lemma1} and \autoref{lemma2}, and show the following theorem which, together with \autoref{thm_associativity} and \autoref{thm_trefoilppty}, finishes the proof of \autoref{mainthm}.

\begin{theorem}\label{thm_trefoilspace}
Let $(X,\mathcal{B},\mu)$ be a trefoil space, $\mathfrak{m}$ the associated macrocosm, and $m:\Omega\rightarrow \realposN\cup\{\infty\}$ be a tracial map. The map $\bar{q}_{m}$ is a $\mathfrak{m}$-circuit-quantifying map.
\end{theorem}

\begin{example}
This setting is a far-reaching extension of our previous work on directed weighted graphs \cite{seiller-goia}. Indeed, the latter framework is recovered as the special case of a discrete space endowed with the counting measure. In this case, one can notice that the map $\rho_{\phi_{\pi}}$ is constantly equal to $1$ and therefore the family of measurement just defined can be computed with the simpler expression $\bar{q}_{m}(\pi)=m(\omega(\pi))$. The family of measurements defined from these functions thus turn out to be equal to the family of measurement considered on graphs \cite{seiller-goia}. In particular, the measurement defined from the map $\bar{q}_{m}=-\log(1-x)$ corresponds to Girard's measurement based on the determinant \cite{seiller-goim}. 
\end{example}

\begin{example}
Let us consider the trefoil space $X=[0,1]$ endowed with Lebesgue measure, and the microcosm of measure-preserving maps. Since each measure-preserving map on $X$ defines a unitary acting on the Hilbert space $L^{2}(X)$ by pre-composition, it is easy to associate to any $\complexN$-weighted graphing $G$ a linear combination $[G]$ of partial isometries on $L^{2}(X)$. This might not define an operator in general since the obtained operator might not be bounded, but we will restrict the discussion to the set of graphings for which $[G]$ is an operator. This set of graphings can be shown to have the following properties:
\begin{itemize}
\item it contains, for each integer $k$, the \enquote{$k\times k$-matrices algebra\footnote{Of course, this is not the matrix algebra, but one can show that the operator $[G]$ associated to such a graphing $G$ is the image of a $k\times k$ matrix through a well-chosen injective morphism.}} of graphings constructed from translations between intervals $I_{l}=[\frac{l}{k},\frac{l+1}{k}]$, i.e.\ directed weighted graphs on $k$ vertices;
\item it has a trace: for each $f\in \mathfrak{m}$ one can define $tr(f)$ as the measure of the set of fixed points of $f$; this trace, when restricted to the $k\times k$ matrix algebra defined above yields the usual (normalized, i.e.\ $tr(1)=1$) trace of matrices;
\end{itemize}
Thus, the set of such graphings plays the rôle of the type {II}$_{1}$ hyperfinite factor. The same reasoning shows that the $\complexN$-weighted graphings in the microcosm of measure-preserving maps on $X=\realN$ with the Lebesgue measure play the rôle of the type {II}$_{\infty}$ hyperfinite factor. 
\end{example}

\begin{remark}
We did not show here a formal correspondance, but a proof of such a result most surely exists. In particular, in the case where $X$ is the real line, it is known that the type {II}$_{\infty}$ factor arises as the von Neumann algebra generated by:
\begin{itemize}
\item elements of $L^{\infty}(\realN)$ acting on $L^{2}(\realN)$ by multiplication;
\item the unitaries induced by precomposition by rational translation.
\end{itemize}
We did not think however that such a result would be of great interest in this paper, as we already know from the discrete case discussed above and results from previous papers \cite{seiller-goim,seiller-goia} that our setting generalises Girard's constructions using operator algebras.
\end{remark}

%\textcolor{red}{The matrices are easily represented with finite spaces endowed with discrete topology and counting measure. The case of bounded operators acting on a separable infinite-dimensional Hilbert space is obtained in a similar fashion by forgetting about the \enquote{finite} part. Other operator algebras are more complex to explain: the type ${II}_{\infty}$ factor could be obtained as the rational translations on the real line since the crossed product $L^{\infty}(\realN,\lambda)\rtimes \rationalN$ is a definition of this factor.}

\subsection{Example: Unification \enquote{Algebras}}

%\note{If $\mathfrak{m}\leqslant\mathfrak{m'}$, then a graphing in $\mathfrak{m}$ is a graphing in $\mathfrak{m'}$, hence the model is more expressive. In the other direction, a circuit-quantifying map for $\mathfrak{m'}$ is a circuit-quantifying map for $\mathfrak{m}$ but not conversely. Hence, the idea is to obtain circuit-quantifying maps for the most general microcosm.}

%\subsubsection{Example: Baire Space and Unification}

%\textcolor{red}{We now show how Girard's unification "algebra" (in fact a semi-ring) is representable in our setting.}

As already mentioned, it can be shown that the framework of graphing generalises Girard's constructions based on operators \cite{goi1,goi2,goi5}. We will now explain how the alternative approach he uses, namely using \enquote{algebras of clauses} \cite{goi3} or \enquote{unification algebra} \cite{goi6,goi6light}, is also a particular case of our constructions on graphings.

We thus show how Girard's notions of flows and wirings can be understood in terms of graphings. This gives intuitions on what he calls the \enquote{unification algebra} which is nothing more than the algebra generated by the set of graphings on the adequate space $B(\Sigma)$.

In the following we fix a countable (infinite) set of variables $\text{Var}$.

\begin{definition}
A \emph{signature} is a tuple $(\text{Const},\text{Fun})$ where $\text{Const}$ contains \emph{symbols of constants} and $\text{Fun}$ contains a finite number of \emph{symbols of functions}. We say the signature is \emph{free} if the set $\text{Const}$ is empty.
\end{definition}

\begin{definition}
The terms defined by a signature $\Sigma$ are defined by the grammar:
$$T:=x~|~c~|~f(T,\dots,T)~~~~~ (x\in\text{Var}, c\in\text{Const}, f\in\text{Fun})$$
A \emph{closed term} is a term that does not contain any variables. A term which is not closed is said to be \emph{open}.
\end{definition}

\begin{definition}
If $\Sigma$ is a free signature, there are no finite closed terms. In this particular case, we define the set of closed terms as the trees defined co-inductively as follows:
$$T:=f(T,\dots, T)~~~~(f\in \text{Fun})$$
We can understand these closed terms as infinite rooted trees labelled by function symbols.
\end{definition}

\begin{definition}
Let $\Sigma$ be a free signature. We define the topological space $B(\Sigma)$ as the set of closed terms considered with the topology induced by the set of open terms: $\mathcal{O}(u(x_{1},\dots,x_{n}))$ is defined as the set of closed terms $u(t_{1},\dots,t_{n})$ where the $t_{i}$ are closed terms. This space can be endowed with a $\sigma$-finite radon measure $\lambda$ defined inductively following the definition of open terms. 

We define for each enumeration $e:\text{Fun}\rightarrow\naturalN^{\ast}$ and each occurence of a variable $x$ in in the open term $u(x_{1},\dots,x_{n})$, the quantity $\lambda^{u(x_{1},\dots,x_{n})}_{e}(x)$ by following the path $f_{1}\dots f_{k}$ of function symbols from the root to $x$ in the syntactic tree of $u(x_{1},\dots,x_{n})$:
$$\lambda^{u(x_{1},\dots,x_{n})}_{e}(x)=\prod_{i=1}^{k} \frac{1}{2^{e(f_{i})}}$$
with the usual convention that an empty product equal $1$.

The measure of $\mathcal{O}(u)=\mathcal{O}(u(x_{1},\dots,x_{n}))$ is then defined as ($\text{Occ}^{u}(x_{i})$ denotes the set of occurences of $x_{i}$ in $u$):
$$\lambda_{e}(\mathcal{O}(u))=\sum_{i=1}^{n}\max_{x\in\text{Occ}^{u}(x_{i})}\lambda_{e}^{u}(x) $$

Notice that for linea terms, i.e.\ when variables occur only once in $u$, this quantity can be defined directly by induction on the structure of $u$:
\begin{equation*}
\begin{array}{rcll}
\lambda_{e}(x_{i})&=&1&\text{ when $x_{i}\in\text{Var}$}\\
\lambda_{e}(f_{i}(t_{1},\dots,t_{k_{i}})&=&\frac{1}{2^{e(f_{i})}}\sum_{i=1}^{k_{i}} \frac{\lambda_{e}(t_{i})}{k_{i}}
\end{array}
\end{equation*}
\end{definition}

\begin{remark}
One can define other measures (which are more satisfying in some respect) in some specific cases:
\begin{itemize}
\item If the set $\text{Fun}$ is finite, we write $K=\sum_{f_{i}\in\text{Fun}}k_{i}$, where $k_{i}$ is the arity of $f_{i}$, we can define 
$$\lambda^{u(x_{1},\dots,x_{n})}_{e}(x)=\prod_{i=1}^{k} \frac{1}{K}$$
The direct inductive definition of the meausre of linear terms then becomes:
\begin{equation*}
\begin{array}{rcll}
\lambda(x_{i})&=&1&\text{ when $x_{i}\in\text{Var}$}\\
\lambda(f_{i}(t_{1},\dots,t_{k_{i}})&=&\frac{1}{K}\sum_{i=1}^{k_{i}} \lambda(t_{i})
\end{array}
\end{equation*}
\item If the set $\text{Fun}$ is infinite but the number of functions of a given arity $k$ is finite (we denote it by $a_{k}$), we can define:
$$\lambda^{u(x_{1},\dots,x_{n})}_{e}(x)=\prod_{i=1}^{k} \frac{1}{a_{i}\times 2^{k_{i}+1}}$$
The direct inductive definition of the meausre of linear terms then becomes:
\begin{equation*}
\begin{array}{rcll}
\lambda(x_{i})&=&1&\text{ when $x_{i}\in\text{Var}$}\\
\lambda(f_{i}(t_{1},\dots,t_{k_{i}})&=&\frac{1}{a_{i}\times 2^{k_{i}+1}}\sum_{i=1}^{k_{i}} \frac{\lambda(t_{i})}{k_{i}}
\end{array}
\end{equation*}
\end{itemize}
\end{remark}

\begin{definition}
Let $\Sigma$ be a non-free signature. We define the topological space $B(\Sigma)$ as the set of closed terms endowed with the discrete topology. This space can be endowed with the counting measure.
\end{definition}

\begin{theorem}
For any signature $\Sigma$, the space $B(\Sigma)$ is a trefoil space. 
\end{theorem}

\begin{proof}
It is clear in the case of a non-free signature. In the case of a free signature, the space is clearly Hausdorff. It is second-countable since we defined the topology as induced by a countable number of open sets.% Finally, it is sequentially compact, hence compact, thus locally compact.
\end{proof}

\begin{definition}
A \emph{flow} is an ordered pair $u\flow t$, where $u,t$ are terms with the same variables. A \emph{wiring} is a sum of flows.
\end{definition}

\begin{definition}
A flow $u(x_{1},\dots,x_{k})\flow t(x_{1},\dots,x_{k})$ represents the following non-singular \preservesmeasurable map from the open $\mathcal{O}(t)$ to the open $\mathcal{O}(u)$:
$$[u\flow t]:=t(T_{1},\dots,T_{k})\mapsto u(T_{1},\dots,T_{k})$$
Given a wiring $W$, we can therefore associate a graphing $[W]$ to it.
\end{definition}

Now, an operation of \emph{composition} is defined on flows, and therefore on wirings. We only sketch here the definitions and refer the reader interested in more details to one of the author's recent joint work with Aubert an Bagnol \cite{lics-ptime}. Composition is defined as follows on flows: if $u\flow t$ and $u'\flow t'$ are flows with disjoint sets of free variables, the composition $(u\flow t)\cdot(u'\flow t')$ is equal to $u\theta\flow t'\theta$ if and only if there exists a \enquote{most general unifier} $\theta$ for $t$ and $u'$ (this is explained in the proof of \autoref{executioncommute}), and is undefined otherwise. This extends to wirings as follows: $(\sum_{i\in I} u_{i}\flow t_{i})\cdot(\sum_{j\in J} u'_{j}\flow t'_{j})=\sum_{(i,j)\in I\times J} (u_{i}\flow t_{i})\cdot(u'_{j}\flow t'_{j})$; notice that this may produce the empty sum of flows if all compositions are undefined. Now, the set of wirings together with the composition and sum operations form an algebra. Among the elements of this algebra, we distinguish \enquote{hermitians} as those elements $\sum_{i\in I} u_{i}\flow t_{i}=\sum_{i\in I} t_{i}\flow u_{i}$, in analogy with operator algebras. Following this analogy, an hermitian $\sigma$ produces a projection $\sigma^{2}$: one can check that $\sigma^{4}=\sigma^{2}=\sigma^{\ast}$ where $(\cdot)^{\ast}$ is the involution defined by $(\sum_{i\in I} u_{i}\flow t_{i})^{\ast}=\sum_{i\in I} t_{i}\flow u_{i}$. Notice the flow $x\flow x$ defines an identity. We write $1-\sigma^{2}$ the projection such that $ \sigma^{2}+(1-\sigma^{2})$ is the identity. 

Now, given wirings $U$ and $\sigma$ such that $\sigma$ is a hermitian and $U\sigma$ is nilpotent, we define their \emph{execution}\footnote{The terminology follows more recent work in geometry of interaction. It is important to note that the first work using flows and wirings \cite{goi3} uses a different terminology and calls this expression the \enquote{result of the execution}, while the execution is simply the sum $\sum_{i\geqslant 0}U(\sigma U)^{k}$.} as the wiring $(1-\sigma^{2})\left(\sum_{i\geqslant 0}U(\sigma U)^{k}\right)(1-\sigma^{2})$. This is a straightforward translation of the operator-algebraic execution formula introduced with Girard's fist paper on geometry of interaction \cite{goi1}. We refer to the author's recent work on maximal abelian subalgebras \cite{seiller-masas} for an overview of the GoI program and the execution formula (based on the operator-algebraic approach).

\begin{theorem}\label{executioncommute}
Let $\Sigma$ be a signature. The map $W\mapsto [W]$ commutes with execution. 
\end{theorem}

\begin{proof}
This is shown easily by looking at the definition of composition of flows. A global substitution is a map from the set of variables to the set of terms, and we denote by $v\theta$ the result of the substitution of each variable $x_{i}$ in $v$ by the term $\theta(x_{i})$. We say two terms $v,v'$ are \emph{unifiable} when there exists a global substitution $\theta$ such that $v\theta=v'\theta$. In this case, there exists a \emph{principal unifier}, i.e.\ a substitution $\theta_{0}$ such that any substitution $\theta$ satisfying $v\theta=v'\theta$ can be factorised through $\theta_{0}$, i.e.\ there exists $\theta'$ such that $\theta=\theta_{0}\theta'$. The composition $(u\flow v)(v'\flow w)$ is then equal to $0$ if $v$ and $v'$ are not unifiable, and to $u\theta_{0}\flow w\theta_{0}$ if they are unifiable and $\theta_{0}$ is the principal unifier.

Now, it is not hard to see that two terms $u,v$ are unifiable if and only if the open sets $\mathcal{O}(u)$ and $\mathcal{O}(v)$ have a non-trivial (of strictly positive measure) intersection. This intersection is then an open set equal to $\mathcal{O}(u\theta_{0})=\mathcal{O}(v\theta_{0})$ where $\theta_{0}$ is the principal unifier of $u$ and $v$. Thus composition of flows corresponds to considering the partial composition of the associated measurable maps.

This implies that the composition of wirings $W\circ W'=(\sum_{i\in I} f_{i})\circ(\sum_{j\in J} g_{j})$, which is defined as $\sum_{i\in I, j\in J} f_{i}\circ g_{j}$, corresponds to taking the graphing of alternating paths of length $2$ between the graphings $[W]$ and $[W']$.

Finally, the execution formula\footnote{\label{footnoteref}We restrict ourselves here to the simple case of the execution formula where $\sigma$ corresponding to the modus ponens. The case of the more involved formula corresponding to the general cut rule obviously holds, as it can be recovered from the simpler one considered here.} $\Ex(U,\sigma) = (1-\sigma^{2})U(1-\sigma U)^{-1}(1-\sigma^{2})$, which is computed as:
$$\Ex(U,\sigma) = (1-\sigma^{2})\left(\sum_{i\geqslant 0}U(\sigma U)^{k}\right)(1-\sigma^{2})$$
corresponds to the execution of graphings $[U]\plugmes{}[\sigma]$ because:
\begin{itemize}
\item the conjugation by $(1-\sigma^{2})$ is used to restrict the result to wirings living \emph{outside of the cut}, which is dealt with in the execution of graphings by considering the restriction of paths $\phi$ to their \emph{outside component} $[\phi]_{o}^{o}$;
\item for each integer $k$, the terms $U(\sigma U)^{k}$ correspond to the set of alternating paths of length $2k+1$ as we already noticed, which are the only possible lengths of alternating paths in this case\footnote{This is due to the fact that we restricted to the simpler case of the execution formula, see \autoref{footnoteref}.}
\end{itemize}
Thus the embedding $W\mapsto [W]$ commutes with execution, i.e.\ the execution of graphings computes the execution formula on wirings.
\end{proof}

This shows that the notion of graphing is a non-trivial generalisation of the notion of wirings considered lately by Girard. In particular, the construction of GoI models based on wirings can be expressed in terms of graphings. Moreover, the notion of graphing is much more powerful than Girard's notion of wiring. Indeed, the syntactic definitions of wiring do not allow for the quantitative features of graphings, namely the family of measurement considered above. In particular, the only definable notion of orthogonality one can consider on wiring is defined as the nilpotency of the product of two wirings, which corresponds to the measurement defined above with the dull circuit-quantifying map $m(x)=\infty$.

\subsection{Digression: Topological Graphings}\label{topographings}

One generalisation of Girard's framework based on unification would be to consider a weakened definition of flow $t\flow u$ where the variables of $u$ and $t$ do not match exactly but only the inclusion $\text{Var}(u)\subset\text{Var}(t)$ holds. In this case, the interpretation of flows as non-singular maps is no longer valid as the \enquote{weakening} thus allowed makes it possible that the inverse image of a set of measure zero is of strictly positive measure. This is seen by taking the inverse image through $[u(x)\flow v(x,y)]$ of a closed term $u(T)$ in the free signature case.

This mismatch is due to the fact that in the process of generalising the notion of flows, we stepped outside of measure theory. The map interpreting the flows are no longer non-singular (they are still measurable though, but non-singularity is necessary to obtain the associativity of execution), but they are continuous. A topological notion of graphing, corresponding somehow to a notion of \emph{pseudo-monoid} to recall Cartan's notion of pseudo-group \cite{cartan1,cartan2,pseudogroup_book}, could be applied here instead of our measurable approach. 

Indeed, define a \emph{topological graphing} on a topological space $X$ as a countable family $F=\{(\omega_{e}^{F},\phi_{e}^{F}: S_{e}^{F}\rightarrow T_{e}^{F})\}_{e\in E^{F}}$, where, for all $e\in E^{F}$ (the set of \emph{edges}):
\begin{itemize}[noitemsep,nolistsep]
\item $\omega_{e}^{F}$ is an element of $\Omega$, the \emph{weight} of the edge $e$;
\item $S_{e}^{F}$ and $T_{e}^{F}$ are open sets, the \emph{source} and \emph{target} of the edge $e$;
\item $\phi_{e}^{F}$ is an open continuous map from $S_{e}^{F}$ to $T_{e}^{F}$, the \emph{realiser} of the edge $e$.
\end{itemize}

Then the notions of paths and cycles can be defined as in the more complex case of measurable graphings considered until now. We can therefore define the execution between topological graphings and show associativity. We note here that the mismatch with associativity in the measurable case which arose from the non-singularity is no longer a problem since we do not quotient by sets of measure $0$ anymore. In the same way we defined refinements, one can define refinements in this topological setting and define a corresponding equivalence relation. It is easily shown that execution is compatible with this equivalence relation, and we therefore can mimic almost all results of \autoref{sec_execution} and \autoref{sec_measurement}, forgetting about almost-everywhere equality. However, the contents of \autoref{sec_trefoil} depends greatly on the fact that we are dealing with measure spaces. The only obvious way to obtain the trefoil property in the topological case is therefore to consider the measurement to be $\infty$ when there exists a cycle and $0$ otherwise. This means that the only sensible notion of orthogonality one can define corresponds to nilpotency. Of course, one may be able to define other measurements, but it would be much more difficult than in the measurable case where we can use the radon measure on the space. 

This explains why, even though one could do all the constructions we considered in this easier setting, we chose to work with measure spaces. The fact that the topological approach is easier comes with its drawback: the topological setting is much poorer and we would miss the quantitative flavor we obtained here. In particular, we loose the generalisation of the determinant measure, as well as any measurement built on circuit-quantifying maps which take values outside $\{0,\infty\}$.

%\textcolor{red}{Is this finished?}

%\textcolor{red}{generalisation of the notion of flows and topological graphings}

\section{The Real Line and Quantification}\label{sec_mall2}

We now consider any microcosm on the real line endowed with Lebesgue measure which contains the microcosm of affine\footnote{An affine map is a map $x\mapsto \alpha x+\beta$ where $\alpha,\beta$ are real numbers. These maps are the only ones we will use in order to interpret proofs of MALL$^{2}$.} maps. We fix $\Omega=]0,1]$ endowed with the usual multiplication and we chose any map $m:\Omega\rightarrow \realposN\cup\{\infty\}$ such that $m(1)=\infty$. Then, as we showed in the preceding section, the map $q_{m}$ is a $\mathfrak{m}$-circuit-quantifying map. We can thus define the measurement $\meas{\cdot,\cdot}$ corresponding to $q_{m}$ following \autoref{measurement}. This measurement and the execution of graphings satisfy the trefoil property. We will now show how to interpret in this case multiplicative-additive linear logic with second-order quantification. 

As remarked earlier, the set of $\Omega$-weighted graphings in the microcosm of measure-preserving maps on the real line with Lebesgue measure corresponds intuitively to the hyperfinite type {II}$_{\infty}$ factor. We are therefore considering an extension of the setting of Girard's hyperfinite geometry of interaction by considering the larger microcosm of affine transformations. The general result we obtained earlier allows us to do so while still disposing of a measurement, which in the particular case of the map $m(x)=-\log(1-x)$ generalises the measurement based on Fuglede-Kadison determinant \cite{FKdet}. We point out that this would correspond in Girard's setting to extend the set of operators considered (i.e.\ consider an algebra $\vn{A}$ containing strictly the type {II}$_{\infty}$ hyperfinite factor), while still disposing of the Fuglede-Kadison determinant. The existence of such an extension is not clear, and should it exists, its definition would be far from trivial!

The extension to affine maps gives us the possibility of defining real second-order quantification, which was not the case of Girard in his hyperfinite GoI model. Indeed, the fact that projects -- which interpret proofs -- have a \emph{location} forces him to consider quantification over a given location, something that we also consider here. However, Girard cannot interpret the right existential introduction (from $\vdash B[A/X],\Gamma$ deduce $\vdash \exists X~B,\Gamma$) correctly because the location of the formula $A$ and the location of the variable $X$ might not have the same size\footnote{The restriction to operators in the type {II}$_{\infty}$ factor implies that unitaries preserve the sizes.}! We bypass this problem here by using \emph{measure-inflating faxes}, i.e.\ bijective bi-measurable transformations that multiply the size by a scalar.

\subsection{Basic Definitions}\label{sec_basicdefs}

The model is based on the same constructions as the one described in previous work \cite{seiller-goia}. We recall the basic definitions of projects and behaviors, which will be respectively used to interpret proofs and formulas, as well as the definition of connectives.

\begin{itemize}[noitemsep,nolistsep]
\item a \emph{project} with carrier $V^{A}$ is a triple $\de{a}=(a,V^{A},A)$, where $a$ is a real number, $A=\sum_{i\in I^{A}} \alpha^{A}_{i} A_{i}$ is a finite formal (real-)weighted sum of graphings with carrier included in $V^{A}$; here the projects considered always have a carrier of finite measure;
\item two projects $\de{a,b}$ are \emph{orthogonal} when:
$$\sca{a}{b}=a(\sum_{i\in I^{A}}\alpha^{B}_{i})+b(\sum_{i\in I^{B}}\alpha^{B}_{i})+\sum_{i\in I^{A}}\sum_{j\in I^{B}} \alpha_{i}^{A}\alpha^{B}_{j}\meas{A_{i},B_{j}}\neq 0,\infty$$
\item the \emph{execution} of two projects $\de{a,b}$ is defined as ($\Delta$ denotes the symmetric difference):
$$\de{a\plug b}=(\sca{a}{b},V^{A}\Delta V^{B},\sum_{i\in I^{A}}\sum_{j\in I^{B}} \alpha^{A}_{i}\alpha^{B}_{j} A_{i}\plugmes B_{j})$$
\item the \emph{sum} of two projects $\de{a,b}$ is defined as:
$$\de{a+b}=(a+b,V^{A}\cup V^{B},\sum_{i\in I^{A}}\alpha^{A}_{i}A_{i}+\sum_{j\in I^{B}} \alpha^{B}_{j} B_{j})$$
\item The \emph{zero project} with carrier $V$ is defined as the project $(0,\emptyset)$, where $\emptyset$ denotes the (unique up to a.e. equality) empty graphing (\autoref{emptygraphing}).
\item if $\de{a}$ is a project and $V$ is a measurable set such that $V^{A}\subset V$, we define the extension $\de{a}_{\uparrow V}$ as the project $(a,V,A)$;
\item a \emph{conduct} $\cond{A}$ with carrier $V^{A}$ is a set of projects with carrier $V^{A}$ which is equal to its bi-orthogonal, i.e.\ $\cond{A}=\cond{A}^{\pol\pol}$. A \emph{behavior} is a conduct $\cond{A}$ such that for all $\lambda\in\realN$, 
$$\begin{array}{rcl}\de{a}\in\cond{A} &\Rightarrow& \de{a+\lambda 0}\in\cond{A}\\
\de{b}\in\cond{A}^{\pol} &\Rightarrow& \de{b+\lambda 0}\in\cond{A}^{\pol}\end{array}$$
\item we define, for every measurable set the \emph{empty} behavior with carrier $V$ as the empty set $\cond{0}_{V}$, and the \emph{full behavior} with carrier $V$ as its orthogonal $\cond{T}_{V}=\{\de{a}~|~\de{a}\text{ of support }V\}$;
\item if $\cond{A,B}$ are two behaviors of disjoint carriers, we define:
\begin{eqnarray*}
\cond{A \otimes B}&=&\{\de{a\plug b}~|~\de{a}\in\cond{A},\de{b}\in\cond{B}\}^{\pol\pol}\\
\cond{A \multimap B}&=&\{\de{f}~|~\forall\de{a}\in\cond{A}, \de{f\plug a}\in\cond{B}\}\\
\cond{A \oplus B}&=&(\{\de{a}_{\uparrow V^{A}\cup V^{B}}~|~\de{a}\in\cond{A}\}^{\pol\pol}\cup \{\de{b}_{\uparrow V^{A}\cup V^{B}}~|~\de{b}\in\cond{B}\}^{\pol\pol})^{\pol\pol}\\
\cond{A \with B}&=&\{\de{a}_{\uparrow V^{A}\cup V^{B}}~|~\de{a}\in\cond{A^{\pol}}\}^{\pol}\cap \{\de{b}_{\uparrow V^{A}\cup V^{B}}~|~\de{b}\in\cond{B}^{\pol}\}^{\pol}
\end{eqnarray*}
\end{itemize}

We now define \emph{localized second order quantification} and show the duality between second order universal quantification and second order existential quantification. %We will not dwell on the strange properties that are due to localization, which are more or less the same as quantification in Ludics \cite{locussolum}, since we will not be talking about second order quantification in the rest of the paper. Indeed, we will prove soundness for Elementary Linear Logic without quantifiers. However, we give in the concluding section some ideas on how a soundness result for Elementary Linear Logic with second order quantification might be obtained. Obtaining this result would mean a slight change in the definition of graphings so as to allow transports of measure as edges instead of restricting to measure-preserving maps. We believe that this modification can be performed without loosing any of the properties needed to define the connectives of linear logic, but this would make the setting we are working with even more complex and we believe these modifications extend beyond the scope of this paper.

\begin{definition}
We define the localized second order quantification as, for any measurable set $L$:
\begin{eqnarray*}
\forall_{L} \cond{X}~\cond{F(X)}&=&\bigcap_{\cond{A}, V^{A}=L} \cond{F(A)}\\
\exists_{L} \cond{X}~\cond{F(X)}&=&\left(\bigcup_{\cond{A}, V^{A}=L} \cond{F(A)}\right)^{\pol\pol}
\end{eqnarray*}
\end{definition}

\begin{proposition}
\begin{equation*}
(\forall_{L} \cond{X}~\cond{F(X)})^{\pol}=\exists_{L} \cond{X}~\cond{(F(X))^{\pol}}
\end{equation*}
\end{proposition}

\begin{proof}
The proof is straightforward. Using the definitions:
\begin{eqnarray*}
(\forall_{L} \cond{X}~\cond{F(X)})^{\pol}&=&\left(\bigcap_{\cond{A}, V^{A}=L} \cond{F(A)}\right)^{\pol}\\
&=&\left(\bigcup_{\cond{A}, V^{A}=L} \cond{(F(A))^{\pol}}\right)^{\pol\pol}\\
&=&\exists_{L} \cond{X}~F^{\pol}(X)
\end{eqnarray*}
Where we used the fact that taking the orthogonal turns an intersection into a union.
\end{proof}

\subsection{Truth}

We now define a notion of \emph{successful project}, which intuitively correspond to the notion of \emph{winning strategy} in game semantics. This notion should be understood as a tentative characterization of those projects which arise as interpretation of proofs. The notion of success defined here is the natural generalisation of the corresponding notion on graphs \cite{seiller-goim,seiller-goia}. The graphing of a successful project will therefore be a disjoint union of \enquote{transpositions}. %Such a graphing can be represented with a set of vertices that could be infinite. Since we are working with equivalence classes of graphings modulo refinements one can always find a simpler representation: a graphing with exactly two edges. 
In the following, we say a weighted sum of graphings $\sum_{i\in I^{A}} \alpha^{A}_{i}A_{i}$ is \emph{balanced} when for all $i,j\in I^{A}$, we have $\alpha^{A}_{i}=\alpha_{j}^{A}$.

\begin{definition}
A project $\de{a}=(a,A)$ is \emph{successful} when it is balanced, $a=0$ and $A$ is a disjoint union of transpositions:
\begin{itemize}[noitemsep,nolistsep]
\item for all $e\in E^{A}$, $\omega^{A}_{e}=1$;
\item for all $e\in E^{A}$, $\exists e^{\ast}\in E^{A}$ such that $\phi^{A}_{e^{\ast}}=(\phi_{e}^{A})^{-1}$ -- in particular $S_{e}^{A}=T_{e^{\ast}}^{A}$ and $T_{e}^{A}=S_{e^{\ast}}^{A}$;
\item for all $e,f\in E^{A}$ with $f\not\in\{e,e^{\ast}\}$, $S^{A}_{e}\cap S^{A}_{f}$ and $T^{A}_{e}\cap T^{A}_{f}$ are of null measure;
\end{itemize}
A conduct $\cond{A}$ is \emph{true} when it contains a successful project.
\end{definition}

%\begin{lemma}
%If $\de{a}=(0,A)$ is successful, one can find a representative $\hat{A}$ of the equivalence class modulo $\sim_{\leqslant}$ of $A$ such that $E^{\hat{A}}$ has its cardinality equal to $2$.
%\end{lemma}
%
%\begin{proof}
%One can define a partition $E_{1},E_{2}$ of $E^{A}$ such that for all couple $e,e^{\ast}$ in $E^{A}$, one has $e\in E_{1}$ if and only if $e^{\ast}\in E_{2}$. We then define: 
%\begin{equation*}
%\hat{A}=\{(1,\bigcup_{e\in E_{1}} \phi^{A}_{e}:\bigcup_{e\in E_{1}} S^{A}_{e}\rightarrow\bigcup_{e\in E_{1}} T^{A}_{e}),(1,\bigcup_{e\in E_{2}} \phi^{A}_{e}:\bigcup_{e\in E_{2}} S^{A}_{e}\rightarrow\bigcup_{e\in E_{2}} T^{A}_{e})\}
%\end{equation*}
%It is clear that $A$ is a refinement of $\hat{A}$ which concludes the proof.
%\end{proof}

\begin{proposition}[Consistency]
The conducts $\cond{A}$ and $\cond{A}^{\pol}$ cannot be simultaneously true.
\end{proposition}

\begin{proof}
We suppose that $\de{a}=(0,A)$ and $\de{b}=(0,B)$ are successful project in the conducts $\cond{A}$ and $\cond{A}^{\pol}$ respectively. Then:
\begin{equation*}
\sca{a}{b}=\meas{A,B}
\end{equation*}
If there exists a cycle whose support is of strictly positive measure between $A$ and $B$, then $\meas{A,B}=\infty$ since we suppose that $m(1)=\infty$. Otherwise, $\meas{A,B}=0$. In both cases we obtained a contradiction since $\de{a}$ and $\de{b}$ cannot be orthogonal.
\end{proof}

\begin{proposition}[Compositionnality]
If $\cond{A}$ and $\cond{A \multimap B}$ are true, then $\cond{B}$ is true.
\end{proposition}

\begin{proof}
Let $\de{a}\in\cond{A}$ and $\de{f}\in\cond{A\multimap B}$ be successful projects. Then:
\begin{itemize}[noitemsep,nolistsep]
\item If $\sca{a}{f}=\infty$, the conduct $\cond{B}$ is equal to $\cond{T}_{V^{B}}$, which is a true conduct since it contains $(0,\emptyset)$;
\item Otherwise $\sca{a}{f}=0$ (this is shown in the same manner as in the preceding proof) and it is sufficient to show that $F\plugmes A$ is a disjoint union of transpositions. But this is straightforward: to each path there corresponds an opposite path and the weights of the paths are all equal to $1$, the conditions on the source and target sets $S_{\pi}$ and $T_{\pi}$ are then easily checked.
\end{itemize}
Finally, if $\cond{A}$ and $\cond{A\multimap B}$ are true, then $\cond{B}$ is true.
\end{proof}

\subsection{Interpretation of proofs}

We now introduce the sequent calculus MALL$^{2}_{\cond{T,0}}$. This is the usual sequent calculus for second-order multiplicative-additive linear logic without multiplicative units. The reason multiplicative units are not dealt with is explained in the author's work on additives \cite{seiller-goia}: multiplicative units exist in the model but they are not behaviors. In a nutshell, multiplicative units are in some ways exponentials of additives units, and an exponentiated formula cannot be a behavior since behaviors do not satisfy the weakening rule. A more involved sequent calculus could be introduced to add the treatment of multiplicative units; the amount of work needed to do so is however too great to be justified without considering exponential connectives as well. We refer the interested reader to the author's forthcoming papers dealing with exponential connectives \cite{seiller-goie,seiller-goif}.

%To deal with the three kinds of conducts we are working with (behaviors, perennial and co-perennial conducts), we introduce three types of formulas.

\begin{definition}
We fix an infinite (countable) set of variables $\mathcal{V}$ and w define formulas of MALL$^{2}_{\cond{T,0}}$ inductively by the following grammar:
\begin{eqnarray*}
F &:=& \cond{T}~|~\cond{0}~|~X ~|~ X^{\pol} ~|~ F\otimes F ~|~ F\parr F~|~ F\oplus F~|~ F\with F~|~\forall X~F~|~\exists X~F~~~~~(X\in\mathcal{V})
\end{eqnarray*}
\end{definition}

\begin{definition}[The Sequent Calculus MALL$^{2}_{\cond{T,0}}$]
A proof in the sequent calculus MALL$^{2}$ is a derivation tree constructed from the derivation rules shown in \autoref{ellcomp} page \pageref{ellcomp}.
\end{definition}

\begin{figure}
\centering
\subfigure[Identity Group]{
\framebox{
\centering
\begin{tabular}{cc}
\begin{minipage}{5.2cm}
\begin{prooftree}
\AxiomC{}
\RightLabel{\scriptsize{ax}}
\UnaryInfC{$\vdash C^{\pol},C$}
\end{prooftree}
\end{minipage}
&
\begin{minipage}{5.15cm}
\begin{prooftree}
\AxiomC{$\Delta_{1} \vdash \Gamma_{1},C$}
\AxiomC{$\Delta_{2} \vdash \Gamma_{2},C^{\pol}$}
\RightLabel{\scriptsize{cut}}
\BinaryInfC{$\Delta_{1},\Delta_{2}\vdash \Gamma_{1},\Gamma_{2}$}
\end{prooftree}
\end{minipage}
\end{tabular}
}
}
\subfigure[Multiplicative Group]{
\framebox{
\begin{tabular}{c}
\begin{tabular}{cc}
\begin{minipage}{5cm}
\begin{prooftree}
\AxiomC{$\Delta_{1}\vdash \Gamma_{1},C_{1}$}
\AxiomC{$\Delta_{2}\vdash \Gamma_{2},C_{2}$}
\RightLabel{\scriptsize{$\otimes$}}
\BinaryInfC{$\Delta_{1},\Delta_{2}\vdash\Gamma_{1},\Gamma_{2},C_{1}\otimes C_{2}$}
\end{prooftree}
\end{minipage}
&
\begin{minipage}{4.95cm}
\begin{prooftree}
\AxiomC{$\Delta \vdash \Gamma,C_{1},C_{2}$}
\RightLabel{\scriptsize{$\parr$}}
\UnaryInfC{$\Delta\vdash \Gamma,C_{1}\parr C_{2}$}
\end{prooftree}
\end{minipage}
\end{tabular}
\end{tabular}
}
}
\subfigure[Additive Group]{
\framebox{
\begin{tabular}{cc}
\begin{minipage}{5cm}
\begin{prooftree}
\AxiomC{$\vdash\Gamma,C_{i}$}
\RightLabel{\scriptsize{$\oplus_{i}$}}
\UnaryInfC{$\vdash \Gamma,C_{1}\oplus C_{2}$}
\end{prooftree}
\end{minipage}
&
\begin{minipage}{5.4cm}
\begin{prooftree}
\AxiomC{$\vdash \Gamma, C_{1}$}
\AxiomC{$\vdash \Gamma, C_{2}$}
\RightLabel{\scriptsize{$\with$}}
\BinaryInfC{$\vdash \Gamma, C_{1}\with C_{2}$}
\end{prooftree}
\end{minipage}
\\~\\
\begin{minipage}{4cm}
\begin{prooftree}
\AxiomC{}
\RightLabel{\scriptsize{$\top$}}
\UnaryInfC{$\vdash \Gamma, \top$}
\end{prooftree}
\end{minipage}
&
\begin{minipage}{4cm}
\centering
No rules for $0$.
\end{minipage}
\end{tabular}
}
}
\subfigure[Quantifier Group]{
\framebox{
\begin{tabular}{c}
\begin{tabular}{cc}
\begin{minipage}{5cm}
\begin{prooftree}
\AxiomC{$\vdash \Gamma,C$}
\AxiomC{$X\not\in FV(\Gamma)$}
\RightLabel{\scriptsize{$\forall$}}
\BinaryInfC{$\vdash\Gamma,\forall X~ C$}
\end{prooftree}
\end{minipage}
&
\begin{minipage}{4.95cm}
\begin{prooftree}
\AxiomC{$\vdash \Gamma,C[A/X]$}
\RightLabel{\scriptsize{$\exists$}}
\UnaryInfC{$\vdash \Gamma,\exists X~C$}
\end{prooftree}
\end{minipage}
\end{tabular}
\end{tabular}
}
}
\caption{Rules for the sequent calculus MALL$^{2}_{\cond{T,0}}$}\label{ellcomp}
\end{figure}

To prove soundness, we will follow the proof technique used in our previous papers \cite{seiller-goim,seiller-goia}. We will first define a localized sequent calculus and show a result of soundness for it. The soundness result for the non-localized calculus is then obtained by noticing that one can always \emph{localize} a derivation. We will consider here that the variables are defined with the carrier equal to an interval in $\realN$ of the form $[i,i+1[$.

\begin{definition}
We fix a set $\mathcal{V}=\{X_{i}(j)\}_{i,j\in\naturalN\times\integerN}$ of \emph{localized variables}. For $i\in\naturalN$, the set $X_{i}=\{X_{i}(j)\}_{j\in\integerN}$ will be called the \emph{variable name $X_{i}$}, and an element of $X_{i}$ will be called a \emph{variable of name $X_{i}$}.
\end{definition}
For $i,j\in\naturalN\times\integerN$ we define the \emph{location} $\sharp X_{i}(j)$ of the variable $X_{i}(j)$ as the set $$\{x\in\realN~|~ 2^{i}(2j+1)\leqslant x< 2^{i}(2j+1)+1\}$$

\begin{definition}[Formulas of locMALL$^{2}_{\cond{T,0}}$]
We inductively define the formulas of \emph{localized second order multiplicative-additive linear logic} locMALL$^{2}_{\cond{T,0}}$ as well as their \emph{locations} as follows:
\begin{itemize}[noitemsep,nolistsep]
\item A variable $X_{i}(j)$ of name $X_{i}$ is a formula whose location is defined as $\sharp X_{i}(j)$;
\item If $X_{i}(j)$ is a variable of name $X_{i}$, then $(X_{i}(j))^{\pol}$ is a formula whose location is $\sharp X_{i}(j)$.
\item The constants $\cond{T}_{\sharp \Gamma}$ are formulas whose location is defined as $\sharp\Gamma$;
\item The constants $\cond{0}_{\sharp\Gamma}$ are formulas whose location is defined as $\sharp\Gamma$.
\item If $A,B$ are formulas with respective locations $X,Y$ such that $X\cap Y=\emptyset$, then $A\otimes B$ (resp. $A\parr B$, resp. $A\with B$, resp. $A\oplus B$) is a formula whose location is $X\cup Y$;
\item If $X_{i}$ is a variable name, and $A(X_{i})$ is a formula of location $\sharp A$, then $\forall X_{i}~A(X_{i})$ and $\exists X_{i}~A(X_{i})$ are formulas of location $\sharp A$.
\end{itemize}
\end{definition}

\begin{definition}[Interpretations]
An \emph{interpretation basis} is a function $\Phi$ which associates to each variable name $X_{i}$ a behavior with carrier\footnote{We consider $[0,1[\times\{\ast\}$ and not simply $[0,1[$ only to ensure that the image of $\Phi$ is disjoint from the locations of the variables.} $[0,1[\times\{\ast\}$.
\end{definition}

\begin{definition}[Interpretation of locMALL$^{2}_{\cond{T,0}}$ formulas]
Let $\Phi$ be an interpretation basis. We define the interpretation $I_{\Phi}(F)$ along $\Phi$ of a formula $F$ inductively:
\begin{itemize}[noitemsep,nolistsep]
\item If $F=X_{i}(j)$, then $I_{\Phi}(F)$ is the delocation (i.e.\ a behavior) of $\Phi(X_{i})$ defined by the function $x\mapsto 2^{i}(2j+1)+x$;
\item If $F=(X_{i}(j))^{\pol}$, we define the behavior $I_{\Phi}(F)=(I_{\Phi}(X_{i}(j)))^{\pol}$;
\item If $F=\cond{T}_{\sharp\Gamma}$ (resp. $F=\cond{0}_{\sharp\Gamma}$), we define $I_{\Phi}(F)$ as the behavior $\cond{T}_{\sharp\Gamma}$ (resp. $\cond{0}_{\sharp\Gamma}$);
\item If $F=\cond{1}$ (resp. $F=\cond{\bot}$), we define $I_{\Phi}(F)$ as the behavior $\cond{1}$ (resp. $\cond{\bot}$);
\item If $F=A\otimes B$, we define the conduct $I_{\Phi}(F)=I_{\Phi}(A)\otimes I_{\Phi}(B)$;
\item If $F=A\parr B$, we define the conduct $I_{\Phi}(F)=I_{\Phi}(A)\parr I_{\Phi}(B)$;
\item If $F=A\oplus B$, we define the conduct $I_{\Phi}(F)=I_{\Phi}(A)\oplus I_{\Phi}(B)$;
\item If $F=A\with B$, we define the conduct $I_{\Phi}(F)=I_{\Phi}(A)\with I_{\Phi}(B)$;
\item If $F=\forall X_{i} A(X_{i})$, we define the conduct $I_{\Phi}(F)=\cond{\forall X_{i}} I_{\Phi}(A(X_{i}))$;
\item If $F=\exists X_{i} A(X_{i})$, we define the conduct $I_{\Phi}(F)=\cond{\exists X_{i}} I_{\Phi}(A(X_{i}))$.
\end{itemize}
Moreover, a sequent $\vdash \Gamma$ will be interpreted as the $\parr$ of formulas in $\Gamma$, which will be written $\bigparr \Gamma$.
\end{definition}

\begin{definition}\label{gainwrtGoI5}
Let $F$ be a formula, $A$ a subformula of $F$, $n$ the number of occurrences of $A$ in $F$, and $X_{i}$ be a variable name that does not appear in $F$. We define an enumeration $e_{A/F}$ of the occurrences of $A$ in $F$ whose image is $\{1,\dots,n\}$. For each $j\in\{1,\dots,n\}$, we define $\psi_{j}: \sharp e^{-1}(j)\rightarrow \sharp X_{i}(j)$ as the natural (order-preserving) measure-inflating map between $\sharp e^{-1}(j)$, a disjoint union of unit segments, and $\sharp X_{i}(j)$, a unit segment. We then define the \emph{measure-inflating fax} $[e^{-1}(j)\leftrightarrow X_{i}(j)]$ as the graphing:
$$\{(1,\psi),(1,\psi^{-1})\}$$
\end{definition}

\begin{definition}[Interpretation of locMALL$^{2}_{\cond{T,0}}$ proofs]\label{interpretationpreuvesellcomp}
Let $\Phi$ be an interpretation basis. We define the interpretation $I_{\Phi}(\pi)$ -- a project -- of a proof $\pi$ inductively:
\begin{itemize}[noitemsep,nolistsep]
\item if $\pi$ is a single axiom rule introducing the sequent $\vdash (X_{i}(j))^{\pol},X_{i}(j')$, we define $I_{\Phi}(\pi)$ as the project $\de{Fax}$ defined by the translation $x \mapsto 2^{i}(2j'-2j)+x$;
\item if $\pi$ is composed of a single rule $\cond{T}_{\sharp \Gamma}$, we define $I_{\Phi}(\pi)=\de{0}_{\sharp\Gamma}$;
\item if $\pi$ is obtained from $\pi'$ by using a $\parr$ rule, then $I_{\Phi}(\pi)=I_{\Phi}(\pi')$;
\item if $\pi$ is obtained from $\pi_{1}$ and $\pi_{2}$ by performing a $\otimes$ rule, we define $I_{\Phi}(\pi)=I_{\Phi}(\pi_{1})\otimes I_{\Phi}(\pi')$;
\item if $\pi$ is obtained from $\pi'$ using a $\oplus_{i}$ rule introducing a formula of location $V$, we define $I_{\Phi}(\pi)=I_{\Phi}(\pi')\otimes\de{0}_{V}$;
\item if $\pi$ of conclusion $\vdash \Gamma, A_{0}\with A_{1}$ is obtained from $\pi_{0}$ and $\pi_{1}$ using a $\with$ rule, we define the interpretation of $\pi$ in the same way it was defined in the author's paper on additives. We first define the maps
\begin{eqnarray*}
\psi_{i}&:& x\mapsto (x,i)~~~~(i=0,1)\\
\tilde{\psi_{i}}&=&((\psi_{i})\restr{\sharp\Gamma})^{-1}~~~~(i=0,1)\\
\dot{\psi_{i}}&=&((\psi_{i})\restr{\sharp A_{i}})^{-1}~~~~(i=0,1)
\end{eqnarray*}
The interpretation of $\pi$ is then defined as:
\begin{equation*}
I_{\Phi}(\pi)=\de{Distr}^{\tilde{\psi_{0}},\tilde{\psi_{1}}}_{\dot{\psi_{0}},\dot{\psi_{1}}}\plug(\psi_{0}(I_{\Phi}(\pi_{0}))\otimes\de{0}_{\sharp A_{1}}+\psi_{1}(I_{\Phi}(\pi_{1}))\otimes\de{0}_{\sharp A_{0}})
\end{equation*}

where $\de{Distr}^{\tilde{\psi_{0}},\tilde{\psi_{1}}}_{\dot{\psi_{0}},\dot{\psi_{1}}}$ is a project implementing distributivity \cite[Proposition 73]{seiller-goia}.
;
\item If $\pi$ is obtained from a $\forall$ rule applied to a derivation $\pi'$, we define $I_{\Phi}(\pi)=I_{\Phi}(\pi')$;
\item If $\pi$ is obtained from a $\exists$ rule applied to a derivation $\pi'$ replacing the formula $\cond{A}$ by the variable name $X_{i}$, we define $I_{\Phi}(\pi)=I_{\Phi}(\pi')\plugmes (\bigotimes [e^{-1}(j)\leftrightarrow X_{i}(j)])$;
\item if $\pi$ is obtained from $\pi_{1}$ and $\pi_{2}$ by applying a $\text{cut}$ rule, we define $I_{\Phi}(\pi)=I_{\Phi}(\pi_{1})\exec I_{\Phi}(\pi_{2})$.
\end{itemize}
\end{definition}

\begin{theorem}[locMALL$^{2}_{\cond{T,0}}$ soundness]
Let $\Phi$ be an interpretation basis. Let $\pi$ be a derivation in locMALL$^{2}_{\cond{T,0}}$ of conclusion $\vdash\Gamma$. Then $I_{\Phi}(\pi)$ is a successful project in $I_{\Phi}(\vdash\Gamma)$.
\end{theorem}

%\note{change the following}

Now, one can chose an enumeration of the occurrences of variables in order to \enquote{localize} any formula $A$ and any proof $\pi$ of MALL$^{2}_{\cond{T,0}}$: we then obtain formulas $A^{e}$ and proofs $\pi^{e}$ of locMALL$^{2}_{\cond{T,0}}$. The following theorem is therefore a direct consequence of the preceding one.

\begin{theorem}[Full MALL$^{2}_{\cond{T,0}}$ Soundness]
Let $\Phi$ be an interpretation basis, $\pi$ an MALL$^{2}_{\cond{T,0}}$ proof of conclusion $\Delta\vdash \Gamma;$ and $e$ an enumeration of the occurrences of variables in the axioms in $\pi$. Then $I_{\Phi}(\pi^{e})$ is a successful project in $I_{\Phi}(\Delta^{e}\vdash \Gamma^{e};)$.
\end{theorem}

%\subsection{Invariance by Cut-Elimination}
%
%\note{Write this}

\section{Perspectives}

We described in this paper a general construction of models of multiplicative-additive linear logic (MALL). This general construction can be performed on any \emph{trefoil space}, that is a measure space subject to a few conditions. Given a trefoil space $X$, we obtain a hierarchy of models of MALL corresponding to the hierarchy of \emph{microcosms}, i.e.\ monoids of non-singular transformations from $X$ to itself, and the hierarchy of weight monoids. In particular, all previously considered geometry of interaction constructions can be recovered for particular trefoils spaces $X$, weight monoids and microcosms. 

The perspectives of this work are numerous. First, one can extend the model on the real line described at the end of this paper in order to deal with exponential connectives, following the approach described in the author's PhD thesis \cite{seiller-phd}. Following this approach, the author obtained models of Elementary Linear Logic \cite{seiller-goie} and full linear logic (i.e.\ without restricted exponentials) \cite{seiller-goif}.

The most exciting perspective of this work concerns the field of computational complexity. As described in a short note \cite{seiller-lcc14} and a perspective paper \cite{seiller-towards}, we can show a correspondence between a part of the hierarchy of models obtained here and a part of the hierarchy of complexity classes. Indeed, as we consider bigger microcosms, the type of predicates $\oc\cond{Nat_{2}}\multimap\cond{Bool}$ becomes larger. Intuitively, a microcosm describes the computational principles allowed in the system. By adapting earlier results obtained with von Neumann algebras \cite{seiller-conl,seiller-lsp} we can define microcosms for which the type of predicates characterizes the class of regular languages on one hand, and the class of logarithmic space predicates on the other. We believe the theory of dynamical systems and ergodic theory might shed new light on the field of computational complexity through this approach. In particular, we will study how mathematical invariants, such as $\ell^{2}$-Betti numbers, are related to computation.

We can also apply the techniques developed here for quantum computation. Indeed, it is possible to model quantum circuits in a very nice way in some of the models defined in this paper. Once again, one could gain from the possibility of considering smaller and/or larger microcosms. For instance, one could study restrictions of these models of quantum computation where the available unitary gates are limited to a chosen basis. It would then be possible to understand how the different choices of bases of unitaries affect the model from a computational and/or logical point of view.

\bibliographystyle{alpha}
\bibliography{thomas}

\end{document}